\documentclass[reqno]{amsart}
\usepackage[margin=1.5in,bottom=1.25in]{geometry}	

\usepackage{amsmath}	
\usepackage{amssymb}	
\usepackage{amsfonts}	
\usepackage{amsthm}	
\usepackage[foot]{amsaddr}	

\usepackage[centercolon=true]{mathtools}	
\mathtoolsset{%
}

\usepackage[utf8]{inputenc}	
\usepackage[T1]{fontenc}	

\usepackage[
cal=cm,
]
{mathalfa}

\usepackage{dsfont}	

\usepackage[light,scaled=.95]{AlegreyaSans}

\usepackage[ttdefault,scale=.95]{AnonymousPro}

\usepackage[proportional,tabular,lining,sf,mono=false]{libertine}

\usepackage{acronym}	
\newcommand{\acdef}[1]{\define{\acl{#1}} \textup{(\acs{#1})}\acused{#1}}	
\newcommand{\acdefp}[1]{\define{\aclp{#1}} \textup{(\acsp{#1})}\acused{#1}}	

\usepackage[labelfont={bf,small},labelsep=colon,font=small]{caption}	
\usepackage{subcaption}	

\usepackage[svgnames]{xcolor}	
\colorlet{MyRed}{Crimson!60!DarkRed}
\colorlet{MyBlue}{DodgerBlue!75!black}
\colorlet{MyGreen}{DarkGreen}
\colorlet{MyViolet}{DarkMagenta!80!MyBlue}

\colorlet{MyLightBlue}{DodgerBlue!20}
\colorlet{MyLightGreen}{MyGreen!20}

\colorlet{PrimalColor}{MyBlue}
\colorlet{PrimalFill}{MyLightBlue}
\colorlet{DualColor}{MyRed}

\colorlet{AlertColor}{MyRed}	
\colorlet{BadColor}{MyRed}	
\colorlet{GoodColor}{MyGreen}	
\colorlet{LinkColor}{MediumBlue}	
\colorlet{MacroColor}{MyViolet}
\colorlet{RevColor}{MediumBlue}	

\newcommand{\afterhead}{.\;}	
\newcommand{\para}[1]{\smallskip\paragraph{\textbf{#1\afterhead}}}

\usepackage{latexsym}	
\usepackage{fontawesome}	
\usepackage{pifont}	

\usepackage{tikz}	
\usepackage{tikz-cd}	

\usepackage{array}	
\usepackage{booktabs}	
\usepackage[inline,shortlabels]{enumitem}	
\setlist[1]{topsep=\smallskipamount,itemsep=\smallskipamount,left=\parindent}
\setlist[2]{left=0pt}

\usepackage[kerning=true]{microtype}	

\usepackage{tabto}	
\usepackage{xspace}	

\usepackage[authoryear,sort&compress]{natbib}	

\bibpunct[, ]{[}{]}{,}{}{,}{,}
\setcitestyle{numbers,square}

\usepackage{hyperref}
\hypersetup{
final,
colorlinks=true,
linktocpage=true,
pdfstartview=FitH,
breaklinks=true,
pdfpagemode=UseNone,
pageanchor=true,
pdfpagemode=UseOutlines,
plainpages=false,
bookmarksnumbered,
bookmarksopen=false,
bookmarksopenlevel=1,
hypertexnames=true,
pdfhighlight=/O,
hyperfootnotes=false,
urlcolor=LinkColor,linkcolor=LinkColor,citecolor=LinkColor,	
pdftitle={},
pdfauthor={},
pdfsubject={},
pdfkeywords={},
pdfcreator={pdfLaTeX},
pdfproducer={LaTeX with hyperref}
}

\newcommand{\EMAIL}[1]{\email{\href{mailto:#1}{#1}}}

\usepackage[sort&compress,capitalize,nameinlink]{cleveref}	

\crefname{equation}{Eq.}{Eqs.}
\crefname{algo}{Algorithm}{Algorithms}
\crefname{assumption}{Assumption}{Assumptions}
\crefname{case}{Case}{Cases}
\crefname{figure}{Fig.}{Figs.}
\crefname{proofstep}{Step}{Steps}

\usepackage{algorithm}	
\usepackage{algpseudocode}	

\usepackage{thmtools}	
\usepackage{thm-restate}	

\theoremstyle{plain}
\newtheorem{theorem}{Theorem}	

\newtheorem{lemma}{Lemma}	
\newtheorem{proposition}{Proposition}	
\newtheorem{propositionApp}{Proposition}[section]	
\newtheorem{lemmaApp}[propositionApp]{Lemma}	

\newtheorem{corollaryApp}[propositionApp]{Corollary}	

\newtheorem*{theorem*}{Theorem}	
\newtheorem*{corollary*}{Corollary}	

\theoremstyle{definition}
\newtheorem{definition}{Definition}	
\newtheorem{definitionApp}[propositionApp]{Definition}	
\newtheorem{example}{Example}	

\newtheorem*{definition*}{Definition}	
\newtheorem*{assumption*}{Assumptions}	
\newtheorem*{example*}{Example}	

\theoremstyle{remark}

\newtheorem*{remark*}{Remark}	
\newtheorem*{notation*}{Notation}	

\def\endenv{\hfill{\small$\blacklozenge$}}	

\newcounter{proofstep}
\newenvironment{proofstep}[1]
{\vspace{3pt}
\refstepcounter{proofstep}%
\par\textit{Step~\arabic{proofstep}:~#1}.\,}
{\smallskip}

\numberwithin{example}{section}	

\usepackage[showdeletions]{color-edits}	
\newcommand{\draft}[1]{#1}	

\newcommand{\define}[1]{\emph{\draft{#1}}}	

\newcommand{\explain}[1]{\tag*{\small\texttt{\%}\;#1}}

\newcommand{\newmacro}[2]{\newcommand{#1}{\draft{#2}}}	
\newcommand{\newop}[2]{\DeclareMathOperator{#1}{\draft{#2}}}	
\newcommand{\newopstar}[2]{\DeclareMathOperator*{#1}{\draft{#2}}}	

\newcommand{\eps}{\varepsilon}	
\newcommand{\pd}{\partial}	

\DeclarePairedDelimiter{\braces}{\{}{\}}	
\DeclarePairedDelimiter{\bracks}{[}{]}	
\DeclarePairedDelimiter{\parens}{(}{)}	
\DeclarePairedDelimiter{\of}{(}{)}	

\DeclarePairedDelimiter{\abs}{\lvert}{\rvert}	

\DeclarePairedDelimiter{\setof}{\{}{\}}	
\DeclarePairedDelimiterX{\setdef}[2]{\{}{\}}{#1:#2}	
\DeclarePairedDelimiterXPP{\exclude}[1]{\mathopen{}\setminus}{\{}{\}}{}{#1}

\DeclarePairedDelimiterX{\braket}[2]{\langle}{\rangle}{#1,#2}	
\DeclarePairedDelimiterX{\inner}[2]{\langle}{\rangle}{#1,#2}	
\DeclarePairedDelimiterX{\dualp}[2]{\langle}{\rangle}{#1,#2}	
\DeclarePairedDelimiter{\norm}{\lVert}{\rVert}	
\DeclarePairedDelimiterXPP{\dnorm}[1]{}{\lVert}{\rVert}{_{\ast}}{#1}	
\DeclarePairedDelimiterXPP{\onenorm}[1]{}{\lVert}{\rVert}{_{1}}{#1}	
\DeclarePairedDelimiterXPP{\twonorm}[1]{}{\lVert}{\rVert}{_{2}}{#1}	
\DeclarePairedDelimiterXPP{\supnorm}[1]{}{\lVert}{\rVert}{_{\infty}}{#1}	

\newcommand{\defeq}{\coloneqq}	
\newcommand{\eqdef}{\eqqcolon}	

\newcommand{\from}{\colon}	

\newmacro{\F}{\mathbb{F}}	
\newmacro{\N}{\mathbb{N}}	
\newmacro{\Z}{\mathbb{Z}}	
\newmacro{\Q}{\mathbb{Q}}	

\newmacro{\real}{x}	
\newmacro{\reals}{\mathbb{R}}	
\newmacro{\R}{\reals}	

\newmacro{\complex}{z}	
\newmacro{\complexes}{\mathbb{C}}	
\newmacro{\C}{\complexes}	

\newopstar{\argmax}{arg\,max}	
\newopstar{\argmin}{arg\,min}	
\newopstar{\intersect}{\bigcap}	
\newopstar{\union}{\bigcup}	

\newop{\aff}{aff}	
\newop{\bd}{bd}	
\newop{\bigoh}{\mathcal{O}}	
\newop{\card}{card}	
\newop{\cl}{cl}	
\newop{\conv}{conv}	
\newop{\crit}{crit}	
\newop{\curl}{curl}	
\newop{\diag}{diag}	
\newop{\diam}{diam}	
\newop{\dist}{dist}	
\newop{\diver}{div}	
\newop{\dom}{dom}	
\newop{\eig}{eig}	
\newop{\ess}{ess}	
\newop{\grad}{grad}	
\newop{\Hess}{Hess}	
\newop{\ind}{ind}	
\newop{\im}{im}	
\newop{\intr}{int}	
\newop{\Jac}{Jac}	
\newop{\one}{\mathds{1}}	
\newop{\proj}{pr}	
\newop{\prox}{prox}	
\newop{\rank}{rank}	
\newop{\relint}{ri}	
\newop{\sign}{sgn}	
\newop{\supp}{supp}	
\newop{\Sym}{Sym}	
\newop{\tr}{tr}	
\newop{\unif}{unif}	
\newop{\vol}{vol}	

\newcommand{\cf}{cf.\xspace}	
\newcommand{\eg}{e.g.,\xspace}	
\newcommand{\ie}{i.e.,\xspace}	
\newcommand{\vs}{vs.\xspace}	
\newcommand{\viz}{viz.\xspace}	

\newcommand{\textpar}[1]{\textup(#1\textup)}	

\newcommand{\eqdot}{\,.}	

\newcommand{\alt}[1]{#1'}		
\newcommand{\altalt}[1]{#1''}		

\newmacro{\ball}{\mathbb{B}}	
\newmacro{\sphere}{\mathbb{S}}	

\newmacro{\argdot}{\cdot}	
\newmacro{\dd}{\:d}	
\newmacro{\ddt}{\frac{d}{dt}}	
\newmacro{\del}{\partial}	

\newcommand{\insum}{\sum\nolimits}	

\newmacro{\const}{c}	
\newmacro{\Const}{C}	

\newmacro{\param}{\theta}	
\newmacro{\params}{\Theta}	

\newmacro{\coef}{\lambda}	

\newmacro{\fn}{f} 

\newmacro{\pexp}{p}	
\newmacro{\qexp}{q}	
\newmacro{\rexp}{r}	

\newmacro{\idx}{i}
\newmacro{\idxalt}{j}
\newmacro{\idxaltalt}{k}
\newmacro{\nIndices}{I}
\newmacro{\indices}{\mathcal{I}}

\newmacro{\point}{x}	
\newmacro{\pointalt}{\alt\point}	
\newmacro{\pointaltalt}{\altalt\point}	
\newmacro{\points}{\mathcal{X}}	
\newmacro{\intpoints}{\relint\points}	

\newmacro{\base}{p}	
\newmacro{\basealt}{q}	
\newmacro{\basealtalt}{u}	

\newmacro{\set}{\mathcal{S}}	

\newmacro{\borel}{\mathcal{B}}	
\newmacro{\closed}{\mathcal{C}}	
\newmacro{\cpt}{\mathcal{K}}	
\newmacro{\nhd}{\mathcal{U}}	
\newmacro{\open}{\mathcal{U}}	

\newmacro{\domain}{\mathcal{D}}	
\newmacro{\region}{\mathcal{R}}	

\newmacro{\interval}{\mathcal{I}}	
\newmacro{\rectangle}{\mathcal{R}}	

\newmacro{\vecspace}{\mathcal{V}}	
\newmacro{\subspace}{\mathcal{Z}}	
\newmacro{\dualspace}{\pspace^{\ast}}	
\newmacro{\vdim}{d}	

\newmacro{\unitvec}{u}	
\newmacro{\bvec}{e}	
\newmacro{\bvecs}{\mathcal{E}}	

\newmacro{\pvec}{z}	
\newmacro{\pvecalt}{\alt\pvec}	
\newmacro{\pvecaltalt}{\altalt\pvec}	
\newmacro{\pvecs}{\vecspace}	
\newmacro{\pspace}{\pvecs}	

\newmacro{\dvec}{w}	
\newmacro{\dvecalt}{\alt\pvec}	
\newmacro{\dvecaltalt}{\altalt\pvec}	
\newmacro{\dvecs}{\dualspace}	
\newmacro{\dspace}{\dvecs}	

\newmacro{\coord}{i}	
\newmacro{\coordalt}{j}	
\newmacro{\coordaltalt}{k}	
\newmacro{\nCoords}{d}	

\newmacro{\vecfield}{v}	
\newmacro{\vbound}{V}	

\newmacro{\cvx}{\mathcal{C}}	

\newmacro{\subd}{\partial}	
\newmacro{\subsel}{\nabla}	
\newmacro{\dir}{\kern0pt\subd\mkern-1mu{}}	

\newop{\tcone}{TC}	
\newop{\dcone}{\tcone^{\ast}}	
\newop{\ncone}{NC}	
\newop{\pcone}{PC}	
\newop{\hull}{\Delta}	

\newop{\Opt}{\mathsf{Opt}}	
\newop{\Sol}{\mathsf{Sol}}	
\newop{\gap}{\mathsf{Gap}}	
\newop{\orcl}{\mathsf{G}}	

\newmacro{\obj}{f}	
\newmacro{\objalt}{g}	
\newmacro{\sobj}{F}	

\newcommand{\sol}[1][\point]{#1^{\ast}}	

\newmacro{\gvec}{g}	
\newmacro{\oper}{A}	

\newmacro{\lips}{G}	
\newmacro{\strong}{\alpha}	
\newmacro{\smooth}{\beta}	

\newmacro{\radius}{r}
\newmacro{\Radius}{R}

\newmacro{\mfld}{\mathcal{M}}	

\newmacro{\gmat}{g}	
\newmacro{\gdist}{\dist_{\gmat}}	

\newmacro{\tanvec}{z}	
\newmacro{\form}{\omega}	

\newop{\ex}{\mathbb{E}}	
\newop{\prob}{\mathbb{P}}	
\newop{\Var}{\mathbb{V}}	
\newop{\simplex}{\Delta}	

\DeclarePairedDelimiterXPP{\exof}[1]{\ex}{[}{]}{}{
 #1}

\DeclarePairedDelimiterXPP{\exwrt}[2]{\ex_{#1}}{[}{]}{}{
 #2}

\DeclarePairedDelimiterXPP{\probof}[1]{\prob}{(}{)}{}{
 #1}

\DeclarePairedDelimiterXPP{\probwrt}[2]{\prob_{#1}}{[}{]}{}{
 #2}

\DeclarePairedDelimiterXPP{\oneof}[1]{\one}{\{}{\}}{}{#1}	

\newmacro{\rv}{X}	
\newmacro{\rvalt}{Y}	

\newmacro{\event}{E}	
\newmacro{\eventalt}{H}	

\newmacro{\seed}{\theta}	
\newmacro{\seeds}{\Theta}	
\newmacro{\pdist}{P}	
\newmacro{\history}{\mathcal{H}}	

\newmacro{\sample}{\omega}	
\newmacro{\samples}{\Omega}	

\newmacro{\filter}{\mathcal{F}}	
\newmacro{\probspace}{(\samples,\filter,\prob)}	

\newmacro{\mean}{\mu}	
\newmacro{\sdev}{\sigma}	
\newmacro{\variance}{\sdev^{2}}	
\newmacro{\covmat}{\Sigma}	

\newmacro{\seq}{a}	
\newmacro{\seqalt}{b}	
\newmacro{\seqaltalt}{c}	

\newmacro{\beforestart}{0}	
\newmacro{\start}{1}	
\newmacro{\afterstart}{2}	
\newmacro{\running}{\start,\afterstart,\dotsc}	
\newmacro{\halfrunning}{1/2,1,3/2,\dotsc}	

\newmacro{\run}{n}	
\newmacro{\runalt}{k}	
\newmacro{\runaltalt}{m}	
\newmacro{\offset}{r}
\newmacro{\nRuns}{T}	
\newmacro{\runs}{\mathcal{\nRuns}}	

\newmacro{\curve}{\gamma}	
\DeclarePairedDelimiterXPP{\curveof}[1]{\curve}{(}{)}{}{#1}	
\DeclarePairedDelimiterXPP{\dotcurveof}[1]{\dot\curve}{(}{)}{}{#1}	

\newmacro{\tstart}{0}	
\renewcommand{\time}{\draft{t}}	
\newmacro{\timealt}{\tau}	
\newmacro{\timealtalt}{s}	
\newmacro{\horizon}{T}	
\newmacro{\tend}{\horizon}	
\newmacro{\window}{[\tstart,\tend]}	

\newcommand{\new}[1][\point]{#1^{+}}	

\newmacro{\state}{\strat}	
\newmacro{\statealt}{\score}	
\newmacro{\statealtalt}{\scorediff}	

\newcommand{\init}[1][\state]{\draft{#1}_{\start}}	

\newcommand{\iter}[1][\state]{\draft{#1}_{\runalt}}	

\newcommand{\curr}[1][\state]{\draft{#1}_{\run}}	
\renewcommand{\next}[1][\state]{\draft{#1}_{\run+1}}	

\newcommand{\beforelead}[1][\state]{\draft{#1}_{\run-1/2}}		
\newcommand{\lead}[1][\state]{\draft{#1}_{\run+1/2}}		

\newcommand{\afterlast}[1][\state]{\draft{#1}_{\nRuns+1}}	

\newcommand{\curri}[1][\state]{\draft{#1}_{\play,\run}}	
\newcommand{\nexti}[1][\state]{\draft{#1}_{\play,\run+1}}	

\newcommand{\leadi}[1][\state]{\draft{#1}_{\play,\run+1/2}}		

\newmacro{\mat}{M}	
\newmacro{\hmat}{H}	

\newmacro{\ones}{\mathbf{1}}	
\newmacro{\eye}{I}	
\newmacro{\zer}{\mathbf{0}}	

\newmacro{\eigval}{\lambda}	
\newmacro{\eigvec}{u}	

\newop{\Nash}{NE}	
\newop{\CE}{CE}	
\newop{\CCE}{CCE}	
\newop{\NI}{NI}	

\newop{\brep}{br}	
\newop{\reg}{Reg}	
\newop{\preg}{\overline{Reg}}	
\newop{\val}{val}	

\newmacro{\play}{i}	
\newmacro{\others}{-\play}	
\newmacro{\playalt}{j}	
\newmacro{\playaltalt}{k}	
\newmacro{\nPlayers}{N}	
\newmacro{\players}{\mathcal{\nPlayers}}	

\newmacro{\pure}{\alpha}	
\newmacro{\purealt}{\beta}	
\newmacro{\purealtalt}{\gamma}	
\newmacro{\purebench}{\hat{\pure}}	
\newmacro{\nPures}{A}	
\newmacro{\pures}{\mathcal{\nPures}}	
\newmacro{\effpures}{\hat{\pures}} 

\newmacro{\findex}{\play \pure_{\play}} 
\newmacro{\benchfindex}{\play \purebench_{\play}} 

\newmacro{\strat}{x}	
\newmacro{\stratalt}{\alt\strat}	
\newmacro{\strataltalt}{\altalt\strat}	
\newmacro{\strats}{\mathcal{X}}	
\newcommand{\stratsi}[1][\play]{\strats_{\play}}
\newmacro{\intstrats}{\relint{\strats}}	

\newmacro{\corr}{z}		
\newmacro{\corralt}{\alt\corr}		
\newmacro{\corrs}{\mathcal{Z}}		

\newcommand{\eq}{\sol[\strat]}	
\newmacro{\bench}{\base}

\newmacro{\pay}{u}	
\newmacro{\payns}{k}	
\newmacro{\loss}{\ell}	

\newmacro{\payv}{v}	
\newmacro{\payfield}{v}	

\newmacro{\game}{\mathcal{G}}	
\newmacro{\gamefull}{\game(\players,\points,\pay)}	

\newmacro{\fingame}{\Gamma}	
\newmacro{\fingamefull}{\Gamma(\players,\pures,\pay)}	
\newmacro{\mixgame}{\Delta(\fingame)}	

\newmacro{\minmax}{L}	

\newmacro{\minvar}{\point_{1}}	
\newmacro{\minvaralt}{\alt\minvar}	
\newmacro{\minvars}{\points_{1}}	

\newmacro{\maxvar}{\point_{2}}	
\newmacro{\maxvaralt}{\alt\maxvar}	
\newmacro{\maxvars}{\points_{2}}	

\newmacro{\pot}{\phi}	
\newmacro{\potgame}{\fingame_{\!\textup{pot}}}	
\newmacro{\potfield}{P}	
\newmacro{\potrepfield}{\repfield_{\textup{pot}}}	

\newmacro{\harmgame}{\fingame_{\!\textup{harm}}}	
\newmacro{\harmfield}{H}	

\newmacro{\gauge}{\const}	
\newmacro{\nsgame}{\mathrm{NS}}	
\newmacro{\nsfield}{\Const}	

\newmacro{\incgame}{\fingame_{\!\textup{inc}}}	
\newmacro{\incfield}{B}	

\newmacro{\pureA}{\mathtt{A}}	
\newmacro{\pureB}{\mathtt{B}}	

\newmacro{\hreg}{h}	
\newmacro{\hconj}{\hreg^{\ast}}	
\newmacro{\breg}{D}	
\newmacro{\mprox}{P}	
\newmacro{\mirror}{Q}	
\newmacro{\effmirror}{\hat{\mirror}}	
\newmacro{\fench}{F}	
\newmacro{\hstr}{K}	
\newmacro{\hrange}{H}	
\newmacro{\proxdom}{\points_{\hreg}}	

\DeclarePairedDelimiterXPP{\bregof}[2]{\breg}{(}{)}{}{#1,#2}	
\DeclarePairedDelimiterXPP{\proxof}[2]{\mprox_{#1}}{(}{)}{}{#2}	

\newmacro{\zone}{\mathbb{D}}	

\newop{\Eucl}{\Pi}	
\newop{\logit}{\Lambda}	
\newop{\dkl}{KL}	

\newmacro{\dpoint}{y}	
\newmacro{\dpointalt}{\alt\dpoint}	
\newmacro{\dpointaltalt}{\altalt\dpoint}	
\newmacro{\dpoints}{\mathcal{Y}}	
\newmacro{\dstate}{Y}	

\newmacro{\flowmap}{\Theta}	
\DeclarePairedDelimiterXPP{\flowof}[2]{\flowmap_{#1}}{(}{)}{}{#2}	

\newmacro{\traj}{x}	
\newmacro{\difftraj}{\dot\traj}	

\DeclarePairedDelimiterXPP{\trajof}[1]{\traj}{(}{)}{}{#1}	
\DeclarePairedDelimiterXPP{\difftrajof}[1]{\difftraj}{(}{)}{}{#1}	

\newmacro{\trajalt}{y}	
\newmacro{\trajaltalt}{z}	

\newmacro{\signal}{\hat\vecfield}	
\newmacro{\step}{\gamma}	
\newmacro{\learn}{\eta}	
\newmacro{\regweight}{\lambda}	

\newmacro{\error}{Z}	
\newmacro{\noise}{U}	
\newmacro{\bias}{b}	

\newmacro{\sbound}{M}	
\newmacro{\nbound}{\sdev}	
\newmacro{\bbound}{B}	

\newmacro{\snoise}{\xi}	
\newmacro{\sbias}{\chi}	

\newmacro{\mix}{\delta}	
\newmacro{\perturb}{z}	
\newmacro{\pivot}{\point}	

\newmacro{\vertex}{v}	
\newmacro{\vertexalt}{w}	
\newmacro{\vertexaltalt}{u}	
\newmacro{\nVertices}{V}	
\newmacro{\vertices}{\mathcal{V}}	

\newmacro{\edge}{e}	
\newmacro{\edgealt}{\alt\edge}	
\newmacro{\edgealtalt}{\altalt\edge}	
\newmacro{\nEdges}{E}	
\newmacro{\edges}{\mathcal{\nEdges}}	

\newmacro{\graph}{\mathcal{G}}	
\newmacro{\graphfull}{\graph(\vertices,\edges)}	

\newmacro{\source}{O}	
\newmacro{\sink}{D}	

\newmacro{\pair}{i}	
\newmacro{\pairalt}{j}	
\newmacro{\pairaltalt}{k}	
\newmacro{\nPairs}{N}	
\newmacro{\pairs}{\mathcal{\nPairs}}	

\newmacro{\route}{p}	
\newmacro{\routealt}{\alt\route}	
\newmacro{\routealtalt}{\altalt\route}	
\newmacro{\nRoutes}{P}	
\newmacro{\routes}{\mathcal{\nRoutes}}	

\newmacro{\flow}{f}	
\newmacro{\flowalt}{\alt\flow}	
\newmacro{\flowaltalt}{\altalt\flow}	
\newmacro{\flows}{\mathcal{F}}	

\newmacro{\load}{x}	
\newmacro{\loadalt}{\alt\load}	
\newmacro{\loadaltalt}{\altalt\load}	
\newmacro{\loads}{\mathcal{X}}	

\newmacro{\meas}{\mu}	
\newmacro{\meascod}{\reals_{++}}	
\newmacro{\mass}{m}	

\DeclarePairedDelimiterXPP{\stratof}[1]{\strat}{(}{)}{}{#1}	
\DeclarePairedDelimiterXPP{\stratofi}[1]{\strat_{\play}}{(}{)}{}{#1}	
\DeclarePairedDelimiterXPP{\stratofx}[2]{\strat_{#1}}{(}{)}{}{#2}	

\DeclarePairedDelimiterXPP{\corrof}[1]{\corr}{(}{)}{}{#1}	
\DeclarePairedDelimiterXPP{\corrofx}[2]{\corr_{#1}}{(}{)}{}{#2}	

\newmacro{\ambient}{\prod_{\playalt} \R^{\pures_{\playalt}}}
\newcommand{\ambienti}[1][\play]{\R^{\pures_{#1}}}

\newmacro{\score}{y}	
\newmacro{\scorealt}{\alt\score}	
\newmacro{\scorealtalt}{\altalt\score}	
\newmacro{\scores}{\mathcal{Y}}	
\newmacro{\primalambientplay}{\R^{\pures_{\play}}}
\newmacro{\dualambientplay}{\left(\primalambientplay\right)^{\ast}}	
\DeclarePairedDelimiterXPP{\scoreof}[1]{\score}{(}{)}{}{#1}	
\DeclarePairedDelimiterXPP{\scoreofi}[1]{\score_{\play}}{(}{)}{}{#1}	
\DeclarePairedDelimiterXPP{\scoreofx}[2]{\score_{#1}}{(}{)}{}{#2}	
\DeclarePairedDelimiterXPP{\dotscoreof}[1]{\dot\score}{(}{)}{}{#1}	
\DeclarePairedDelimiterXPP{\dotscoreofi}[1]{\dot\score_{\play}}{(}{)}{}{#1}	
\DeclarePairedDelimiterXPP{\dotscoreofx}[2]{\dot\score_{#1}}{(}{)}{}{#2}	

\newmacro{\diffquot}{\Pi} 
\newmacro{\scorediff}{z} 
\newmacro{\scorediffalt}{\alt\scorediff} 
\newmacro{\scorediffsdynamics}{Z} 
\newmacro{\scorediffs}{\mathcal{\scorediffsdynamics}} 
\newmacro{\diffdyn}{\scorediffsdynamics} 

\newmacro{\energy}{E}	
\newmacro{\lyap}{L}	
\newmacro{\centroid}{q}	

\addauthor{PM}{Crimson}

\newmacro{\present}{a}
\newmacro{\past}{b}
\newmacro{\reflect}{c}

\newmacro{\limstrat}{\hat\strat}
\newmacro{\limscore}{\hat\score}

\addauthor{DL}{DarkGreen}

\newop{\sgn}{sgn} 
\newmacro{\vfield}{X}
\newcommand{\flowt}[2]{\flowmap_{#1}(#2)} 
\newmacro{\postime}{[0, \infty)}
\DeclarePairedDelimiterXPP{\pointof}[1]{\point}{(}{)}{}{#1}	
\newmacro{\pstart}{\point_{\tstart}} 
\newmacro{\pstarts}{\mathcal{U}} 
\newmacro{\mfldalt}{\alt\mfld}

\newmacro{\payalt}{\alt\pay} 
\newmacro{\payaltalt}{\alt{\alt{\pay}}} 

\newmacro{\findexalt}{\play \purealt_{\play}} 
\newmacro{\findexothers}{\others \pure_{\others}} 
\newmacro{\findexco}{\play \pure_{\others}} 
\newmacro{\bindex}{\benchfindex} 

\newmacro{\ppure}{ \pure_{\play}; \pure_{\others} }

\newcommand{\dev}[3]{ \pay_{#1}(#2_{#1}; #2_{-#1}) - \pay_{#1}(#3_{#1}; #2_{-#1}) }

\newmacro{\comeas}{\gamma} 
\newmacro{\wei}{\mass} 
\newcommand{\lnorm}[1]{\abs{#1}} 
\newmacro{\normeas}{\meas^{\ast}}	
\newmacro{\unimeas}{\bar{\meas}} 
\newmacro{\heq}{\centroid} 
\newmacro{\stratcenter}{(\wei, \heq)} 

\DeclarePairedDelimiterXPP{\scorediffof}[1]{\scorediff}{(}{)}{}{#1}	
\DeclarePairedDelimiterXPP{\scorediffofi}[1]{\scorediff{\play}}{(}{)}{}{#1}	

\newmacro{\vambient}{\mathcal{V}} 
\newmacro{\vambientdual}{\vambient^{\ast}} 
\newmacro{\primaldyn}{X}
\newmacro{\dualdyn}{Y}

\newmacro{\PI}{\diffquot} 
\newmacro{\Pii}{\diffquot_{\play}} 
\newmacro{\strati}{\strat_{\play}} 
\newmacro{\scori}{\score_{\play}} 
\newmacro{\scorin}{\score_{\play,\run}}
\newmacro{\mirrori}{\mirror_{\play}} 
\newmacro{\diffi}{\scorediff_{\play}} 
\newmacro{\diffs}{\scorediffs} 
\newmacro{\diffsi}{\scorediffs_{\play}} 
\newmacro{\piscore}{\scorediff} 
\newop{\jac}{Jac}
\newmacro{\puri}{\pure_{\play}} 
\newmacro{\puresi}{\pures_{\play}} 
\newmacro{\purialt}{\purealt_{\play}} 
\newmacro{\Zdyn}{\diffdyn} 
\newmacro{\effmirrori}{\effmirror_{\play}}
\newmacro{\henergy}{\fench_{\wei, \heq}} 

\newcommand{\restatableeq}[3]{\label{#3}#2\gdef#1{#2\tag{\ref{#3}}}}

\newacro{DO}{dual orbit} 
\newacro{TP}{trajectory of play} 
\newacroplural{TP}{trajectories of play}
\newacro{PR}{Poincaré recurrent} 

\addauthor{BP}{Blue}

\begin{document}


\title
[No-Regret Learning in Harmonic Games]
{No-Regret Learning in Harmonic Games:\\
Extrapolation in the Face of Conflicting Interests}	

\author
[Legacci]
{Davide Legacci$^{c,\ast}$}
\address
{$^{c}$%
Corresponding author.}
\address
{$^{\ast}$%
Univ. Grenoble Alpes, CNRS, Inria, Grenoble INP, LIG, 38000 Grenoble, France.}
\EMAIL{davide.legacci@univ-grenoble-alpes.fr}
\author
[Mertikopoulos]
{Panayotis Mertikopoulos$^{\ast}$}
\EMAIL{panayotis.mertikopoulos@imag.fr}
\author
[Papadimitriou]
{\\Christos H.~Papadimitriou$^{\diamond,\sharp}$}
\address
{$^{\diamond}$%
Columbia University, NYC.}
\address
{$^{\sharp}$%
Archimedes/Athena RC, Greece.}
\EMAIL{christos@columbia.edu}
\author
[Piliouras]
{Georgios Piliouras$^{\P}$}
\address
{$^{\S}$%
Google DeepMind, London, UK.}
\EMAIL{gpil@google.com}
\author
[Pradelski]
{Bary~ S.~R.~Pradelski$^{\S}$}
\address
{$^{\S}$%
CNRS, Maison Française d'Oxford, 2–10 Norham Road, Oxford, OX2 6SE, United Kingdom.}
\EMAIL{bary.pradelski@cnrs.fr}


\thanks{The authors are grateful to Victor Boone and Marco Scarsini for fruitful discussions.}

\newacro{LHS}{left-hand side}
\newacro{RHS}{right-hand side}
\newacro{iid}[i.i.d.]{independent and identically distributed}
\newacro{wlog}[w.l.o.g.]{without loss of generality}
\newacro{lsc}[l.s.c.]{lower semi-continuous}
\newacro{ODE}{ordinary differential equation}
\newacro{VI}{variational inequality}
\newacroplural{VI}{variational inequalities}

\newacro{NE}{Nash equilibrium}
\newacroplural{NE}[NE]{Nash equilibria}
\newacro{CE}{correlated equilibrium}
\newacroplural{CE}[CE]{correlated equilibria}
\newacro{CCE}{coarse correlated equilibrium}
\newacroplural{CCE}[CCE]{coarse correlated equilibria}
\newacro{PE}{preference-equivalent}
\newacro{SE}{strategically equivalent}

\newacro{ZSG}{zero-sum game}
\newacro{2ZSG}{two-player zero-sum game}
\newacro{NZSG}{$\nPlayers$-player zero-sum game}
\newacro{PG}{potential game}
\newacro{HG}{harmonic game}
\newacro{GHG}{generalized harmonic game}
\newacro{UHG}{uniform harmonic game}

\newacro{FTL}{follow-the-leader}
\newacro{FTRL}{follow-the-regularized-leader}
\newacro{FTRL+}{extrapolated \acs{FTRL}}
\newacro{ExtraFTRL}{Extra-Step \acs{FTRL}}
\newacro{OptFTRL}{optimistic \acs{FTRL}}
\newacro{FTLL}{follow-the-linearized-leader}
\newacro{RLD}{regularized learning dynamics}
\newacro{Q}{regularized best response}

\newacro{MD}{mirror descent}
\newacro{DA}{dual averaging}
\newacro{DE}{dual extrapolation}
\newacro{MP}{mirror-prox}
\newacro{OMD}{optimistic mirror descent}

\newacro{GD}{gradient descent}
\newacro{EG}{extra-gradient}
\newacro{OG}{optimistic gradient}

\newacro{EW}{exponential\,/\,multiplicative weights}
\newacro{EXP3}[\ExpThree]{exponential weights algorithm for exploration and exploitation}
\newacro{MWU}{multiplicative weights update}
\newacro{OMW}{optimistic multiplicative weights}
\newacro{IWE}{importance-weighted estimator}


\newmacro{\HG}{\fingame_{\meas}}	
\newmacro{\HGfull}{\HG(\players,\pures,\pay)}	

\newacro{GAN}{generative adversarial network}

\newacro{PG}{potential game}
\newacro{HG}{harmonic game}

\newacro{EW}{exponential\,/\,multiplicative weights}
\newacro{RD}{replicator dynamics}

\begin{abstract}
%
%
The long-run behavior of multi-agent learning \textendash\ and, in particular, \emph{no-regret learning} \textendash\ is relatively well-understood in potential games, where players have aligned interests.
By contrast, in harmonic games \textendash\ the strategic counterpart of potential games, where players have \textit{conflicting} interests \textendash\ very little is known outside the narrow subclass of $2$-player zero-sum games with a fully-mixed equilibrium.
Our paper seeks to partially fill this gap by focusing on the full class of (generalized) harmonic games and examining the convergence properties of \acf{FTRL}, the most widely studied class of no-regret learning schemes.
As a first result, we show that the continuous-time dynamics of \ac{FTRL} are \emph{Poincaré recurrent}, that is, they return arbitrarily close to their starting point infinitely often, and hence fail to converge.
In discrete time, the standard, ``vanilla'' implementation of \ac{FTRL} may lead to even worse outcomes, eventually trapping the players in a perpetual cycle of best-responses.
However, if \ac{FTRL} is augmented with a suitable extrapolation step \textendash\ which includes as special cases the optimistic and mirror-prox variants of \ac{FTRL} \textendash\ we show that learning converges to a \acl{NE} from any initial condition, and all players are guaranteed at most $\bigoh(1)$ regret.
These results provide an in-depth understanding of no-regret learning in harmonic games, nesting prior work on $2$-player zero-sum games, and showing at a high level that harmonic games are the canonical complement of potential games, not only from a strategic, but also from a dynamic viewpoint.
\end{abstract}
\acresetall	

\allowdisplaybreaks	
\acresetall	
\maketitle

\section{Introduction}
\label{sec:introduction}

The question of ``as if'' rationality \textendash\ that is, whether selfishly-minded, myopic agents may learn to behave ``\define{as if}'' they were fully rational \textendash\ has been one of the cornerstones of non-cooperative game theory, and for good reason.
Especially in modern applications of game theory to machine learning and data science \textendash\ from online ad auctions to recommender systems and multi-agent reinforcement learning \textendash\ the standard postulates of rationality (knowledge of the game, capacity to compute an equilibrium, flawless execution of equilibrium strategies, common knowledge of rationality, etc.) are almost never met in practice;
as a result, game-theoretic predictions that rely on these assumptions are likewise put into question.
By contrast, given the ease of implementing and deploying cheap, computationally efficient learning algorithms and policies at a large scale, it is often more logical to turn to the policy being deployed as the object of interest.
The aim is then to understand its long-run behavior \textendash\ and, in particular, whether it ultimately leads to equilibrium.

A major obstacle in this approach is the complexity of computing a \acl{NE}, a problem which is known to be complete for \textsf{PPAD} \textendash\ and hence intractable \textendash\ by the seminal work of \citet{DGP09-acm}.
This result implies that it is not plausible to expect any algorithm to converge to \acl{NE} in \emph{all} games (at least, not in a reasonable amount of time), so it dovetails naturally with the impossibility results of \citet{HMC00,HMC06} who showed that there are no uncoupled learning dynamics that converge to \acl{NE} in all games.
On that account, it is natural to ask in which classes of games we can expect a learning algorithm to converge, in which classes we cannot, and under what conditions.

Perhaps the most well-behaved class of games in terms of learning is the class of \define{potential games} \cite{MS96,San01}, where players have \emph{common} interests \textendash\ not necessarily driving them to play the same strategy, but in the sense that externalities are symmetric and aligned along a common objective (the potential of the game).
In this class of games, the behavior of learning dynamics \textendash\ and, in particular, no-regret learning \cite{daskalakis2021nearoptimal,GPD20,LZMJ20,MZ19,NEURIPS2020_db346ccb,MHC24,HAM21,HACM22,SALS15} \textendash\ are relatively well understood, and there is a wide range of equilibrium convergence results, from continuous  to discrete time, and even with bandit, payoff-based feedback \cite{HS98,San01,HCM17}.

By contrast, in the presence of \emph{conflicting} interests, the situation can be quite different.
In \aclp{2ZSG} with a fully-mixed equilibrium \textendash\ such as Matching Pennies \textendash\ the continuous-time dynamics of no-regret, regularized learning are recurrent in the sense of Poincaré \textendash\ that is, the induced trajectory of play returns arbitrarily close to where it started infinitely many times \cite{PS14,MPP18}.
In discrete time, the situation becomes more complicated:
the vanilla version of \ac{FTRL} \textendash\ the most widely studied family of no-regret algorithms \textendash\ is no longer recurrent, but it diverges away from equilibrium in the same class of games \cite{MLZF+19,GBVV+19}.
On the other hand,
if players employ an optimistic / extra-gradient variant of \ac{FTRL}, the induced trajectory of play converges to equilibrium \cite{MLZF+19,pmlr-v151-fasoulakis22a} and, under certain conditions, it is even possible to show that it converges at a geometric rate \cite{WLZL21}.

At the same time, zero-sum games may also admit a potential function, so it is not possible to predict the outcome of a learning process based on where it stands along the potential\,/\,zero-sum axis.
The non-trivial intersection of these classes means that potential and zero-sum games are \emph{not} complementary, and this, not only from a strategic, but also from a dynamic viewpoint.
Instead, the true strategic complement of potential games is the class of \define{harmonic games}.
This class was first considered by \citet{CMOP11}, who established a remarkable decomposition result:
Every game in normal form can be decomposed as the sum of a potential game and a harmonic game, and this decomposition is unique up to affine transformations that do not alter the equilibrium outcomes of the game.
In particular, the class of potential and harmonic games intersect trivially (up to strategic equivalence), and all \aclp{2ZSG} with an interior equilibrium are harmonic, thus lending credence to the fact that it is harmonic games, not \aclp{ZSG}, that correctly capture the notion of conflicting interests in this context.
This raises the following natural question:
\begin{center}
\itshape
What is the behavior of no-regret algorithms and dynamics in harmonic games?
\end{center}

Except for a very recent paper by \citet{LMP24} (which we discuss below), almost nothing is known on this question.
Accordingly, against this backdrop, our contributions can be summarized as follows:
\begin{enumerate}
\item
Starting with a continuous-time model of no-regret learning, we show that all \ac{FTRL} dynamics are Poincaré recurrent in all harmonic games.
This generalizes and extends the recent result of \citet{LMP24} for the \acl{RD} in uniform harmonic games (a subclass of harmonic games in which the uniform distribution is always a \acl{NE}).%
\footnote{In more detail, the way that \citet{LMP24} obtained their result hinges on the so-called \emph{Shahshahani metric}, a choice which is essentially ``mandated'' by the structure of the replicator dynamics.
Specifically, the key property of the Shahshahani metric is that incompressibility of the replicator field is equivalent to the underlying game being uniformly harmonic;
however, finding a variant of the Shahshahani metric attuned to \ac{FTRL} seems to be a formidable task, and likewise for non-uniform harmonic games.
Because of this, the ``incompressibility'' approach of \cite{LMP24} does not seem applicable to our setting \textendash\ at least, not in a straightforward way.}
\item
In discrete-time models of learning, the standard implementation of \ac{FTRL} cannot be expected to converge (since it fails to do so in Matching Pennies).
To correct this behavior, we consider a flexible algorithmic template, inspired by \citet{AIMM24} and dubbed \acdef{FTRL+}, which augments \ac{FTRL} with a forward-looking, extrapolation step (including as special cases the optimistic and extra-step variants of \ac{FTRL}, \cf \cref{sec:algorithms}).
We then establish the following results:
\begin{enumerate}
\item
Under \acl{FTRL+}, players are guaranteed constant individual regret (so, as a consequence, the players' empirical frequency of play converges to \acl{CCE} at a rate of $\bigoh(1/\nRuns)$).%
\footnote{We clarify here that ``constant'' refers to the horizon $\nRuns$ of play;
the dependence on the number of actions may be logarithmic or worse (depending on the specific regularized learning scheme employed by the players).}
This should be contrasted with the results of \cite{daskalakis2021nearoptimal,farina2022nearoptimal} who showed that players can achieve \emph{polylogarithmic} regret in any game (finite or convex).
\item
The induced trajectory of play converges to \acl{NE} from any initial condition.
\end{enumerate}
\end{enumerate}
Our results aim to provide an in-depth understanding of no-regret learning in harmonic games, nesting prior work on $2$-player zero-sum games \textendash\ from Poincaré recurrence \cite{PS14,MPP18} to constant regret \cite{HAM21} and convergence under optimistic / extra-gradient schemes \cite{MLZF+19,WLZL21,DP19,pmlr-v151-fasoulakis22a,GBVV+19}.
In partiucular, at a high level, our results show that harmonic games are the canonical complement of potential games, not only from a strategic, but also from a dynamic, learning viewpoint.
\acresetall	

\section{Preliminaries}
\label{sec:prelims}

\subsection{Preliminaries on finite games}

Throughout the sequel, we will work with \define{finite games in normal form}.
Formally, such games consist of
\begin{enumerate*}
[\upshape(\itshape i\hspace*{.5pt}\upshape)]
\item
a finite set of \define{players} $\play\in\players \equiv \{1,\dotsc,\nPlayers\}$;
\item
a finite set of \define{actions} $\pures_{\play}$ per player $\play\in\players$;
and
\item
an ensemble of \define{payoff functions} $\pay_{\play} \from \prod_{\playalt}\pures_{\playalt} \to \R$, each determining the reward $\pay_{\play}(\pure)$ of player $\play\in\players$ in a given action profile $\pure = (\pure_{1},\dotsc,\pure_{\nPlayers})$.
\end{enumerate*}
Putting everything together, we will write $\pures \defeq \prod_{\play} \pures_{\play}$ for the game's \define{action space} and $\fingame \equiv \fingamefull$ for the game with primitives as above.

During play, each player selects an action according to some \define{mixed strategy}, that is, a probability distribution $\strat_{\play}$ over $\pures_{\play}$ which assigns probability $\strat_{\play\pure_{\play}}$ to $\pure_{\play}\in\pures_{\play}$.
In a slight abuse of notation, if $\strat_{\play}$ assigns all probability mass to some action $\pure_{\play}\in\pures_{\play}$ (that is, $\strat_{\play\pure_{\play}} = 1$), we will identify $\strat_{\play}$ with $\pure_{\play}$ and we will call it \define{pure}.
We will also write
$\strats_{\play} \defeq \simplex(\pures_{\play}) \subseteq \primalambientplay$ for the mixed strategy space of player $\play$,
$\strat = (\strat_{1},\dotsc,\strat_{\nPlayers})$ for the \define{strategy profile} collecting the strategies of all players,
and
$\strats \defeq \prod_{\play}\strats_{\play}$ for the game's \define{strategy space}.

The \define{mixed payoff} of player $\play$ under a mixed strategy profile $\strat\in\strats$ may then be written as
\begin{equation}
\label{eq:payfield}
\pay_{\play}(\strat)
	= \exwrt{\pure\sim\strat}{\pay_{\play}(\pure)}
	= \sum_{\pure\in\pures} \pay_{\play}(\pure) \, \strat_{\pure}
	= \sum_{\pure_{\play}\in\pures_{\play}} \pay_{\play}(\pure_{\play};\strat_{-\play}) \, \strat_{\play\pure_{\play}}
\end{equation}
where
$\strat_{\pure} \defeq \prod_{\play} \strat_{\play\pure_{\play}}$ denotes the joint probability of $\pure = (\pure_{1},\dotsc,\pure_{\nPlayers}) \in \pures$ under $\strat \in \strats$,
and, in standard game-theoretic notation,
we write $(\strat_{\play};\strat_{-\play}) = (\strat_{1},\dotsc,\strat_{\play},\dotsc,\strat_{\nPlayers})$ for the profile where player $\play$ plays $\strat_{\play}\in\strats_{\play}$ against the strategy $\strat_{-\play} \in \strats_{-\play} \defeq \prod_{\playalt\neq\play} \strats_{\playalt}$ of all other players.
We also respectively define the \define{individual payoff field} of player $\play$ and the \define{game's payoff field} as
\begin{equation}
\label{eq:payv}
\payfield_{\play}(\strat)
	= (\pay_{\play}(\pure_{\play};\strat_{-\play}))_{\pure_{\play}\in\pures_{\play}}
	\quad
	\text{and}
	\quad
\payfield(\strat)
	= (\payfield_{1}(\strat),\dotsc,\payfield_{\nPlayers}(\strat))
\end{equation}
so $\pay_{\play}(\strat) = \sum_{\pure_{\play}\in\pures_{\play}} \payfield_{\play\pure_{\play}}(\strat) \strat_{\play\pure_{\play}} \equiv \dualp{\payfield_{\play}(\strat)}{\strat_{\play}}$, where
$\dualp{\argdot}{\argdot}$ is the standard duality pairing on $\ambienti$.
By multilinearity, each player's individual payoff field is Lipschitz continuous on $\strats$, and we will write $\lips_{\play}$ for its Lipschitz modulus, that is
\begin{equation}
\label{eq:Lips}
\dnorm{\payfield_{\play}(\stratalt) - \payfield_{\play}(\strat)}
	\leq \lips_{\play} \norm{\stratalt - \strat}
	\quad
	\text{for all $\strat,\stratalt\in\strats$}.
\end{equation}
\begin{remark*}
In the above and throughout, $\norm{\cdot}$ denotes an ambient norm on $\ambienti$ (usually the $L^{1}$ norm), and $\dnorm{\cdot}$ is the corresponding dual norm (usually the $L^{\infty}$ norm).
To simplify notation, we will not carry the player index $\play$ in $\norm{\cdot}$, and we will instead rely on the context to resolve any ambiguities.
\end{remark*}

In terms of solution concepts, we will focus almost exclusively on the notion of a \acdef{NE}, \ie a strategy profile $\eq\in\strats$ that is unilaterally stable in the sense that
\begin{equation}
\label{eq:Nash}
\tag{NE}
\pay_{\play}(\eq)
	\geq \pay_{\play}(\strat_{\play};\eq_{-\play})
	\quad
	\text{for all $\strat_{\play}\in\strats_{\play}$, $\play\in\players$}
	\eqdot
\end{equation}
Equivalently, \eqref{eq:Nash} can be expressed in terms of the game's payoff field as a variational inequality of the form
\begin{equation}
\tag{VI}
\label{eq:VI}
\dualp{\payfield(\eq)}{\strat - \eq}
	\leq 0
	\quad
	\text{for all $\strat\in\strats$}
	\eqdot
\end{equation}
Thus, writing $\supp(\eq_{\play}) = \setdef{\pure_{\play}\in\pures_{\play}}{\strat_{\play\pure_{\play}} > 0}$ for the \define{support} of $\eq_{\play}$, it follows that $\eq$ is a \acl{NE} if and only if $\pay_{\play}(\pure_{\play};\eq_{-\play}) \geq \pay_{\play}(\purealt_{\play};\eq_{-\play})$ for all $\pure_{\play}\in\supp(\eq_{\play})$ and all $\purealt_{\play}\in\pures_{\play}$, $\play\in\players$.
We will use all this freely in the rest of our paper.

\subsection{Harmonic games}

Our main focus in what follows will be the class of \define{harmonic games}, first introduced by \citet{CMOP11} as a game-theoretic model for strategic situations with conflicting, anti-aligned interests.
Specifically, as was shown by \citet{CMOP11} \textendash\ and, in a more general setting, by \citet{APSV22} \textendash\ every game in normal form can be decomposed as the sum of a potential game and a harmonic game, and this decomposition is unique up to affine transformations that do not alter the equilibrium outcomes of the game.%
\footnote{We briefly recall here that $\fingame \equiv \fingamefull$ is a potential game if it admits a \define{potential function} $\pot\from\strats\to\R$ such that
\(
\pay_{\play}(\purealt_{\play};\pure_{-\play}) - \pay_{\play}(\pure_{\play};\pure_{-\play})
	= \pot(\purealt_{\play};\pure_{-\play}) - \pot(\pure_{\play};\pure_{-\play})
\)
for all $\pure,\purealt\in\pures$ and all $\play\in\players$ \citep{MS96}.}
In this decomposition, the potential component of a game captures multi-agent strategic interactions with \emph{common} interests, whereas the harmonic component covers interactions with \emph{conflicting} interests.%
\footnote{The terminology ``harmonic'' is due to \citet{CMOP11} and alludes to the harmonic component of the graphical Hodge decomposition \citep{jiangStatisticalRankingCombinatorial2011}.
}

Formally, adapting the more general setup by \citet{APSV22}, we have the following definition:%
\begin{definition}
\label{def:harmonic}
A finite game $\fingame \equiv \fingamefull$ is said to be \define{harmonic} when it admits a \define{harmonic measure}, \ie a collection of weights $\meas_{\play\pure_{\play}}\in(0,\infty)$, $\pure_{\play}\in\pures_{\play}$, $\play\in\players$, such that
\begin{equation}
\tag{HG}
\restatableeq{\EqHarmonic}{
\sum_{\play \in \players}
	\sum_{\purealt_{\play} \in \pures_{\play}}
	\meas_{\play\purealt_{\play}}
	\bracks{ \pay_{\play}(\pure_{\play};\pure_{-\play}) - \pay_{\play}(\purealt_{\play};\pure_{-\play}) }
	= 0
	\quad \text{for all } \pure \in \pures}
{eq:harmonic} \eqdot
\end{equation}
In particular, if $\fingame$ is harmonic relative to the uniform measure $\meas_{\play\pure_{\play}} = 1$, $\pure_{\play}\in\pures_{\play}$, $\play\in\players$, we will say that $\fingame$ is a \acdef{UHG}.
\end{definition}

\begin{remark*}
With regard to terminology, \citet{CMOP11} call ``\aclp{HG}'' what we call ``\aclp{UHG}'', and \citet{APSV22} call ``$\meas$-harmonic games'' what we call ``\aclp{HG}''.%
\footnote{To be even more precise, the definition of \citet{APSV22} involves an additional set of weights, called a \define{comeasure};
however, as we explain in \cref{app:harmonic}, these weights do not change the preference structure of the game, so we disregard this extra degree of generality.}
We use this convention because it simultaneously simplifies notation and terminology while capturing all relevant strategic features of the game;
for a detailed discussion, see \cref{app:harmonic}.
To avoid needless repetition, and unless there is a danger of confusion, when we say that $\fingame$ is harmonic, we will write $\meas_{\play}$ for the corresponding measure, and we will write $\mass_{\play} = \lnorm{\meas_{\play}} = \sum_{\purealt_{\play}\in\pures_{\play}} \meas_{\play\purealt_{\play}}$ for the total mass of $\meas_{\play}$.
\endenv
\end{remark*}

Broadly speaking, in harmonic games, for any player considering a deviation toward a specific pure strategy profile, there exist other players with an incentive to deviate \emph{away} from said profile.
In this regard, \aclp{HG} can be seen as the strategic complement of potential games, where player interests are aligned and sequences of unilateral best responses generate a finite improvement path that terminates at a pure \acl{NE} \citep{MS96}.
By contrast, except for trivial cases (like the zero game) \aclp{HG} \emph{do not} admit pure \aclp{NE}, and they possess non-terminating best-response paths.
For all these reasons, harmonic games can be considered as ``orthogonal'' to potential games, in a sense made precise by the decomposition results of \citet{CMOP11} and \citet{APSV22}.

It is of course natural to ask what is the relation between \aclp{HG} and \aclp{ZSG}.
Games belonging to the latter class \textendash\ such as Matching Pennies and Rock-Paper-Scissors \textendash\ have long been used as prototypical examples of strategic conflict;
at the same time, there are \aclp{ZSG} that are also potential (and even possess strict equilibria), so the potential\,/\,zero-sum distinction does not capture the whole picture.
As a matter of fact, it is not a coincidence that the textbook examples of zero-sum games admit fully-mixed \aclp{NE}:
as we discuss in \cref{app:harmonic}, \aclp{2ZSG} with a fully mixed \acl{NE} are harmonic,
so the existing results for such games are, in a sense, more closely attuned to their harmonic character.

\section{Continuous-time analysis: Poincaré recurrence}
\label{sec:dynamics}

The most basic rationality postulate in the context of online learning is the minimization of a player's (external) regret, \ie the difference between a player's cumulative payoff and that of the player's best possible strategy in hindsight.
In more detail, assuming for the moment that play evolves in continuous time, the \define{regret} of player $\play\in\players$ relative to a sequence of play $\stratof{\time} \in \strats$ is defined as
\begin{equation}
\label{eq:reg-cont}
\reg_{\play}(\horizon)
	= \max_{\bench_{\play}\in\strats_{\play}}
		\int_{\tstart}^{\horizon} \bracks{\pay_{\play}(\bench_{\play};\stratofx{\others}{\time}) - \pay_{\play}(\stratof{\time})} \dd\time
\end{equation}
and we say that the player has \define{no regret} under $\stratof{\time}$ if $\reg_{\play}(\horizon) = o(\horizon)$ as $\horizon\to\infty$.

The most widely used scheme for attaining no regret is the family of policies known as \acdef{FTRL} \cite{SSS06,SS11}.
At a high level, the idea behind \ac{FTRL} is that, at all times $\time\geq\tstart$, each player $\play\in\players$ plays a mixed strategy $\stratofi{\time} \in \strats_{\play}$ that maximizes the player's cumulative payoff up to time $\time$ minus a certain regularization penalty.
In our continuous-time setting, this gives rise to the \ac{FTRL} dynamics
\begin{equation}
\restatableeq{\EqFtrlx}{
\stratofi{\time}
	= \argmax_{\bench_{\play}\in\strats_{\play}}
		\braces*{
			\int_{\tstart}^{\time} \!\!\! \pay_{\play}(\bench_{\play};\stratofx{-\play}{\timealt}) \dd\timealt
			- \hreg_{\play}(\bench_{\play})}
	= \argmax_{\bench_{\play}\in\strats_{\play}}
		\braces*{
			\int_{\tstart}^{\time} \! \braket{\payv_{\play}(\stratof{\timealt})}{\bench_{\play}} \dd\timealt
			- \hreg_{\play}(\bench_{\play})}
}{eq:FTRL-x}
\end{equation}
or, more compactly,
\begin{equation}
\tag{FTRL-D}
\restatableeq{\EqFtrlCont}{
\dotscoreofi{\time}
	= \payfield_{\play}(\stratof{\time})
	\qquad
\stratofi{\time}
	= \mirror_{\play}(\scoreofi{\time})
}{eq:FTRL-cont}
\end{equation}
where
$\hreg_{\play}\from\strats_{\play}\to\R$ is a convex penalty function known as the \define{regularizer} of the method,
$\mirror_{\play}$ denotes the \define{regularized choice map} of player $\play$,
and
$\mirror = (\mirror_{1},\dotsc,\mirror_{\nPlayers})$
denotes the profile thereof.
Formally, writing $\scores_{\play} \equiv \R^{\nPures_{\play}}$ for the \define{payoff space} of player $\play\in\players$
\textendash\ that is, the space of all possible payoff vectors $\payfield_{\play}$ of player $\play$ \textendash\
the regularized choice map $\mirror_{\play}\from\scores_{\play}\to\strats_{\play}$ is defined as
\begin{equation}
\restatableeq{\EqMirrori}{
\mirror_{\play}(\score_{\play})
	= \argmax\nolimits_{\strat_{\play}\in\strats_{\play}} \{ \braket{\score_{\play}}{\strat_{\play}} - \hreg_{\play}(\strat_{\play}) \}
	\quad
	\text{for all $\score_{\play} \in \scores_{\play}$}
	\eqdot
}{eq:mirrori}
\end{equation}
In essence, $\mirror_{\play}$ is a ``soft'' version of the $\argmax$ correspondence $\score_{\play} \mapsto \argmax_{\strat_{\play}\in\strats_{\play}} \braket{\score_{\play}}{\strat_{\play}}$, suitably regularized by a penalty term intended to incentivize exploration.
For technical reasons, we will also assume that each $\hreg_{\play}$ is \define{strongly convex}, \ie
\begin{equation}
\label{eq:hstri}
\hreg_{\play}(t\strat_{\play} + (1-t)\stratalt_{\play})
	\leq t \hreg_{\play}(\strat_{\play})
		+ (1-t) \hreg_{\play}(\stratalt_{\play})
		- \tfrac{1}{2} \hstr_{\play} t(1-t) \norm{\strat_{\play} - \stratalt_{\play}}^{2}
\end{equation}
for some $\hstr_{\play} > 0$ (commonly referred to as the \emph{strong convexity modulus} of $\hreg_{\play}$), and for all $\strat_{\play},\stratalt_{\play}\in\strats$, $t\in[0,1]$.
In plain words, this simply means that $\hreg_{\play}$ has ``enough curvature'' in the sense that it can be bounded from below by a (positive) quadratic function which agrees with $\hreg_{\play}$ to first order.

The go-to example of this setup is the entropic regularizer
\begin{equation}
\label{eq:entropy}
\hreg_{\play}(\strat_{\play})
	= \sum_{\pure_{\play}\in\pures_{\play}} \strat_{\play\pure_{\play}} \log\strat_{\play\pure_{\play}}
\end{equation}
which yields the so-called \define{logit choice map}
\begin{equation}
\label{eq:logit}
\mirror_{\play}(\score_{\play})
	\equiv \logit_{\play}(\score_{\play})
	\defeq \frac{(\exp(\score_{\play\pure_{\play}}))_{\pure_{\play}\in\pures_{\play}}}{\sum_{\pure_{\play}\in\pures_{\play}}\exp(\score_{\play\pure_{\play}})}
	\quad
	\text{for all $\score_{\play}\in\scores_{\play}$}.
\end{equation}
By Pinsker's inequality, the entropic regularizer is $1$-strongly convex relative to the $L^{1}$-norm on $\strats_{\play}$ \citep{SS11}, and by a standard calculation \cite{Rus99,MM10}, the induced sytem \eqref{eq:FTRL-cont} boils down to the \acl{RD} of \citet{TJ78}.
Some other standard examples of \eqref{eq:FTRL-cont} include
the Euclidean projection dynamics of \citet{Fri91} when $\hreg_{\play}(\strat_{\play}) = (1/2) \twonorm{\strat_{\play}}^{2}$,
the $q$-\acl{RD} \cite{Har11,MS16},
etc.
To streamline our presentation, we defer a detailed discussion of these examples to \cref{app:continuous},
and we proceed below to state the main regret guarantee of \eqref{eq:FTRL-cont}, originally due to \cite{KM17}:

\begin{theorem}
\label{thm:reg-cont}
Under \eqref{eq:FTRL-cont}, each player's regret is bounded as $\reg_{\play}(\horizon) \leq \hrange_{\play} \defeq \max\hreg_{\play} - \min\hreg_{\play}$.
\end{theorem}

\Cref{thm:reg-cont} showcases the strong no-regret properties of \eqref{eq:FTRL-cont}:
it is not possible to guarantee less than constant, $\bigoh(1)$ regret, so \eqref{eq:FTRL-cont} is optimal in this regard.
In turn, by standard results \cite{NRTV07}, \cref{thm:reg-cont} implies further that the players' (correlated) empirical frequencies $\corrofx{\pure_{1},\dotsc,\pure_{\nPlayers}}{\time} \defeq (1/\time) \int_{\tstart}^{\time} \prod_{\play} \stratofx{\play\pure_{\play}}{\timealt} \dd\timealt$ converge to the game's set of \acfp{CCE} at a rate of $\bigoh(1/\time)$.

Importantly, this result makes no assumptions about the underlying game, but it does not carry the same predictive power in all games:
for one thing, a game's set of \acp{CCE} may include highly non-rationalizable outcomes (such as dominated strategies and the like) \citep{VZ13};
for another, the time-averaging that is inherent in the definition of empirical distributions may conceal a wide range of non-convergence phenomena, from cycles to chaos \cite{PS14,SAF02}.
On that account, the day-to-day behavior of \eqref{eq:FTRL-cont} in harmonic games cannot be understood from \cref{thm:reg-cont} alone, and requires a closer, more in-depth look.

Our first result below provides such a lense and shows that \eqref{eq:FTRL-cont} is almost-periodic in harmonic games, a property known as \define{Poincaré recurrence.}

\begin{restatable}{theorem}{ThRecurrence}
\label{thm:recurrence}
Suppose $\fingame$ is harmonic.
Then almost every orbit $\stratof{\time}$ of \eqref{eq:FTRL-cont} returns arbitrarily close to its starting point infinitely often:
specifically, for \textpar{Lebesgue} almost every initial condition $\stratof{\tstart} = \mirror(\scoreof{\tstart}) \in \strats$, there exists an increasing sequence of times $\curr[\time]\uparrow\infty$ such that $\stratof{\curr[\time]} \to \stratof{\tstart}$.
\end{restatable}

An immediate consequence of \cref{thm:recurrence} is that no-regret learning under \eqref{eq:FTRL-cont} fails to converge in \emph{any} harmonic game;
in particular, since the orbits of \eqref{eq:FTRL-cont} eventually return to (almost) where they started, it is debatable if the players have learned anything at all,
despite the fact that they incur at most constant regret.
This cyclic, non-convergent landscape is the polar opposite of the long-run behavior of \eqref{eq:FTRL-cont} in \emph{potential} games, where the dynamics are known to converge globally \cite{HCM17}.
Thus, in addition to the strategic viewpoint of the previous section, \cref{thm:recurrence} shows that harmonic games are orthogonal to potential games also from a \emph{dynamic} viewpoint.

\cref{thm:recurrence} also provides a far-reaching generalization of existing results on Poincaré recurrence in (possibly networked) \aclp{2ZSG} with an interior equilibrium \citep{MPP18} to general-sum, $\nPlayers$-player games.
Combined with our previous remark, and given that the zero-sum property is not as meaningful for $\nPlayers$ players as it is for two,%
\footnote{Recall that any $\nPlayers$-player game can be turned into an equivalent zero-sum game by adding a fictitious player.}
the class of harmonic games can be seen as the more natural $\nPlayers$-player generalization of \aclp{2ZSG} from a learning viewpoint.

To the best of our knowledge, the only comparable result to \cref{thm:recurrence} in the literature is the very recent paper of \citet{LMP24} who showed that the \acl{RD} \textendash\ a special case of \eqref{eq:FTRL-cont} \textendash\ are Poincaré recurrent in \emph{uniform} harmonic games, that is, in harmonic games where the uniform distribution is a \acl{NE}, \cf \eqref{eq:harmonic-uniform} and the discussion surrounding \cref{def:harmonic}.
In this regard, \cref{thm:recurrence} extends the recent results of \citet{LMP24} along two axes:
\begin{enumerate*}
[\upshape(\itshape i\hspace*{.5pt}\upshape)]
\item
it applies to the entire class of \ac{FTRL} dynamics (not only the \acl{RD});
and
\item
it applies to the entire class of harmonic games (and not only \emph{uniformly} harmonic games).
\end{enumerate*}

In terms of techniques, \citet{LMP24} obtained their result through a surprising connection between a certain Riemannian metric underlying the \acl{RD} and the defining relation of uniformly harmonic games.
This relation no longer holds for different instances of \eqref{eq:FTRL-cont} or for non-uniform harmonic games, so the techniques of \cite{LMP24} cannot be extended \textendash\ and, in fact, \citet{LMP24} stated this generalization as an open problem.
Our techniques instead rely on the fact that the orbits $\scoreof{\time}$ of \eqref{eq:FTRL-cont} comprise a volume-preserving flow in the game's payoff space $\scores \equiv \prod_{\play}\scores_{\play}$ (though not necessarily on $\strats$), and then deriving a suitable constant of motion.
In the case of the logit map \eqref{eq:logit}, this constant of motion can be written as
\begin{equation}
\label{eq:energy-special}
G(\strat)
	= \prod_{\play\in\players}
		\prod_{\pure_{\play}\in\pures_{\play}}
		\strat_{\play\pure_{\play}}^{\meas_{\play\pure_{\play}}}
	\qquad
	\text{for all $\strat\in\strats$},
\end{equation}
where $\meas = (\meas_{\play\pure_{\play}})_{\pure_{\play}\in\pures_{\play},\play\in\players}$ is the harmonic measure on $\strats$ defining $\fingame$.
In the more general case,
the construction of a constant of motion for \eqref{eq:FTRL-cont} involves a characterization of harmonic games in terms of a ``strategic center'',
which we carry out in detail in \cref{app:continuous}.

\section{Discrete-time analysis: Convergence and constant regret via extrapolation}
\label{sec:algorithms}

We now proceed to examine the regret and convergence properties of regularized learning algorithms in harmonic games.
Starting with the standard, vanilla implementation of \ac{FTRL}, we reproduce a well-known observation that \ac{FTRL} spirals out to a non-terminating cycle of best-responses in Matching Pennies (which is a harmonic game).
Subsequently, to correct this non-convergent behavior, we examine a flexible algorithmic template, which we call \acdef{FTRL+}, and which includes as special cases the optimistic and extra-gradient versions of \ac{FTRL}.

\subsection{Vanilla implementation of \ac{FTRL}}

Building on the discussion of the previous section, the standard implementation of \ac{FTRL} in discrete time for $\run=\running$ is
\begin{equation}
\nexti
	= \argmax_{\bench_{\play}\in\strats_{\play}}
		\braces*{\insum_{\runalt=\start}^{\run} \pay_{\play}(\bench_{\play};\state_{-\play,\run})
			- \regweight_{\play}\hreg_{\play}(\bench_{\play})}
	= \argmax_{\bench_{\play}\in\strats_{\play}}
		\braces*{\insum_{\runalt=\start}^{\run}\braket*{\payfield_{\play}(\iter)}{\bench_{\play}}
			- \regweight_{\play} \hreg_{\play}(\bench_{\play})}
\end{equation}
or, in more compact, iterative notation
\begin{equation}
\label{eq:FTRL}
\tag{FTRL}
\nexti[\statealt]
	= \curri[\statealt] + \learn_{\play} \payfield_{\play}(\curr)
	\qquad
\curri
	= \mirror_{\play}(\curri[\statealt])
\end{equation}
where,
as per \eqref{eq:mirrori}, the map $\mirror_{\play}\from\scores_{\play}\to\strats_{\play}$ denotes the \define{regularized choice map} of player $\play\in\players$,
$\regweight_{\play}$ is a player-specific regularization weight parameter,
and
$\learn_{\play} = 1/\regweight_{\play}$ represents the \define{learning rate} of player $\play$.
Apart from their obvious differences \textendash\ discrete \vs continuous time \textendash\ a salient point that sets \eqref{eq:FTRL} apart from \eqref{eq:FTRL-cont} is the inclusion of the parameter $\learn_{\play}$;
this parameter is necessary to control the algorithm's behavior, and we will discuss it in detail in the sequel.

As mentioned in the introduction, a major shortfall of \eqref{eq:FTRL} \textendash\ and one of the main reasons for the increased popularity of optimistic\,/\,extra-gradient methods \textendash\ is that it may spiral away from \acl{NE}, even in simple $2\times2$ games with a unique equilibrium.
The standard example of this behavior is Matching Pennies, a \acl{2ZSG} with a fully-mixed equilibrium which is also uniformly harmonic, so the trajectories of \eqref{eq:FTRL-cont} are Poincaré recurrent (and, in fact, periodic).
In more detail, this game can be compactly represented by the payoff field $\payfield(\strat_{1},\strat_{2}) = (4\strat_{2}-2,2-4\strat_{1})$ for $\strat_{1},\strat_{2} \in [0,1]$, and its unique \acl{NE} is $\eq = (1/2,1/2)$.
Thus, if we run \eqref{eq:FTRL} with a Euclidean reqularizer \textendash\ that is, $\hreg_{\play}(\strat_{\play}) = \strat_{\play}^{2}/2$ for $\play=1,2$ \textendash\ and the same learning rate $\learn$ for both players, a straightforward calculation shows that the distance $\curr[\breg] = \cramped{(\strat_{1,\run} - \eq_{1})^{2}/2 + (\strat_{2,\run} - \eq_{2})^{2}/2}$ between $\curr$ and $\eq$ evolves as
\begin{equation}
\next[\breg]
	= \tfrac{1}{2} (\strat_{1,\run} + \learn\payfield_{1}(\curr) - \eq_{1})^{2}
		+ \tfrac{1}{2} (\strat_{2,\run} + \learn\payfield_{2}(\curr) - \eq_{2})^{2}
	= (1 + 16\learn^{2}) \curr[\breg]
\end{equation}
as long as $\curr + \learn \payfield(\curr) \in \strats$.
In other words, the distance of the iterates of \eqref{eq:FTRL} from the game's equilibrium grows at a geometric rate until $\curr$ reaches the boundary of $\strats$ and is ultimately trapped in a non-terminating cycle of best responses, \cf \cref{fig:portraits}.
In this regard, the rationality properties of \eqref{eq:FTRL} are even worse than those of \eqref{eq:FTRL-cont} because the game's equilibrium is now \emph{repelling}.

\subsection{\Acl{FTRL+}}

To mitigate this undesirable, divergent behavior of \eqref{eq:FTRL}, a standard approach in the literature is the inclusion of a forward-looking, ``\define{extrapolation step}''. Instead of updating the algorithm's ``base state'' $\curr$ directly, players first move to an interim ``leading state'' $\lead$ using payoff information from $\curr$ (this is the extrapolation step);
subsequently, players update $\curr$ using payoff information from the leading state $\lead$, and the process repeats.
In this way, players attempt to anticipate their payoff landscape and, in so doing, to take a more informed update step at each iteration.

The seed of this idea goes back to \citet{Kor76} and \citet{Pop80} in the context of solving monotone \acl{VI} problems, and it has since percolated to a wide array of ``\define{\acl{EG}}'' or ``\define{optimistic}'' methods, such as the \acl{MP} algorithm of \citet{Nem04}, the \acl{DE} variant of \citet{Nes07}, the \acl{OMD} algorithm of \citet{CYLM+12} and \citet{RS13-NIPS}, and many others.
Given the different operational envelope of each of these methods,
we consider below an integrated algorithmic template, which we call \acdef{FTRL+}, and which is sufficiently flexible to account for a broad range of these schemes.

Formally, the proposed algorithmic blueprint unfolds in two phases as follows:
\begin{flalign}
\label{eq:FTRL+}
\tag{\acs{FTRL+}}
\begin{alignedat}{4}
\quad
	a)\;\;
	&\define{Extrapolation phase:}
	&\qquad
\leadi[\statealt]
	&= \curri[\statealt] + \learn_{\play} \curri[\signal]
	&\qquad
\leadi
	&=\mirror_{\play}(\leadi[\statealt])
	\\
\quad
	b)\;\;
	&\define{Update phase:}
	&\qquad
\nexti[\statealt]
	&= \curri[\statealt] + \learn_{\play} \leadi[\signal]
	&\qquad
\nexti
	&=\mirror_{\play}(\nexti[\statealt])
\end{alignedat}
&&
\end{flalign}

In the above,
$\learn_{\play} > 0$ is the learning rate of player $\play$,
$\curr$ and $\lead$ denote respectively the method's \define{base} and \define{leading} states at stage $\run=\running$,
and
$\curri[\signal]$ and $\leadi[\signal]$ are sequences of ``black-box'' payoff models at $\curr$ and $\lead$ respectively.

Specifically, following \citet{AIMM24}, we will assume throughout that
\begin{subequations}
\begin{equation}
\label{eq:signal-lead}
\leadi[\signal]
	= \payfield_{\play}(\lead)
	\quad
	\text{for all $\play\in\players$ and all $\run=\running$}
\end{equation}
\ie players always update the base state $\curr$ using payoff information from the leading state $\lead$.
By contrast, the leading state $\lead$ can be generated in many different ways, depending on the targeted update structure.
In this regard, we will consider the linear model
\begin{equation}
\label{eq:signal-curr-gen}
\curri[\signal]
	= \present_{\play} \, \payfield_{\play}(\curr)
		+ \past_{\play} \, \payfield_{\play}(\beforelead)
	\quad
	\text{for all $\play\in\players$ and all $\run=\running$}
\end{equation}
where the player-specific coefficients $\present_{\play},\past_{\play} \geq 0$ satisfy $\present_{\play} + \past_{\play} \leq 1$ and represent a mix of past and present payoff information.
\end{subequations}
In this way, depending on the values of $\present_{\play}$ and $\past_{\play}$, we obtain the following prototypical regularized learning methods as special cases of \eqref{eq:FTRL+}:
\begin{subequations}
\begin{enumerate}
[left=\parindent,label={\itshape \alph*\upshape)}]
\item
\define{\acs{FTRL}:}
if $\present_{\play} = \past_{\play} = 0$ for all $\play\in\players$, players essentially forego any look-ahead efforts, so we get
\begin{alignat}{2}
\curr[\signal]
	&= 0
	&\quad
	&\text{for all $\run=\running$}
\intertext{%
In turn, this gives $\lead = \curr$, \ie \eqref{eq:FTRL+} regresses to \eqref{eq:FTRL}.
\item
\define{\Acl{ExtraFTRL}:}
if $\present_{\play}=1$ and $\past_{\play}=0$ for all $\play\in\players$, we have }
\curr[\signal]
	&= \payfield(\curr)
	&\quad
	&\text{for all $\run=\running$}
\intertext{%
\ie players use payoff information from their current state to generate the leading state $\lead$.
This update structure requires two payoff queries per iteration and its origins can be traced back to the work of \citet{Kor76}.
Specifically, depending on the choice of $\hreg_{\play}$, it is essentially equivalent to the mirror-prox \cite{Nem04} and dual extrapolation \cite{Nes07} algorithms, it contains as a special case the forward-looking algorithm of \cite{MLZF+19,pmlr-v151-fasoulakis22a}, etc.
\item
{\define{\Acl{OptFTRL}:}
if $\present_{\play}=0$ and $\past_{\play}=1$ for all $\play\in\players$, we have}}
\curr[\signal]
	&= \payfield(\beforelead)
	&\quad
	&\text{for all $\run=\running$}
\end{alignat}
\ie players reuse the latest available payoff information instead of making a fresh query at $\curr$ (so the algorithm only requires one payoff query per iteration).
In this way, \eqref{eq:FTRL+} recovers the optimistic algorithms of \cite{CYLM+12,RS13-NIPS,SALS15,HIMM19}, the \acs{OMW} update scheme of \cite{SALS15,DP19} when $\mirror = \logit$, etc.
\end{enumerate}
\end{subequations}

Clearly, the list above is not exhaustive:
many other configurations are possible, \eg with different players using different parameter settings for $\present_{\play}$ and $\past_{\play}$, depending on the information they have at hand and any other individual considerations.
To avoid needlessly complicating the analysis, our only standing assumption will be that
$\present_{\play} + \past_{\play} > 0$ for all $\play\in\players$ (since, otherwise, the benefits of the extrapolation step would vanish).
In particular, by rescaling the players' learning rates if needed, we will
normalize $\present_{\play}$ and $\past_{\play}$ to $\present_{\play} + \past_{\play} = 1$,
leading to the convex model
\begin{equation}
\label{eq:signal-curr}
\curri[\signal]
	= \coef_{\play} \, \payfield_{\play}(\curr) + (1-\coef_{\play}) \, \payfield_{\play}(\beforelead)
\end{equation}
for some arbitrarily chosen ensemble of player-specific \define{extrapolation coefficients} $\coef_{\play} \in [0,1]$, $\play\in\players$.

\begin{remark*}
To simplify the presentation of our results, we will assume throughout the rest of our paper that \eqref{eq:FTRL+} is initialized with $\init[\statealt] = \statealt_{1/2} = 0$.
\end{remark*}

\subsection{Analysis \& results}

With all this in hand, we are finally in a position to state our main results for \eqref{eq:FTRL+} in harmonic games.
We begin by showing that \eqref{eq:FTRL+} achieves order-optimal regret:

\begin{restatable}{theorem}{regret}
\label{thm:regret}
Suppose that each player in a harmonic game $\fingame$ is following \eqref{eq:FTRL+} with learning rate $\learn_{\play} \leq \mass_{\play} \hstr_{\play} \bracks{2(\nPlayers+2) \max_{\playalt} \mass_{\playalt} \lips_{\playalt}}^{-1}$
and
payoff models as per \eqref{eq:signal-lead} and \eqref{eq:signal-curr}.
Then the individual regret of each player $\play\in\players$ is bounded as
\begin{equation}
\label{eq:reg-bound}
\reg_{\play}(\nRuns)
	\defeq \max_{\bench_{\play}\in\strats_{\play}}
		\sum_{\run=\start}^{\nRuns}
		\bracks{\pay_{\play}(\bench_{\play};\state_{-\play,\run}) - \pay_{\play}(\curr)}
	\leq \frac{\hrange_{\play}}{\learn_{\play}}
		+ \frac{2 \lips_{\play}}{\nPlayers + 2} \sum_{\playalt\in\players} \frac{\hrange_{\playalt}}{\learn_{\playalt} \lips_{\playalt}}
\end{equation}
where
$\hrange_{\play} = \max\hreg_{\play} - \min\hreg_{\play}$,
and
$\lips_{\play}$ is the Lipschitz modulus of $\payfield_{\play}$.
\end{restatable}

Even though \cref{thm:regret} invites a natural comparison with the constant regret bound of \cref{thm:reg-cont},
the continuous- and discrete-time settings are fundamentally different,
so any conclusions drawn from such a comparison would be specious.
Indeed, constant regret guarantees in the spirit of \eqref{eq:reg-bound} are particularly rare in the context of discrete-time algorithms, and as far as we are aware, similar bounds have only been established for optimistic methods in variationally stable and \aclp{2ZSG} \citep{HAM21};
other than that \textendash\ and always to the best of our knowledge \textendash\ the tightest regret bounds available for general games (finite or convex) seem to be (poly)logarithmic \cite{daskalakis2021nearoptimal,farina2022nearoptimal}.
In this regard, just like the recurrence result of \cref{thm:recurrence}, the $\bigoh(1)$ regret bound of \cref{thm:regret} represents a significant extension of existing results on zero-sum games (and polylogarithmic regret in general games), and suggests that, from a learning viewpoint, harmonic games are the most natural generalization of \aclp{2ZSG} to a general $\nPlayers$-player context.
We defer the proof of \cref{thm:regret} to \cref{app:discrete}.


\begin{figure*}[tbp]
\centering
\footnotesize
\includegraphics[height=32ex]{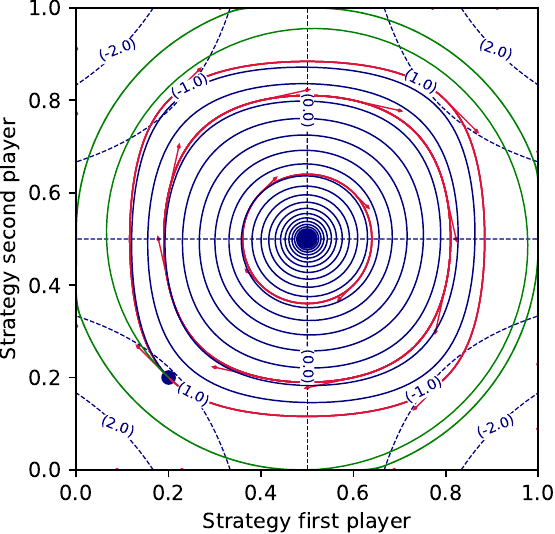}
\hfill
\includegraphics[height=32ex]{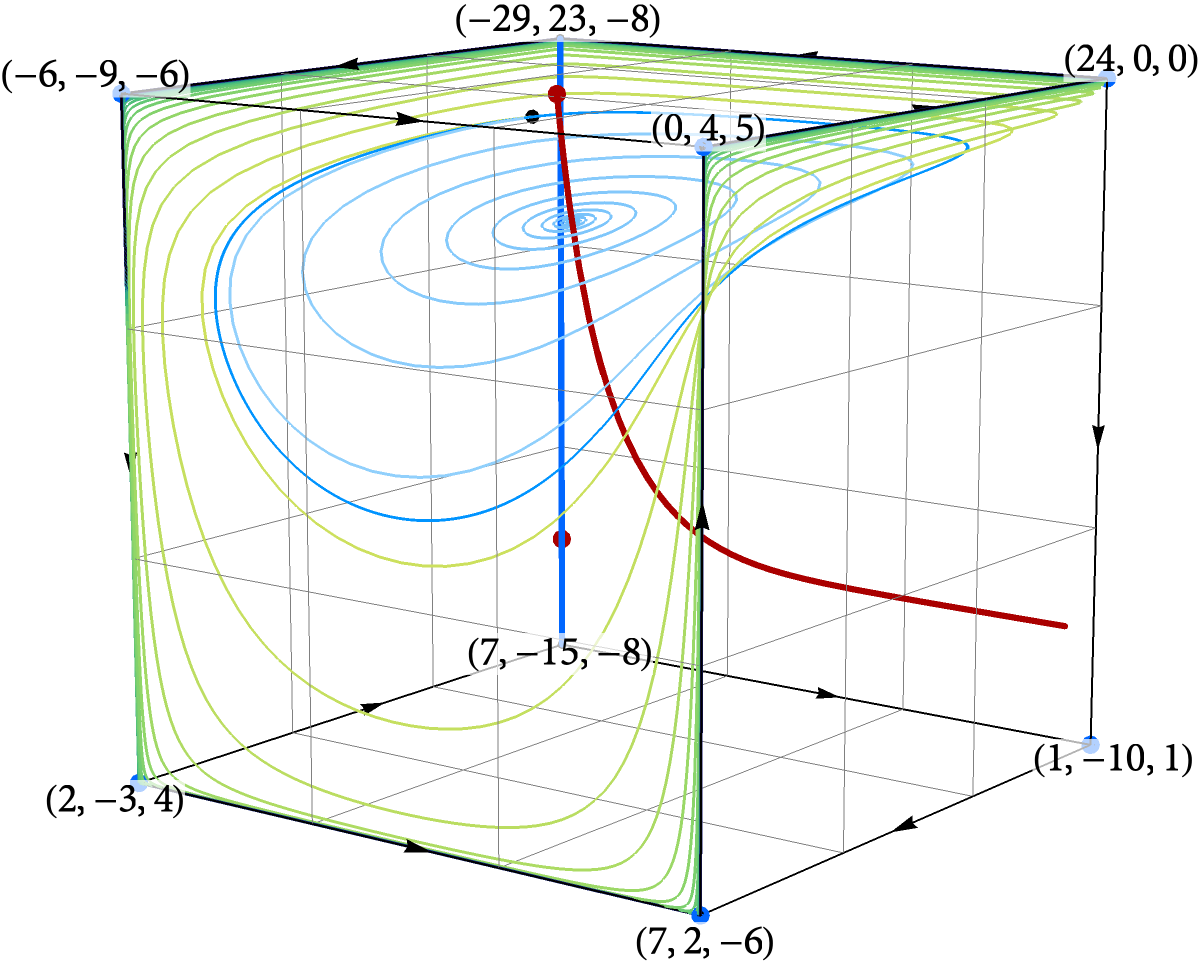}
\hfill
\includegraphics[height=32ex]{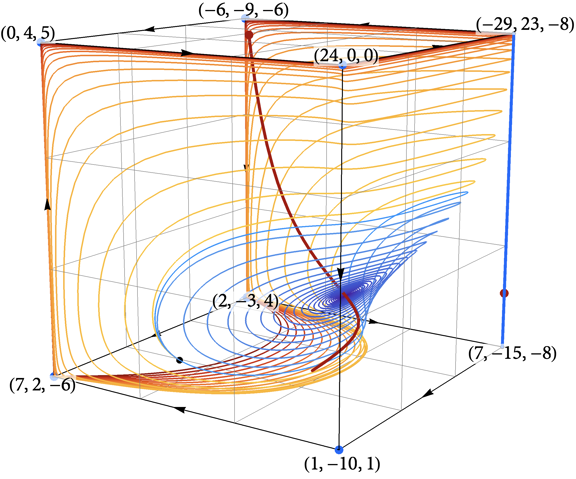}
\caption{The evolution of vanilla \vs extrapolated \ac{FTRL} schemes in harmonic games.
In the left figure, we consider the game of Matching Pennies (blue: \ac{FTRL+}; green: \ac{FTRL}; red: continuous time \ac{FTRL});
in the center and to the right, two different orbits in a $2\times2\times2$ harmonic game from two different viewpoints (blue: \ac{FTRL+}; green/orange:\ac{FTRL}; payoff profiles on vertices).
In all cases, we ran the optimistic variant of \ac{FTRL+} ($\coef_{\play}=0$ for all players), and we see that the trajectories of \eqref{eq:FTRL} diverge away from equilibrium and the trajectories of \eqref{eq:FTRL-cont} are recurrent (actually, periodic), whereas \eqref{eq:FTRL+} converges.
We also see the highly non-convex structure of harmonic games as evidence by their equilibrium set (thick red line in center and right subfigures).}
\label{fig:portraits}
\end{figure*}


As an immediate corollary of the above, we conclude that, under \eqref{eq:FTRL+}, the empirical frequencies of play $\corr_{\pure,\run} \defeq (1/\run) \sum_{\runalt=\start}^{\run} \strat_{\pure,\runalt}$, $\pure\in\pures$, converge to the game's set of \ac{CCE} at a rate of $\bigoh(1/\run)$.
This rate is, again, optimal, but as we discussed in \cref{sec:dynamics}, it offers little information in games where the marginalization of \acp{CCE} does not lead to \acl{NE} \textendash\ and, in general $\nPlayers$-player harmonic games, there is little hope that it would.
In addition, even when the marginalization of \acp{CCE} is Nash, the actual trajectory of play may \textendash\ and, in fact, often \emph{does} \textendash\ behave very differently from the time-averaged frequency of play.

Despite these hurdles, we show below that \eqref{eq:FTRL+} \emph{does} converge to \acl{NE}.
To state this result formally, we will focus on the case where each player's regularizer is \define{smooth} in the sense that
\begin{equation}
\label{eq:hsmooth}
\text{$\hreg_{\play}(\strat_{\play} + t (\stratalt_{\play} - \strat_{\play}))$ is continuously differentiable at $t=0$}
\end{equation}
for all $\strat_{\play}\in\im\mirror_{\play}$ and all $\stratalt_{\play}\in\strats_{\play}$.%
\footnote{The restriction to $\im\mirror_{\play}$ is technical in nature and is related to the subdifferentiability of $\hreg_{\play}$, \cf \cref{app:mirror}.}
Our prototypical examples \textendash\ the entropic and Euclidean regularizers \textendash\ both satisfy this mild requirement, as do all regularizers of the form $\hreg_{\play}(\strat_{\play}) = \sum_{\pure_{\play}\in\pures_{\play}} \theta_{\play}(\strat_{\play\pure_{\play}})$ for some smooth convex function $\theta_{\play}\from[0,1]\to\R$.
We then have the following convergence result:

\begin{restatable}{theorem}{convergence}
\label{thm:convergence}
Suppose that each player in a harmonic game $\fingame$ follows \eqref{eq:FTRL+} with learning rate $\learn_{\play} \leq \mass_{\play} \hstr_{\play} \bracks{2(\nPlayers+2) \max_{\playalt} \mass_{\playalt} \lips_{\playalt}}^{-1}$
and
payoff models as per \eqref{eq:signal-lead} and \eqref{eq:signal-curr}.
Then $\curr$ converges to the set of \aclp{NE} of $\fingame$.
\end{restatable}

To the best of our knowledge, \cref{thm:convergence} is the first result of its kind for harmonic games \textendash\ and, in that regard, it is somewhat unexpected.
To be sure, \aclp{2ZSG} with a fully-mixed equilibrium exhibit a comparable pattern:
\ac{FTRL} is Poincaré recurrent in continuous time, its vanilla discretization is unstable, and its optimistic\,/\,forward-looking implementation is convergent.
However, the convex-concave structure of min-max games which enables this analysis is completely absent in harmonic games, so it is less clear what to expect in this case (where even the set of \aclp{NE} is non-convex, \cf \cref{fig:portraits}).
By this token, the convergence of \eqref{eq:FTRL+} in harmonic games is a property that one could optimistically hope for, but not one that can be taken for granted.

From a technical standpoint, the proof of \cref{thm:regret,thm:convergence} involves two concurrent challenges:
\begin{enumerate}
\item
Deriving a Lyapunov function with a ``sufficient descent'' property for all harmonic games.
\item
Providing an integrated analysis for all possible update structures in \eqref{eq:FTRL+}.
\end{enumerate}
With regard to the first point, our analysis hinges on the ``energy function''
\begin{align}
\label{eq:energy}
\energy(\base,\score)
	&= \sum_{\play\in\players} \frac{\mass_{\play}}{\learn_{\play}} \fench_{\play}(\base_{\play},\score_{\play})
	\qquad
	\base\in\strats,
	\score\in\scores,
\end{align}
In the above, $\base\in\strats$ is a benchmark strategy profile acting as a ``reference point'' for the analysis while
\begin{equation}
\label{eq:Fenchi-max}
\fench_{\play}(\base_{\play},\score_{\play})
	= \max_{\strat_{\play}\in\strats_{\play}} \setof{\braket{\score_{\play}}{\strat_{\play}} - \hreg_{\play}(\strat_{\play})}
		- \bracks{ \braket{\score_{\play}}{\base_{\play}} - \hreg_{\play}(\base_{\play})}
\end{equation}
denotes the \emph{Fenchel coupling} associated to the regularizer $\hreg_{\play}$ of player $\play\in\players$, and represents a ``primal-dual'' measure of divergence between $\base_{\play}\in\strats_{\play}$ and $\score_{\play}\in\scores_{\play}$ (for an in-depth discussion, see \cref{app:mirror,app:discrete}).
Then, letting $\curr[\energy] = \energy(\base,\curr[\score])$, the heavy lifting for our analysis is provided by the ``template inequality''
\begin{align}
\label{eq:template-main}
\next[\energy]
	\leq \curr[\energy]
		&+ \sum_{\play\in\players} \mass_{\play}
			\braket{\payv_{\play}(\lead)}{\leadi - \base_{\play}}
	\notag\\
		&+ \sum_{\play\in\players} \mass_{\play}
			\braket{\payv_{\play}(\lead) - \payfield_{\play}(\curr)}{\nexti - \leadi}
	\notag\\
		&+ \sum_{\play\in\players} \mass_{\play} (1 - \coef_{\play})
			\braket{\payfield_{\play}(\curr) - \payfield_{\play}(\beforelead)}{\nexti - \leadi}
	\notag\\
		&- \sum_{\play\in\players} \frac{\mass_{\play}\hstr_{\play}}{\learn_{\play}}
			\bracks*{\norm{\nexti - \leadi}^{2} + \norm{\leadi - \curri}^{2}}
	\eqdot
\end{align}

A first important consequence of \eqref{eq:template-main} is that the sequences $\curr[A] = \norm{\next - \lead}^{2}$ and $\curr[B] = \norm{\lead - \curr}^{2}$ are both summable:
this requires a repeated use of the Fenchel-Young inequality, and an instantiation of $\base$ to the strategic center $\heq$ of $\fingame$;
we detail the relevant arguments in \cref{app:harmonic,app:discrete}.
Then, by establishing a similar template inequality for \emph{each} player $\play\in\players$, we are able to bound the players' individual regret by the same upper bound that we derived for $\sum_{\run}\curr[A]$ and $\sum_{\run}\curr[B]$, and which is (up to certain secondary factors) the bound \eqref{eq:reg-bound}.

For the convergence to \acl{NE}, the summability argument above also plays a crucial role.
First, by a standard result on numerical sequences, the summability of $\curr[A]$ and $\curr[B]$ coupled with the template inequality \eqref{eq:template-main} implies that the energy $\curr[\energy]$ of the algorithm relative to the game's strategic center converges to some limit value $\energy_{\infty}$.
In turn, this implies that the score sequence $\curr[\statealt]$ is bounded up to a multiple of the vector $(1,\dotsc,1)$, which corresponds to a constant payoff shift in the underlying game.
Then, by focusing on convergent subsequences of $\curr[\statealt]$ and the optimality condition resulting from the definition of $\mirror$, we are able to show that any limit point of $\payv(\curr)$ satisfies the variational characterization \eqref{eq:VI} of \aclp{NE}, from which our claim follows.

\section{Concluding remarks}
\label{sec:discussion}

Our results suggest that the long-run behavior of no-regret algorithms and dynamics in harmonic games is a very rich topic, and one which opens the door to an entirely new class of games where positive convergence results can be obtained.
We find this particularly appealing, not only because harmonic games comprise the strategic complement of potential games, but also because they go beyond standard problems with a convex structure \textendash\ for instance, even their equilibrium set is non-convex.
As such, the fact that it is possible to obtain optimal regret guarantees and positive equilibrium convergence results in this setting is very promising for future work on the topic.

In terms of open questions, it would be important to examine the rate of convergence of \eqref{eq:FTRL+} to equilibrium.
Even though \eqref{eq:FTRL+} has order-optimal regret bounds, this only helps in establishing a convergence rate to the game's set of \aclp{CCE};
for \aclp{NE}, earlier work by \citet{GPD20} and some more recent results by \citet{Cai2022TightLC} and \citet{gorbunov2022lastiterate} have shed some light on the convergence of constrained Euclidean optimistic methods, but the technology therein does not extend to non-monotone, non-Euclidean problems.
Inspired by \citet{WLZL21}, we conjecture that the convergence rate of \eqref{eq:FTRL+} in harmonic games is linear:
this is based on the observation that any harmonic game admits a fully-mixed \acl{NE}, and the weighted sum in the definition of a harmonic game looks formally similar to the condition needed to establish metric subregularity in \cite{WLZL21};
however, a proof would likely require different techniques.

Another important research direction has to do with the information available to the players.
A first open question here concerns the case where players do not have access to full information on their mixed payoff vectors, but can only observe their pure payoffs \textendash\ either in a ``what if'', counterfactual manner, or in the form of bandit, payoff-based feedback.
In a similar manner, the algorithms presented here are not adaptive, in the sense that the players' step-size policy has to satisfy a certain bound that depends on correctly estimating some of the game's parameters.
Obtaining an adaptive version of \eqref{eq:FTRL+} which, in the spirit of \citet{RS13-NIPS} and \citet{HAM21,HACM22,HIMM22}, remains convergent and attains order-optimal regret in both adversarial and game-theoretic settings without any pre-play tuning is also an ambitious question for future research.

\section*{Acknowledgments}
\begingroup
\small
%
%
This research was supported in part by
the French National Research Agency (ANR) in the framework of
the PEPR IA FOUNDRY project (ANR-23-PEIA-0003),
the ``Investissements d’avenir program'' (ANR-15-IDEX-02), the LabEx PERSYVAL (ANR-11-LABX-0025-01), MIAI@Grenoble Alpes (ANR-19-P3IA-0003),
the project IRGA2024-SPICE-G7H-IRG24E90,
and
NSF grant CCF2212233.
PM is also with the Archimedes Research Unit/Athena RC \& NKUA and acknowledges financial support by project MIS 5154714 of the National Recovery and Resilience Plan Greece 2.0 funded by the European Union under the NextGenerationEU Program.
\endgroup

\appendix
\crefalias{section}{appendix}
\crefalias{subsection}{appendix}

\numberwithin{lemma}{section}	
\numberwithin{proposition}{section}	
\numberwithin{equation}{section}	

\section{Harmonic Games}
\label{app:harmonic}


The class of \acdefp{UHG} introduced by \citet{CMOP11} provides a game-theoretic framework for modeling strategic situations with conflicting, anti-aligned interests.\footnote{We include here the word ``uniform'' to distinguish the class of harmonic games introduced by \citet{CMOP11} from the more general class of \aclp{HG} considered in this work, \cf \cref{def:harmonic}.} Broadly speaking, the characterizing property of \aclp{UHG} is the following:  for any player considering a deviation towards a specific pure strategy profile, there exist other players who are motivated to deviate \textit{away} from that profile. 

Given a finite game $\fingame = \fingamefull$, this is formalized by the condition that, for all $\pure\in\pures$,
\begin{equation}
\label{eq:harmonic-uniform}
\sum_{\play \in \players} \sum_{\purealt_{\play} \in \pures_{\play}}
\big[ \dev{\play}{\pure}{\purealt} \big]
= 0 \eqdot
\end{equation}
From a strategic viewpoint, \aclp{UHG} complement potential games: \citet{CMOP11} showed that any finite game can be uniquely decomposed into the sum of a potential game and a \acl{UHG}, up to linear transformations of the payoff functions that do not change the strategic structure of the game.%

Since their introduction, harmonic games have generated a substantial body of literature; for a brief survey, we refer the reader to \citep{LMP24}.

\subsection{Harmonic games, measures and comeasures}
The class of \aclp{UHG} exhibits intriguing, yet restrictive, properties. Notably, a \ac{UHG} always admits the uniformly mixed strategy as a \ac{NE}, and it generally possesses a \textit{continuum} of \aclp{NE} \citep{CMOP11}.  Additionally, the framework of \acp{UHG} and the decomposition proposed by \citet{CMOP11} are incompatible with common game-theoretical transformations, such as the duplication of strategies or rescaling of payoffs \citep{APSV22}. To address the above limitations, \citet{APSV22} extended the definition of harmonic games by the introduction of two parameters: a \define{measure}, that is a positive weight each player assigns to each of \textit{their own} strategy; and a \define{comeasure}, that is a positive weight each player assigns to each of the \textit{other players'} action profiles.
\begin{definitionApp}
\label{def:meas-comeas}
Let $\fingamefull$ be a finite game. A \define{player measure} $\meas_{\play}$ is a function $\meas_{\play}\from\pures_{\play}\to\meascod$; a \define{player co-measure} $\comeas_{\play}$ is a function $\comeas_{\play}\from\pures_{\others}\to\meascod$. Correspondingly, a collection $\meas = \setof{\meas_{\play}}_{\play\in\players}$ (resp. $\comeas = \setof{\comeas_{\play}}_{\play\in\players}$) of player measures (resp. comeasures) is called \define{game measure} (resp. \define{game comeasure}).
If $\meas_{\play}$ is a player measure, we will write $\lnorm{\meas_{\play}} \defeq \sum_{\pure_{\play}}\meas_{\findex}$.
Finally,
a \define{probability measure} is a game measure $\meas$ such that $\lnorm{\meas_{\play}} = 1$ for all $\play \in \players$;
a \define{uniform measure} is a game measure $\meas$ such that $\meas_{\findex} = 1$ for all $\play\in\players, \pure_{\play} \in \pures_{\play}$; and
a \define{uniform comeasure} is a game comeasure $\comeas$ such that $\comeas_{\findexco} = 1$ for all $\play\in\players, \pure_{\others} \in \pures_{\others}$.
\end{definitionApp}
With these notions in place, \citet{APSV22} define a finite game $\fingame$ to be \define{$(\meas,\comeas)$-harmonic} if there exist a game measure $\meas$ and a game comeasure $\comeas$ such that, for all $\pure \in \pures$,
\begin{equation}
\label{eq:harmonic-measure-comeasure}
\sum_{\play} \sum_{\purealt_{\play}}
\meas_{\findexalt} \comeas_{\findexco}
\big[ \dev{\play}{\pure}{\purealt} \big]
= 0 \eqdot
\end{equation}
In this work, we focus solely on harmonic games with \textit{uniform comeasure}. As discussed after \cref{def:harmonic} in the main body of the article, this choice comes without loss of generality: the game comeasure in \cref{eq:harmonic-measure-comeasure} can be absorbed by a payoff rescaling to give a game that is still harmonic, and \define{preference equivalent} to the original game \textendash{} in a sense that we make precise in the next section.

\subsection{Preference equivalence between harmonic games}
\label{app:preference-equivalence}
The strategic structure of a game is preserved under monotonic transformations of the utility functions, since the set of pure Nash equilibria of a game is an ordinal object \textendash{} it depends only on the signs of unilateral payoff differences, and not on their absolute values. For this reason, two games $\fingame(\players,\pures,\pay)$ and $\alt\fingame(\players,\pures,\alt\pay)$ are called \acdef{PE} if for all $\pure,\purealt\in\pures$ and all $\play\in\players$, we have
\begin{equation}
\label{eq:ordinal-strat-equiv}
\sgn{ \big[ \alt\pay_{\play}(\purealt_{\play};\pure_{-\play}) - \alt\pay_{\play}(\pure_{\play};\pure_{-\play})\big]}
	= \sgn{ \big[\pay_{\play}(\purealt_{\play};\pure_{-\play}) - \pay_{\play}(\pure_{\play};\pure_{-\play})\big]}.
\end{equation}

Two games are \acdef{SE} \textendash\ and we write $\fingame \sim \alt\fingame$ \textendash\ if they have the same unilateral payoff differences, that is if
\begin{equation}
\label{eq:strat-equiv}
\alt\pay_{\play}(\purealt_{\play};\pure_{-\play}) - \alt\pay_{\play}(\pure_{\play};\pure_{-\play})
	= \pay_{\play}(\purealt_{\play};\pure_{-\play}) - \pay_{\play}(\pure_{\play};\pure_{-\play})
\end{equation}
for all $\pure,\purealt\in\pures$ and all $\play\in\players$; \acl{SE} games are clearly \acl{PE}.

\begin{lemmaApp}
\label{lemma:harmonic-preference-equivalence}
Let $\fingame_{\meas, \comeas} = \fingame_{\meas, \comeas}(\players, \pures, \pay)$ be a harmonic game in the sense of \cref{eq:harmonic-measure-comeasure}. Then the game $(\players, \pures, \payalt)$ with
$\payalt_{\play}(\ppure)
= \comeas_{\findexco} \pay_{\play}(\ppure)$ is \acl{PE} to the game $\fingame_{\meas, \comeas}$, and it is harmonic in the sense of \cref{eq:harmonic-measure-comeasure} with measure $\meas$ and uniform comeasure.
\end{lemmaApp}
\begin{proof}
Let $\payaltalt_{\play}(\ppure) = \meas_{\findex}\comeas_{\findexco} \pay_{\play}(\ppure)$. Then replacing above, for all $\pure \in \pures$,
\[
\begin{split}
0 
& = \sum_{\play\in\players} \sum_{\purealt_{\play}\in\pures_{\play}}
\meas_{\findexalt} 
\left[\frac{\payaltalt_{\play}(\ppure)}{\meas_{\findex}} - \frac{\payaltalt_{\play}(\purealt_{\play}; \pure_{\others})}{\meas_{\findexalt}} \right] \eqdot
\end{split}
\]
Let
$\payalt_{\play}(\ppure)
= \frac{\payaltalt_{\play}(\ppure)}{\meas_{\findex}}
= \comeas_{\findexco} \pay_{\play}(\ppure)$. The game $\payalt$ is \acl{PE} to $\pay$, and
\begin{equation}
0 = 
\sum_{\play\in\players} \sum_{\purealt_{\play}\in\pures_{\play}}
\meas_{\findexalt} 
\big[ \payalt_{\play}(\pure_{\play}; \pure_{\others}) - \payalt_{\play}(\purealt_{\play}; \pure_{\others})   \big]
\end{equation}
for all $\pure \in \pures$, so $\payalt$ is harmonic in the sense of \ref{eq:harmonic-measure-comeasure} with measure $\meas$ and uniform comeasure.
\end{proof}
In the proof above we perform the intermediate step $\pay \to \payaltalt$ rather than defining directly $\pay \to \payalt$ to stress the difference between rescaling the payoffs of a game by a game measure $\meas$ and by a game comeasure $\comeas$. The game with payoffs $\payalt = \comeas \pay$ (the meaning of this notation made precise in the proof above) is \acl{PE} to the game with payoffs $\pay$, \ie rescaling the payoffs by a comeasure does not change the strategic structure of the game. On the other hand, the game with payoffs $\payaltalt = \meas  \payalt$ \textendash{} again, the meaning made precise in the proof \textendash{} is \textit{not} \ac{PE} to the game with payoffs $\payalt$: rescaling the payoffs by a measure can change the preferences of the players, and leads to a game with intrinsically different strategic structure.

\cref{lemma:harmonic-preference-equivalence} motivates our choice to focus in this work on harmonic games with arbitrary measures and uniform comeasures, and to adopt \eqref{eq:harmonic} from \cref{def:harmonic} to characterize \aclp{HG}: a \acdef{HG} $\HG = \HGfull$ is a finite game $(\players, \pures, \pay)$ with a game measure $\meas$ such that \eqref{eq:harmonic} holds, \ie
\(
\insum_{\play \in \players}
	\insum_{\purealt_{\play} \in \pures_{\play}}
	\meas_{\play\purealt_{\play}}
	\bracks{ \pay_{\play}(\pure_{\play};\pure_{-\play}) - \pay_{\play}(\purealt_{\play};\pure_{-\play}) }
= 0
\)
for all $\pure\in\pures$.

\subsection{Mixed characterization of harmonic games}

The defining property \eqref{eq:harmonic} allows for an equivalent characterization of harmonic games in terms of their mixed payoffs:
\begin{lemmaApp}
\label{lemma:harmonic-mixed}
A finite game $\fingame = \fingamefull$ is harmonic with measure $\meas$ if and only if
\begin{equation}
\label{eq:harmonic-mixed}
\tag{HG-mixed}
\sum_{\play \in \players} \lnorm{\meas_{\play}} \, 
\Big\langle
\payfield_{\play}\of\strat , \strat_{\play} - \frac{\meas_{\play}}{\lnorm{\meas_{\play}}}
\Big\rangle
= 0
\quad \text{for all } \strat \in \strats \eqdot
\end{equation}
\end{lemmaApp}
\begin{proof}
Given a finite game $\fingame = \fingamefull$ and a game measure $\meas$, let $F_{\play}\from\pures\to\R$ be defined by
$F_{\play}(\pure) = \sum_{\purealt_{\play}\in\pures_{\play}} 
\meas_{\findexalt}
\bracks{\dev{\play}{\pure}{\purealt}}$.
 By definition, $\fingame$ is a $\meas$-\acl{HG} if and only if $F(\pure) \defeq \sum_{\play\in\players} F_{\play}\of\pure = 0$ for all $\pure \in \pures$. Denote (with slight abuse of notation) by $F\from \strats \to \R$ the multilinear extension of $F\from\pures\to\R$, \ie $F(\strat) = \sum_{\pure}\strat_{\pure}F(\pure)$, with $\strat_{\pure} \defeq \prod_{\play} \strat_{\play \pure_{\play}}$. Now, $F(\pure) = 0$ for all $\pure \in \pures$ if and only if $F(\strat) = 0$ for all $\strat \in \strats$, which is the case if and only if
\begin{flalign}
0 & = F(\strat) = \insum_{\pure} \strat_{\pure} \insum_{\play} F_{\play}(\pure)
	\notag\\
	&= \insum_{\play} \insum_{\pure_{\play}} \insum_{\pure_{\others}} \strat_{\findex} \strat_{\findexothers} \insum_{\purealt_{\play}} 
\meas_{\findexalt}
\bracks{\dev{\play}{\pure}{\purealt}}
	\notag\\
	&= \insum_{\play} \insum_{\purealt_{\play}} \meas_{\findexalt} \bracks{ \dev{\play}{\strat}{\purealt} }
	\notag\\
	&= \insum_{\play} \big[\lnorm{\meas_{\play}} \dualp{\payfield_{\play}\of\strat}{\strat_{\play}} - \dualp{\payfield_{\play}(\strat)}{\meas_{\play}}\big] \quad \text{for all } \strat \in \strats,
\end{flalign}
from which our claim follows by factoring out the terms involving $\lnorm{\meas_{\play}}$.
\end{proof}
\begin{remark*}
The first equality in the second line holds true for harmonic games with \textit{uniform} comeasure $\comeas_{\findexco} = 1$, since $\comeas_{\findexco} \neq 1$ terms would couple with the corresponding $\strat_{\findexothers}$ terms in the sum.
\end{remark*}
The above result can be  reformulated as follows:
\begin{propositionApp}
\label{prop:harmonic-center}
A finite game $\fingame = \fingamefull$ is harmonic if and only if it admits a \define{strategic center} $\stratcenter$, \viz if there exist
\begin{enumerate*}[(\itshape i\hspace*{1pt}\upshape)]
\item a vector $\wei \in \R^{\nPlayers}_{++}$ and
\item a fully mixed strategy $\heq \in \strats$ such that
\end{enumerate*}
\begin{equation}
\label{eq:harmonic-center}
\tag{HG-center}
\sum_{\play \in \players} \wei_{\play} \, 
\dualp{\payfield_{\play}\of\strat }{ \strat_{\play} - \heq_{\play}}
= 0
\quad \text{for all } \strat \in \strats \eqdot
\end{equation}
\end{propositionApp}
This expression is intriguing: it suggest that a game is harmonic precisely if there exists a fully mixed strategy $\heq$ such that, for all $\strat \in \strats$,
the payoff vector $\payfield\of\strat$ is perpendicular (with respect to a $\wei$-weighted inner product) to $\strat - \heq$; \cf \cref{ex:harmonic-game} and \cref{fig:harmonic-game}. The  striking dynamical consequences of this ``circular'' strategic structure \textendash{} hinted at in \cref{fig:harmonic-game}, showing a \textit{periodic} orbit of FTRL in continuous time \textendash{} are captured precisely by  \cref{thm:recurrence} in the main text.
\begin{proof}[Proof of \cref{prop:harmonic-center}]
Let $\HG = \HGfull$ be harmonic; then by \cref{lemma:harmonic-mixed} that there exist a strategic center $\stratcenter$ given by $\wei_{\play} \defeq \lnorm{\meas_{\play}}$ and $\heq_{\play} \defeq \meas_{\play} / \lnorm{\meas_{\play}}$ with $\play \in \players$. Conversely let $\fingame = \fingamefull$ admit a strategic center $(\wei, \heq)$; then by the same argument $\fingame$ is harmonic with $\meas_{\play}  \defeq \wei_{\play}\heq_{\play}$ for all $\play \in \players$.
\end{proof}
An immediate corollary is the following:
\begin{corollaryApp}
If a finite game $\fingame$ admits a strategic center $\stratcenter$, then $\heq$ is a Nash equilibrium.
\end{corollaryApp}
\begin{proof}
By \cref{prop:harmonic-center} if $\fingame$ admits a strategic center $\stratcenter$ then it is $\meas$-harmonic with $\meas_{\play} = \wei_{\play} \heq_{\play}$ for all $\play\in\players$; and $(\meas_{\play} / \lnorm{\meas_{\play}})_{\play\in\players}$ is always a \ac{NE} for $\meas$-\aclp{HG} \citep[Theorem 1]{APSV22}.
\end{proof}
\begin{remark*}
The converse does not hold: a fully mixed \acl{NE} is not necessarily a strategic center. If it were, a game would be harmonic precisely if it admitted a fully mixed \ac{NE}, which is not the case \textendash{} think for example of coordination or anti-coordination games, that admit a fully mixed \acl{NE} and are not harmonic.

\end{remark*}

\begin{example}[A harmonic game: Siege]
\newcommand{\cost}{c} 
\newcommand{\ats}{a_{s}} 
\newcommand{\atf}{a_{f}} 
\newcommand{\dfs}{d_{s}} 
\newcommand{\dff}{d_{f}} 
\label{ex:harmonic-game}
\begin{figure*}[tbp]
\centering
\includegraphics[height=19em]{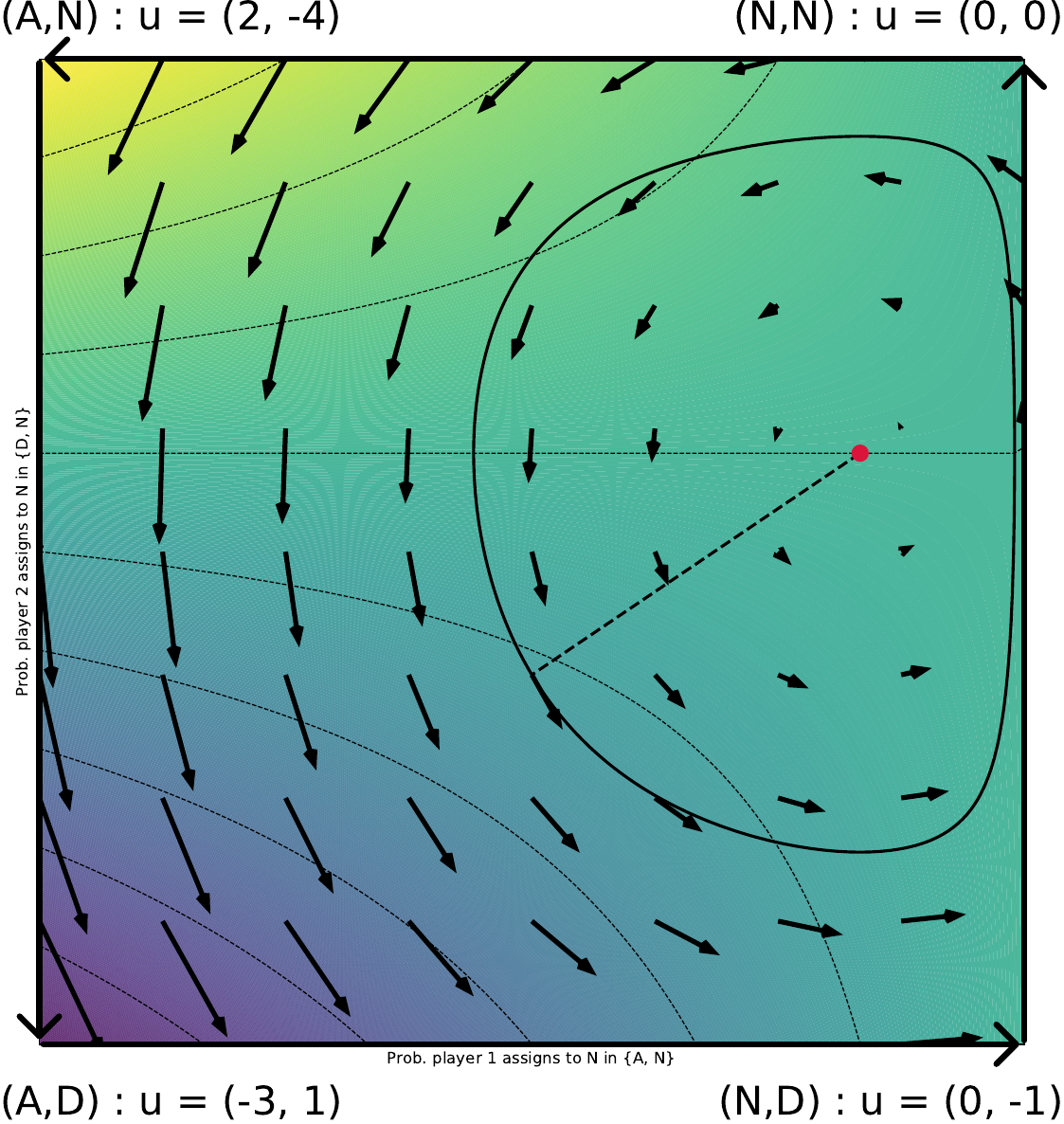}
\caption{Representation of the harmonic payoff structure for the game in \cref{ex:harmonic-game}. Each payoff vector $\payfield\of\strat$ (black arrows) is perpendicular (with respect to a weighted inner product) to the vector $\strat - \heq$ (dotted segment) between the evaluation point $\strat$ of the payoff field and the fully mixed Nash equilibrium $\heq$ (red point). As a consequence every orbit of FTRL in continuous time (such as the one represented by the black curve) is \textit{Poincaré recurrent} (in this low-dimensional example, even periodic), as detailed in \cref{thm:recurrence} in the main text.
Color shading and dotted lines represents player $1$'s utility level sets, with brighter regions indicating higher payoffs.
}
\label{fig:harmonic-game}
\end{figure*}

Consider the following  $2 \times 2$ game: an army (the row player) must choose between Attacking a fortress (pure strategy  $A$ ) and Not attacking (pure strategy  $N$ ). Simultaneously, the fortress (the column player) decides whether to activate its Defenses (pure strategy  $D$ ) or Not (pure strategy  $N$ ). Engaging in either action (the attack or the defense) incurs a preparation cost of $ \cost > 0 $. The army gains
$ \ats > \cost $
if it attacks an undefended fortress, but suffers a loss of $ \atf > 0 $ if it attacks and encounters defenses (the subscripts $s$ and $f$ standing respectively for ``successful'' and ``failed''). Conversely, the fortress benefits by $\dfs > 0 $ if it is defended against an attack, while it incurs a loss of $ \dff > 0 $ if attacked without defenses; defeating the attacking army is worth the preparation cost for the fortress, namely
$\dfs - \cost > -\dff$.
This scenario is captured by the following payoff matrix, specialized on the right to the case
$\cost = 1,
\ats = 3,
\atf = 2,
\dfs = 2,
\dff = 4$:

\begin{equation}
\label{eq:example-harmonic-payoff}
\renewcommand{\arraystretch}{1.2}  
 \begin{array}{c|cc}
	& D & N \\ \hline
A 	& \left( -a_{f} - c, \  d_{s} - c\right) & \left( a_{s} - c, \  -d_{f}\right)\\
N 	& \left( 0, \  -c\right) & \left( 0, \  0\right)
\end{array}
\quad
 \begin{array}{c|cc}
	& D & N \\ \hline
A 	& -3, 1 & 2, -4\\
N 	& 0, -1 & 0,  0
\end{array}
\end{equation}
To determine if the game is harmonic, look for a solution of the linear system
\[
\EqHarmonic \, ,
\]
subject to the constraints $\meas_{\findex} > 0$ for all $\play \in \players, \pure_{\play} \in \pures_{\play}$. For a fixed payoff function $\pay$, this is a system of $\prod_{\playalt \in \players}\nPures_{\playalt}$ linear equations (one for each $\pure \in \pures$) in the $\sum_{\playalt \in \players}\nPures_{\playalt}$ variables
$\left(
	(
		\meas_{\findex}
	)_{\pure_{\play} \in \pures_{\play}}
\right)_{\play\in\players}$, where $\nPures_{\play}$ is the number of pure actions of player $\play\in\players$. With $\pay$ given by \eqref{eq:example-harmonic-payoff} \textendash{} left,
\begin{equation}
\meas = \lambda
\left[
	\left( 
		\frac{c}{a_{f} + c}, 
		\frac{- c + d_{f} + d_{s}}{a_{f} + c}
	\right), 
	\left(
		\frac{a_{s} - c}{a_{f} + c} , 
		1
	\right)
\right]
\end{equation}
is a feasible solution of \eqref{eq:harmonic} for any $\lambda > 0$, so the game is harmonic with a $1$-dimensional set of measures.
The corresponding strategic center $(\mass, \heq)$ with $\wei_{\play} = \sum_{\pure_{\play}} \meas_{\findex}$, $\heq_{\play} = \meas_{\play} / \mass_{\play}$, $\play \in \{1, 2\}$ is
\begin{equation}
\mass = \lambda 
\left(
	\frac{ d_{f} + d_{s}}{a_{f} + c}, 
	\frac{a_{f} + a_{s}}{a_{f} + c}
\right), 
\quad
\heq =
\left[
	\left(
		\frac{c}{d_{f} + d_{s}}, 
		\frac{- c + d_{f} + d_{s}}{d_{f} + d_{s}}
	\right), 
	\left(
		\frac{a_{s} - c}{a_{f} + a_{s}},
		\frac{a_{f} + c}{a_{f} + a_{s}}
	\right)
\right].
\end{equation}
As a sanity check, compute the payoff field and verify that \eqref{eq:harmonic-center} holds true in the specialized case \eqref{eq:example-harmonic-payoff} \textendash{} right. Denoting the mixed strategies of players $1$ and $2$ respectively by $x \in \Delta(\{A,N\})$ and $y \in \Delta(\{D,N\})$, the payoff fields are
$
\payfield_{1}(x, y) = \left(  - 3 y_{D} + 2 y_{N}, 0 \right) , 
\quad
\payfield_{2}(x, y) = \left(  x_{A} - x_{N}, -4x_{A}\right) .
$
Choosing $\lambda = 3$ the strategic center gives weights $\mass = (6, 5)$ and Nash equilibrium
$\heq = \left[ \left(1/6, 5/6\right), \left(2/5, 3/5\right) \right]$. Condition \eqref{eq:harmonic-center} boils down to
$6 	\, \dualp{\payfield_{1}}{x - \heq_{1}}
+ 5 \, \dualp{\payfield_{2}}{y - \heq_{2}} = 0$,
which one readily verifies to hold true by replacing the expressions above and recalling that $x_{A} + x_{N} = 1 = y_{D} + y_{N}$. \cref{fig:harmonic-game} illustrates the situation: each payoff vector $\payfield\of\strat$ (black arrows) is perpendicular (with respect to a weighted inner product) to the vector $\strat - \heq$ (dotted segment) between the evaluation point $\strat$ of the payoff field and the fully mixed Nash equilibrium $\heq$ (red point). 
\endenv
\end{example}

\subsection{Harmonic and zero-sum games}

\citet{CMOP11}'s \aclp{UHG}, defined by \cref{eq:harmonic-uniform}, are precisely the \aclp{HG} with uniform measure, which makes \aclp{UHG} a strict subset of the set of \acp{HG}. Importantly, \acp{HG} include another archetypal class of perfect-competition games: as we show in this section, \acdefp{2ZSG} with an interior \ac{NE} $\eq$ are harmonic with (probability) measure $\meas = \eq$.


To show this, we will need the following definition and lemma:
\begin{definitionApp}[Non-strategic game]
\label{def:non-strategic-game}
A finite normal form game $\fingame = \fingame(\players, \pures, \payns)$ is called \define{non-strategic} if the payoff of each player does not depend on their own choice, \viz if
$
\payns_\play(\pure_\play, \pure_{\others}) =  \payns_\play(\purealt_\play, \pure_{\others})
$
for all $\play \in \players, \pure \in \pures, \purealt_{\play} \in \pures_{\play}$.
\end{definitionApp}
\begin{lemmaApp}
\label{lemma:strat-equiv-iff-difference-non-strategic}
Two finite games $\fingamefull, \alt\fingame(\players,\pures,\alt\pay)$ are strategically equivalent in the sense of \cref{eq:strat-equiv} if and only if their difference is a non-strategic game.
\end{lemmaApp}
\begin{proof}
Let $\fingame - \alt\fingame$ be non-strategic; then $\payns \defeq \alt\pay - \pay$ fulfills the condition of \cref{def:non-strategic-game}, which shows that $\pay$ and $\payalt$ fulfill \cref{eq:strat-equiv}. Conversely let $\fingame$ and $\alt\fingame$ be strategically equivalent; set $\payns \defeq \alt\pay - \pay$ and rearrange the terms in \cref{eq:strat-equiv} to immediately conclude that $\payns$ is a non-strategic game.
\end{proof}
\begin{propositionApp}
Let $\HG = \HGfull$ be a \acl{HG}. If the measure $\meas$ fulfills $\lnorm{\meas_{\play}} = \lnorm{\meas_{\playalt}}$ for all $\play, \playalt \in \players$ then $\HG$ is strategically equivalent to a \acl{ZSG}.
\end{propositionApp}
\begin{proof}
Recall that $\lnorm{\meas_{\play}} \equiv \sum_{\pure_\play}\meas_{\findex}$. Under the assumption $\lnorm{\meas_{\play}} = \lnorm{\meas_{\playalt}}$ for all $\play, \playalt \in \players$, let $\const := \lnorm{\meas_{\play}}$ for any $\play \in \players$. By \eqref{eq:harmonic}, the payoff $\pay$ of $\HG$ in this case fulfills
$
\sum_{\play\in\players} \bracks{ \pay_{\play}\of\pure
- \payns_{\play}\of\pure }
= 0
$
for all $\pure \in \pures$, with $\payns_{\play}\of{\ppure} \defeq \const^{-1} \sum_{\purealt_{\play}} \meas_{\findexalt} \pay_{\play}(\purealt_{\play}, \pure_{\others})$. Set $\payalt_{\play} \defeq \pay_{\play} - \payns_{\play}$. By definition $\payalt$ is a \acl{ZSG}; furthermore, the difference between $\pay_{\play}$ and $\payalt_{\play}$ is non-strategic, since $\payns_{\play}(\ppure)$ does not depend on $\pure_{\play}$. Thus $\pay_{\play}$ and $\payalt_{\play}$ are \acl{SE} by \cref{lemma:strat-equiv-iff-difference-non-strategic}.
\end{proof}
In particular we have the following:
\begin{corollaryApp}
Let $\HG = \HGfull$ be a \acl{HG}. If the measure $\meas$ is a probability measure, then $\HG$ is strategically equivalent to a \acl{ZSG}.
\end{corollaryApp}
The converse holds true only in the case of two-player games:
\begin{propositionApp}
Every \acl{2ZSG} with an interior \acl{NE} $\eq$ is harmonic, with (probability) measure $\meas = \eq$.
\end{propositionApp}
\begin{proof}
Let $\fingame = \fingamefull$ be a \acl{2ZSG} with interior \acl{NE} $\eq$. If we show that
\begin{equation}
\label{eq:2zs-is-harmonic}
\sum_{\play \in \players} \lnorm{\eq_{\play}} \, 
\Big\langle
\payfield_{\play}\of\strat , \strat_{\play} - \frac{\eq_{\play}}{\lnorm{\eq_{\play}}}
\Big\rangle
= 0
\quad \text{for all } \strat \in \strats \, ,
\end{equation}
then we can conclude by \cref{lemma:harmonic-mixed} that $\fingame$ is harmonic with measure $\eq$. \cref{eq:2zs-is-harmonic} holds indeed true: $\lnorm{\eq_{\play}} = 1$ for all $\play \in \players$, and it is well known \cite{MPP18, MLZF+19} that \aclp{2ZSG} with an interior equilibrium $\eq$ fulfill 
$
\sum_{\play \in \players}
\dualp{\payfield_{\play}\of\strat}{\strat_{\play} - \eq_{\play}}
= 0
$ for all $\strat \in \strats$, so we are done.
\end{proof}


Harmonic games thus encompass and substantially generalize two prototypical classes of games with anti-aligned incentives, serving as an ideal complement to the class of potential games. This is made precise in \citep{APSV22}: building on the work of \citet{CMOP11}, \citet{APSV22} showed that, for any choice of game measure $\meas$, every finite game can be uniquely decomposed into the sum of a potential and a $\meas$-harmonic game, up to strategic equivalence.

This establishes harmonic games as the natural complement of potential games from a strategic perspective; \Cref{{thm:recurrence}} in the main text shows that this holds true from a \textit{dynamic} perspective as well.

\section{Basic properties of regularizers and the induced choice maps}
\label{app:mirror}

In this appendix, we collect a number of properties concerning regularizers and the associated choice maps.
To avoid carrying around the player index $\play\in\players$, we state all our results for a generic convex subset $\cvx$ of some real vector space $\vecspace$.
The desired properties for \ac{FTRL} will then be obtained by specializing $\cvx$ to $\strats_{\play}$ or $\strats$ and $\vecspace$ to $\ambienti$ or $\ambient$, depending on the context.

\subsection{Preliminary definitions}

To begin, let $\vecspace$ be a $\vdim$-dimensional normed space with norm $\norm{\cdot}$.
In what follows, we will write
$\dpoints \defeq \dspace$ for the dual space of $\vecspace$,
$\braket{\dpoint}{\point}$ for the canonical pairing between $\point\in\vecspace$ and $\dpoint\in\dspace$,
and
$\dnorm{\dpoint} = \max\setdef{\braket{\dpoint}{\point}}{\norm{\point} \leq 1}$ for the induced dual norm on $\dpoints$.
Following standard conventions in convex analysis, functions will be allowed to take values in the extended real line $\R\cup\{\infty\}$, and if $\obj\from\vecspace\to\R\cup\{\infty\}$ is a convex function on $\vecspace$, we will denote its \define{effective domain} as
\begin{equation}
\label{eq:domfun}
\dom\obj
	\defeq \setdef{\point\in\vecspace}{\obj(\point) < \infty}
	\eqdot
\end{equation}
In addition, assuming $\dom\obj\neq\varnothing$, the \define{subdifferential} of $\obj$ at $\point$ is defined as
\begin{equation}
\label{eq:subdiff}
\subd\obj(\point)
	\defeq \setdef{\dpoint\in\dpoints}{\obj(\pointalt) \geq \obj(\point) + \braket{\dpoint}{\pointalt - \point} \; \text{for all $\pointalt\in\vecspace$}}
\end{equation}
and we denote the \define{domain of subdifferentiability} of $\obj$ as
\begin{equation}
\label{eq:domdiff}
\dom\subd\obj
	= \setdef{\point\in\vecspace}{\subd\obj(\point) \neq \varnothing}
	\eqdot
\end{equation}
Finally, to ease notation, a convex function $\obj\from\cvx\to\R$ will be identified with the extended-real-valued function $\bar\obj\from\vecspace\to\R\cup\{\infty\}$ that agrees with $\obj$ on $\cvx$ and is identically equal to $\infty$ on $\vecspace\setminus\cvx$.

With all this in hand, let $\cvx$ be a closed convex subset of $\vecspace$, and let $\hreg\from\cvx\to\R$ be a $\hstr$-strongly convex \define{regularizer} on $\cvx$, that is,
\begin{equation}
\label{eq:hstr}
\hreg(t\point + (1-t)\pointalt)
	\leq t \hreg(\point)
		+ (1-t) \hreg(\pointalt)
		- \frac{\hstr}{2} t(1-t) \norm{\pointalt - \point}^{2}
	\eqdot
\end{equation}
By standard arguments in convex analysis, this readily implies that
\begin{equation}
\label{eq:hstr-diff}
\hreg(\pointalt)
	\geq \hreg(\point)
		+ \dir\hreg(\point;\pointalt - \point)
		+ \frac{\hstr}{2} \norm{\pointalt - \point}^{2}
	\quad
	\text{for all $\point,\pointalt\in\points$},
\end{equation}
where
\begin{equation}
\dir\hreg(\point;\pointalt-\point)
	= \lim_{\theta\to0^{+}} \bracks{\hreg(\point + \theta(\pointalt-\point)) - \hreg(\point)} / \theta
\end{equation}
denotes the one-sided directional derivative of $\hreg$ at $\point$ along the direction of $\pointalt-\point$.
To proceed, we will need the following basic objects:
\begin{enumerate}
\item
The \define{convex conjugate} $\hconj\from\dpoints\to\R$ of $\hreg$:
\begin{alignat}{2}
\label{eq:conj}
\hconj(\dpoint)
	&= \max_{\point\in\points} \{ \braket{\dpoint}{\point} - \hreg(\point) \}
	&\qquad
	&\text{for all $\dpoint\in\dpoints$}.
\intertext{%
\item
The \define{regularized choice map} \textendash\ or \define{mirror map} \textendash\ $\mirror\from\dpoints\to\points$ induced by $\hreg$:%
}
\label{eq:mirror}
\mirror(\dpoint)
	&= \argmax_{\point\in\points} \{ \braket{\dpoint}{\point} - \hreg(\point) \}
	&\qquad
	&\text{for all $\dpoint\in\dpoints$}
\intertext{%
\item
The associated \define{Fenchel coupling} $\fench\from\points\times\dpoints\to\R$ of $\hreg$:%
}
\label{eq:Fench}
\fench(\base,\dpoint)
	&= \hreg(\base)
		+ \hconj(\dpoint)
		- \braket{\dpoint}{\base}
	&\quad
	&\text{for all $\base\in\points$, $\dpoint\in\dpoints$}.
\end{alignat}
\end{enumerate}
\smallskip

\begin{remark*}
The terminology ``Fenchel coupling'' is due to \cite{MS16,MZ19}, which we follow closely in terms of notation and conventions.
\end{remark*}

The proposition below provides some basic properties concerning the first two objects above:

\begin{proposition}
\label{prop:mirror}
Let $\hreg$ be a $\hstr$-strongly convex regularizer on $\cvx$.
Then:
\begin{enumerate}
[\upshape(\itshape a\hspace*{.5pt}\upshape)]
\item
$\mirror$ is single-valued on $\dpoints$;
in particular, for all $\point\in\dom\subd\hreg$ and all $\dpoint\in\dpoints$, we have:
\begin{equation}
\label{eq:hinv}
\point
	= \mirror(\dpoint)
	\quad
	\text{if and only if}
	\quad
\dpoint
	\in \subd\hreg(\point)
	\eqdot
\end{equation}

\item
The image $\im\mirror$ of $\mirror$ satisfies $\relint\cvx \subseteq \im\mirror = \dom\subd\hreg \subseteq \cvx$.

\item
The convex conjugate $\hconj\from\dpoints\to\R$ of $\hreg$ is differentiable and
\begin{equation}
\label{eq:Danskin}
\mirror(\dpoint)
	= \nabla\hconj(\dpoint)
	\quad
	\text{for all $\dpoint\in\dpoints$}.
\end{equation}

\item
$\mirror$ is $(1/\hstr)$-Lipschitz continuous, that is,
\begin{equation}
\label{eq:QLips}
\norm{\mirror(\dpointalt) - \mirror(\dpoint)}
	\leq (1/\hstr) \dnorm{\dpointalt - \dpoint}
	\quad
	\text{for all $\dpoint,\dpointalt\in\dpoints$}.
\end{equation}

\item
Fix some $\dpoint\in\dpoints$ and set $\point = \mirror(\dpoint)$.
Then, for all $\pointalt\in\points$ we have:
\begin{equation}
\label{eq:hdir}
\dir\hreg(\point;\pointalt - \point)
	\geq \braket{\dpoint}{\pointalt - \point}
	\eqdot
\end{equation}
In particular, if $\subd\hreg$ admits a continuous selection $\subsel\hreg\from\dom\subd\hreg\to\dpoints$, we have
\begin{equation}
\label{eq:subsel-var}
\braket{\subsel\hreg(\point)}{\pointalt - \point}
	\geq \braket{\dpoint}{\pointalt - \point}
	\quad
	\text{for all $\point\in\dom\subd\hreg$ and all $\point\in\cvx$},
\end{equation}
or, equivalently,
\begin{equation}
\label{eq:subsel-cone}
\subd\hreg(\point)
	= \subsel\hreg(\point) + \pcone(\point)
	\quad
	\text{for all $\point\in\dom\subd\hreg$},
\end{equation}
where
\begin{equation}
\label{eq:pcone}
\pcone(\point)
	= \setdef{\dvec\in\dpoints}{\braket{\dvec}{\pointalt-\point} \leq 0 \;\; \text{for all $\pointalt\in\points$}}
\end{equation}
denotes the \define{polar cone} to $\cvx$ at $\point$.
\end{enumerate}
\end{proposition}

\begin{proof}
These properties are fairly well known (except possibly the last one), so we only provide a quick proof or a precise pointer to the literature.
\begin{enumerate}
[\upshape(\itshape a\hspace*{.5pt}\upshape)]

\item
The maximum in \eqref{eq:mirror} is attained for all $\dpoint\in\dspace$ and is unique because $\hreg$ is strongly convex.
Furthermore, $\point$ solves \eqref{eq:mirror} if and only if $\dpoint - \subd\hreg(\point) \ni 0$, \ie if and only if $\dpoint\in\subd\hreg(\point)$.

\item
By \eqref{eq:hinv}, we readily get $\im\mirror = \dom\subd\hreg$.
Consequently, the rest of our claim follows from standard results in convex analysis, see \eg \citet[Chap.~26]{Roc70}.

\item
The equality $\mirror = \nabla\hconj$ follows immediately from Danskin's theorem, see \eg \citet[Proposition 5.4.8, Appendix B]{Ber15}.

\item
See \citet[Theorem 12.60(b)]{RW98}.

\item
Since $\dpoint \in \subd\hreg(\point)$ by \eqref{eq:hinv}, we readily get that
\begin{equation}
\hreg(\point + \theta(\pointalt-\point))
	\geq \hreg(\point) + \theta \braket{\dpoint}{\pointalt-\point}
	\quad
	\text{for all $\theta\in[0,1]$}
	\eqdot
\end{equation}
Hence, by rearranging and taking the limit $\theta\to0^{+}$,%
\footnote{The existence of the limit is guaranteed by standard results, see \eg \citet[Appendix B]{Ber15}.}
we conclude that
\begin{equation}
\label{eq:dir-subdiff}
\dir\hreg(\point;\pointalt-\point)
	= \lim_{\theta\to0^{+}} \frac{\hreg(\point + \theta(\pointalt-\point)) - \hreg(\point)}{\theta}
	\geq \braket{\dpoint}{\pointalt-\point}
\end{equation}
as claimed.
Finally, for our last assertion, let $\tanvec = \pointalt - \point$ and set
\begin{equation}
\label{eq:hreg-ray}
\phi(\theta)
	= \hreg(\point + \theta\tanvec)
		- \bracks{\hreg(\point) + \braket{\dpoint}{\theta\tanvec}}
	\quad
	\text{for all $\theta\in[0,1]$}
\end{equation}
so $\phi(\theta) \geq \hstr \theta^{2}\norm{\tanvec}^{2}/2 \geq 0$ for all $\theta\in[0,1]$.
By construction, it is straightforward to verify that the function $\psi(\theta) = \braket{\subsel\hreg(\point+\theta\tanvec) - \dpoint}{\tanvec}$ is a selection of subgradients of $\phi$, \ie
\begin{equation}
\phi(\alt\theta)
	\geq \phi(\theta) + \psi(\theta) (\alt\theta - \theta)
	\quad
	\text{for all $\theta,\alt\theta\in[0,1]$}.
\end{equation}
Since $\psi$ is in addition continuous (because $\subsel\hreg$ is), it follows that $\phi'(\theta) = \psi(\theta)$ for all $\theta\in[0,1]$ by a well-known characterization of the one-sided derivatives of convex functions, \cf \citet[Theorem~24.2]{Roc70}.
Hence, with $\phi$ convex and $\phi(\theta) \geq \phi(0)$ for all $\theta\in[0,1]$, we conclude that $\braket{\subsel\hreg(\point) - \dpoint}{\tanvec} = \psi(0) = \phi'(0) \geq 0$, and our proof is complete.
\qedhere
\end{enumerate}
\end{proof}

The next proposition collects some basic properties of the Fenchel coupling.

\begin{proposition}
\label{prop:Fench}
Let $\hreg$ be a $\hstr$-strongly convex regularizer on $\cvx$.
Then, for all $\base\in\points$ and all $\dpoint,\dpointalt\in\dpoints$, we have:
\begin{subequations}
\begin{flalign}
\label{eq:Fench-posdef}
\quad
(a)
	&\;\;
	\fench(\base,\dpoint)
	\geq 0
	\;\;
	\text{with equality if and only if $\base = \mirror(\dpoint)$}.
	&
	\\
\label{eq:Fench-norm}
\quad
(b)
	&\;\;
	\fench(\base,\dpoint)
	\geq \tfrac{1}{2} \hstr \, \norm{\mirror(\dpoint) - \base}^{2}.
	&
\end{flalign}
\end{subequations}
\end{proposition}

\begin{proof}
These properties are also fairly standard, but we provide a quick proof for completeness.
\begin{enumerate}
[\upshape(\itshape a\hspace*{.5pt}\upshape)]

\item
By the Fenchel\textendash Young inequality, we have $\hreg(\base) + \hconj(\dpoint) \geq \braket{\dpoint}{\base}$ for all $\base\in\points$, $\dpoint\in\dpoints$, with equality if and only if $\dpoint\in\subd\hreg(\base)$.
Our claim then follows from \eqref{eq:hinv}.

\item
Let $\point = \mirror(\dpoint)$ so $\dpoint \in \subd\hreg(\point)$ by \eqref{eq:hinv}.
Then, by the definition of $\fench$,
we have
\begin{align}
\fench(\base,\dpoint)
	&= \hreg(\base) + \hconj(\dpoint) - \braket{\dpoint}{\base}
	\notag\\
	&= \hreg(\base) + \braket{\dpoint}{\point} - \hreg(\point) - \braket{\dpoint}{\base}
	\explain{because $\dpoint \in \subd\hreg(\point)$}
	\\
	&\geq \hreg(\base) - \hreg(\point) - \dir\hreg(\point;\base-\point)
	\explain{by \cref{prop:mirror}}
	\\
	&\geq \tfrac{1}{2} \hstr \norm{\point - \base}^{2}
	\explain{by \eqref{eq:hstr}}
\end{align}
and our proof is complete.
\qedhere
\end{enumerate}
\end{proof}

In view of \cref{prop:Fench}, $\fench(\base,\dpoint)$ can be seen a ``primal-dual'' measure of divergence between $\base\in\points$ and $\dpoint\in\dpoints$, and the alternate expression \eqref{eq:Fenchi-max} is straightforward.
This observation will play a major role in the sequel.

\subsection{Basic lemmas}

Moving forward, we note that the various update steps in \eqref{eq:FTRL+} can be written as
\begin{equation}
\label{eq:update}
\new[\dpoint]
	= \dpoint + \dvec
	\quad
	\text{and}
	\quad
\new
	= \mirror(\new[\dpoint])
\end{equation}
for some $\dpoint,\dvec\in\dpoints$.
With this in mind, we proceed below to state a series of basic lemmas for the Fenchel coupling before and after an update of the form \eqref{eq:update}.
These results are not new, \cf \cite{JNT11,MZ19,MLZF+19} and references therein;
however, the assumptions used to derive them vary significantly in the literature, so we provide detailed proofs for completeness.

All of the results that follow below are stated for a $\hstr$-strongly convex regularizer on $\cvx$.
The first result is a primal-dual version of the so-called ``three-point identity'' for \acl{MD} \citep{CT93}:

\begin{lemma}
\label{lem:3point}
Fix some $\base\in\points$, $\dpoint\in\dpoints$, and let $\point = \mirror(\dpoint)$.
Then, for all $\new[\dpoint]\in\dpoints$, we have:
\begin{equation}
\label{eq:3point}
\fench(\base,\new[\dpoint])
	= \fench(\base,\dpoint)
		+ \fench(\point,\new[\dpoint])
		+ \braket{\new[\dpoint] - \dpoint}{\point - \base}.
\end{equation}
\end{lemma}

\begin{proof}
By definition, we have:
\begin{subequations}
\begin{align}
\label{eq:Fench1}
\fench(\base,\new[\dpoint])
	&= \hreg(\base)
		+ \hconj(\new[\dpoint])
		- \braket{\new[\dpoint]}{\base}
	\\
\label{eq:Fench2}
\fench(\base,\dpoint)
	&= \hreg(\base)
		+ \hconj(\dpoint)
		- \braket{\dpoint}{\base}
	\\
\label{eq:Fench3}
\fench(\point,\new[\dpoint])
	&= \hreg(\point)
		+ \hconj(\new[\dpoint])
		- \braket{\new[\dpoint]}{\point}
\end{align}
\end{subequations}
Thus, subtracting \eqref{eq:Fench2} and \eqref{eq:Fench3} from \eqref{eq:Fench1}, and rearranging, we get
\begin{equation}
\fench(\base,\new[\dpoint])
	= \fench(\base,\dpoint)
		+ \fench(\point,\new[\dpoint])
		- \hreg(\point)
		- \hconj(\dpoint)
		+ \braket{\new[\dpoint]}{\point}
		- \braket{\new[\dpoint] - \dpoint}{\base}
	\eqdot
\end{equation}
Our assertion then follows by recalling that $\point = \mirror(\dpoint)$, so $\hreg(\point) + \hconj(\dpoint) = \braket{\dpoint}{\point}$.
\end{proof}

The next result we present concerns the Fenchel coupling before and after a direct update step;
similar results exist in the literature, but we again provide a proof for completeness.

\begin{lemma}
\label{lem:onestep}
Fix some $\base\in\points$ and $\dpoint,\dvec\in\dpoints$.
Then, letting $\point = \mirror(\dpoint)$, $\new[\dpoint] = \dpoint + \dvec$, and $\new = \mirror(\new[\dpoint])$ as per \eqref{eq:update}, we have:
\begin{subequations}
\label{eq:onestep}
\begin{align}
\fench(\base,\new[\dpoint])
	&= \fench(\base,\dpoint)
		+ \braket{\dvec}{\new - \base}
		- \fench(\new,\dpoint)
	\label{eq:onestep1}
	\\
	&\leq \fench(\base,\point)
		+ \braket{\dvec}{\point - \base}
		+ \tfrac{1}{2} \hstr \dnorm{\dvec}^{2}
	\eqdot
	\label{eq:onestep2}
\end{align}
\end{subequations}
\end{lemma}

\begin{proof}
By the three-point identity \eqref{eq:3point}, we have
\begin{equation}
\fench(\point,\dpoint)
	= \fench(\point,\new[\dpoint])
		+ \fench(\new,\point)
		+ \braket{\dpoint - \new[\dpoint]}{\new - \base}
\end{equation}
so our first claim follows by rearranging.
For our second claim, simply note that
\begin{align}
\fench(\base,\dpoint)
		+ \braket{\dvec}{\new - \base}
		- \fench(\new,\dpoint)
	&= \fench(\base,\dpoint)
		+ \braket{\dvec}{\point - \base}
		+ \braket{\dvec}{\new-\point}
		- \fench(\base,\dpoint)
	\notag\\
	&\leq \fench(\base,\dpoint)
		+ \braket{\dvec}{\point - \base}
		+ \frac{1}{2\hstr} \dnorm{\dvec}^{2}
		+ \frac{\hstr}{2} \norm{\point-\base}^{2}
		- \fench(\base,\dpoint)
\end{align}
so our claim follows from \cref{prop:Fench}.
\end{proof}

The last result we present here is sometimes referred to as a ``four-point lemma'', and concerns the Fenchel coupling before and after an \emph{extrapolation} step:
\begin{lemma}
\label{lem:twostep}
Fix some $\base\in\points$ and $\dpoint,\dvec_{1},\dvec_{2}\in\dpoints$.
Then, letting $\point = \mirror(\dpoint)$, $\new[\dpoint_{i}] = \dpoint + \dvec_{i}$, and $\new[\point_{i}] = \mirror(\new[\dpoint_{i}])$, $i=1,2$, as per \eqref{eq:update}, we have:
\begin{subequations}
\label{eq:twostep}
\begin{align}
\fench(\base,\new[\dpoint_{2}])
	&= \vphantom{\frac{1}{2}}
		\fench(\base,\dpoint)
		+ \braket{\dvec_{2}}{\new[\point_{1}] - \base}
		+ \bracks*{ \braket{\dvec_{2}}{\new[\point_{2}] - \new[\point_{1}]} - \fench(\new[\point_{2}],\dpoint) }
	\label{eq:twostep1}
	\\
	&= \vphantom{\frac{1}{2}}
		\fench(\base,\dpoint)
		+ \braket{\dvec_{2}}{\new[\point_{1}] - \base}
		+ \braket{\dvec_{2} - \dvec_{1}}{\new[\point_{2}] - \new[\point_{1}]}
		- \fench(\new[\point_{2}],\new[\dpoint_{1}])
		- \fench(\new[\point_{1}],\dpoint)
	\label{eq:twostep2}
	\\
	&\leq \vphantom{\frac{1}{2}}
		\fench(\base,\dpoint)
		+ \braket{\dvec_{2}}{\new[\point_{1}] - \base}
		+ \frac{1}{2\hstr} \dnorm{\dvec_{2} - \dvec_{1}}^{2}
		- \frac{\hstr}{2} \norm{\new[\point_{1}] - \point}^{2}
	\eqdot
	\label{eq:twostep3}
\end{align}
\end{subequations}
\end{lemma}

\begin{proof}
By \cref{lem:onestep}, we have
\begin{align}
\fench(\base,\new[\dpoint_{2}])
	&= \fench(\base,\dpoint)
		+ \braket{\dvec_{2}}{\new[\point_{2}] - \base}
		- \fench(\new[\point_{2}],\dpoint)
\end{align}
so \eqref{eq:twostep1} follows by writing $\braket{\dvec_{2}}{\new[\point_{2}] - \base} = \braket{\dvec_{2}}{\new[\point_{1}] - \base} + \braket{\dvec_{2}}{\new[\point_{2}] - \new[\point_{1}]}$, and \eqref{eq:twostep2} follows from the three-point identity \eqref{eq:3point} for the Fenchel coupling.
Finally, for \eqref{eq:twostep3}, the Fenchel-Young inequality in Peter-Paul form yields
\begin{equation}
\label{eq:twostepx}
\braket{\dvec_{2} - \dvec_{1}}{\new[\point_{2}] - \new[\point_{1}]}
	\leq \frac{1}{2\hstr} \dnorm{\dvec_{2} - \dvec_{1}}^{2}
		+ \frac{\hstr}{2} \norm{\new[\point_{2}] - \new[\point_{1}]}^{2}
\end{equation}
and our claim follows again by invoking \cref{prop:Fench} to write
\begin{equation}
\frac{\hstr}{2} \norm{\new[\point_{2}] - \new[\point_{1}]}^{2}
	- \fench(\new[\point_{2}],\new[\dpoint_{1}])
	- \fench(\new[\point_{1}],\dpoint)
	\leq -\fench(\new[\point_{1}],\dpoint)
	\leq -\frac{\hstr}{2} \norm{\new[\point_{1}] - \point}^{2}
\end{equation}
and then substituting the result in \eqref{eq:twostepx}
\end{proof}

\Cref{lem:onestep,lem:twostep} will be responsible for most of the heavy lifting to derive a Lyapunov function for \eqref{eq:FTRL+}.
We discuss the relevant details in \cref{app:discrete}.

We conclude this section with a variational characterization of the abstract update \eqref{eq:update} in the case where $\subd\hreg$ of $\hreg$ admits a continuous selection \textendash\ or, alternatively, $\hreg$ is smooth in the sense of \eqref{eq:hsmooth}.

\begin{lemma}
\label{lem:subsel-var-new}
Fix some $\dpoint,\new[\dpoint]\in\dpoints$, and let $\new = \mirror(\new[\dpoint])$.
Then, for all $\base\in\points$, we have
\begin{equation}
\label{eq:subsel-var-new}
\braket{\new[\dpoint] - \dpoint}{\base - \new}
	\leq \braket{\subsel\hreg(\new) - \dpoint}{\base - \new}
	\eqdot
\end{equation}
\end{lemma}

\begin{proof}
Invoking \eqref{eq:subsel-var} in \cref{prop:mirror} with $\dpoint \gets \new[\dpoint]$, $\point \gets \new$, and $\pointalt \gets \base$, we get
\begin{equation}
\braket{\new[\dpoint]}{\base - \new}
	\leq \braket{\subsel\hreg(\new)}{\base - \new}
	\eqdot
\end{equation}
Our claim then follows by subtracting $\braket{\dpoint}{\base - \new}$ from both sides of the above.
\end{proof}

\section{Continuous-time analysis}
\label{app:continuous}

\subsection{Dynamical systems notions}
\label{sec:dynamical-systems-app}

To fix notation, we recall here some basics from the theory of dynamical systems, roughly following \citep{robinsonDynamicalSystemsStability1998, arnoldMathematicalMethodsClassical1989}. In this section, $\mfld$ is an open subset of a Euclidean space of dimension $\vdim$.

We consider a system of \acp{ODE} of the form
\begin{equation}
\tag{DS}
\label{eq:dynamical-system}
\dot{\point}\of\time = \vfield(\pointof\time) \, ,
\end{equation}
where $\pointof\time$ is a curve in $\mfld$ defined on an open interval $\interval \subseteq \R$ (that without loss of generality we assume to contain $0$), and $\vfield\from\mfld\to\R^{\vdim}$ is a smooth function. The function $\vfield$ is called \define{vector field} because it assigns a vector $\vfield\of\point$ to each point $\point$ in $\mfld$, and \eqref{eq:dynamical-system} is called \define{dynamical system generated by $\vfield$}.

Given $\pstart \in \mfld$, an \define{orbit with initial condition $\pstart$} is a solution $\pointof\time$ of \eqref{eq:dynamical-system} with $\pointof\tstart = \pstart$. The \define{flow generated by $\vfield$} is the smooth function $\flowmap\from \interval \times \mfld \to\mfld$ such that $\flowt{\tstart}{\pstart} = \pstart$ for all $\pstart \in \mfld$ and $\frac{d}{dt} \flowt{\time}{\point} = \vfield( \flowt{\time}{\point} )$ for all $\time \in \interval$. In words, $\flowt{\time}{\pstart}$ is the orbit $\pointof\time$  with initial condition $\pstart$; the existence and uniqueness of this function is guaranteed by the existence and uniqueness theorem of solutions of \aclp{ODE}.

A flow $\flowmap$ is called \define{volume-preserving} if $\vol\parens{ \flowt{\time}{\pstarts}} = \vol(\pstarts) $ for any (Lebesgue) measurable subset $\pstarts \subseteq \mfld$ and all $\time \in \interval$. Liouville's theorem gives a sufficient condition for a flow to be volume-preserving based on the \define{divergence} of its generating field:\footnote{Recall here that the divergence is a differential operator mapping a vector field $\vfield$ on $\mfld$ to the real-valued function $\diver \vfield \of \point \defeq \sum_{\pure = 1}^{\vdim} \pd_{\pure}\vfield^{\pure}\of\point$, where $\pd_{\pure}$ is a shorthand for the partial derivative $\pd/\pd\point_{\pure}$}
\begin{theorem*}[Liouville]
If $\diver\vfield \equiv 0$ then the flow generated by $\vfield$ is volume-preserving.
\end{theorem*}
Volume-preserving flows are closely related to recurrent dynamical patterns. A point $\point \in \mfld$ is said to be \define{recurrent} under \eqref{eq:dynamical-system} if, for every neighborhood $\pstarts$ of $\point \in \mfld$,
there exists an increasing sequence of time $\time_{\run}\uparrow\infty$
such that $\flowt{\time_{\run}}{\point}$ is defined and falls in $\pstarts$ for all $\run$. Moreover, \eqref{eq:dynamical-system} is said to be \define{\acl{PR}} if almost every point $\point \in \mfld$ is recurrent.
The celebrated Poincaré recurrence theorem gives a sufficient condition for a dynamical system to be \acl{PR}:
\begin{theorem*}[Poincaré]
Let $\vfield$ be a smooth vector field on $\mfld$. If the flow induced by $\vfield$ is volume-preserving and all the orbits of \eqref{eq:dynamical-system} are bounded, then \eqref{eq:dynamical-system} is \acl{PR}.
\end{theorem*}

\subsection{Basic properties of \ac{FTRL}}
\label{sec:continuous-ftrl-app}
In this section we survey some of the properties of the \acl{FTRL} learning scheme in a continuous-time, multi-agent setting, in line with the presentations of \citep{MS16, MPP18, FVGL+20}. For ease of reference we recall here some of the notions introduced in \cref{app:mirror} and in \cref{sec:prelims,sec:dynamics} from the main body of the paper.

Let $\fingame = \fingamefull$ be a finite normal form game, and let $\payfield$ denote its payoff field. The game's strategy space is $\strats = \prod_{\playalt \in \players} \simplex(\pures_{\playalt}) \subseteq \vambient \defeq \ambient $, and the game's payoff space is
$\scores \defeq \vambientdual$. The payoff field is a map $\payfield\from\vambient\to\scores$ that evaluated at a strategy $\strat \in \strats$ acts linearly on any $\stratalt \in \strats$ by
\begin{equation}
\begin{split}
\dualp{\payfield\of\strat}{\stratalt}
& = \insum_{\play\in\players} \dualp{\payfield_{\play}\of\strat}{\stratalt_{\play}}
= \insum_{\play\in\players} \insum_{\pure_{\play} \in \pures_{\play}} \payfield_{\findex}\of\strat \,  \stratalt_{\findex} \\
& = \insum_{\play\in\players} \pay_{\play}(\stratalt_{\play}, \strat_{\others}) \in \R \eqdot
\end{split}
\end{equation}

Assume now that $\fingame$ is played continuously over time. As discussed in \cref{sec:dynamics}, the main idea behind the \acl{FTRL} learning scheme is that, at any given time $\time\geq\tstart$, each player $\play\in\players$ tracks their cumulative payoff up to time $\time$ and plays a ``regularized'' best response strategy in light of this information. Concretely, given a cumulative payoff vector $\scoreofi\time \in \scores_{\play}$, each player $\play \in \players$ selects this optimal strategy $\stratofi\time \in \strats_{\play}$ by means of a \define{regularized best response} map $\mirror_{\play} \from \scores_{\play} \to \strats_{\play}$, a single-valued analogue of the best-response correspondence $\score_{\play} \mapsto \argmax_{\strat_{\play}\in\strats_{\play}} \braket{\score_{\play}}{\strat_{\play}}$. A standard way \cite{SS11} of obtaining such map is to introduce a \textit{regularizer function} $\hreg_{\play}\from\strats_{\play}\to\R$ that is \begin{enumerate*}
[(\itshape i\hspace*{1pt}\upshape)]
\item continuous on $\strats_{\play}$, 
\item smooth on $\intstrats_{\play}$, the relative interior of $\strats_{\play}$, and
\item strongly convex on $\strats_{\play}$ (as per \cref{eq:hstr})
\end{enumerate*}; and to consider the induced \define{choice map} $\mirror_{\play} \from \scores_{\play} \to \strats_{\play}$ defined by
\[
\EqMirrori
\]
By \cref{prop:mirror}, $\mirror_{\play}$ is well-defined and Lipschitz continuous, and it coincides with the differential $\nabla \hconj_{\play}$ of $\hconj_{\play}:\scores_{\play}\to\R$, the \textit{convex conjugate} of $\hreg_{\play}$.

In a continuous time setting, this regularized learning scheme translates into the following implicit equations of motion, which govern the evolution of the cumulative payoff $\scoreof\time \in \scores$ and of the mixed strategy profile $\stratof\time \in \strats$ as the players attempt to maximize their payoff over time:
\begin{equation}
\label{eq:FTRL-y}
\score_{\findex}\of\time
	= \score_{\findex}\of\tstart + \int_{\tstart}^{\time}
	\payfield_{\findex}(\stratof\timealt)
	\dd\timealt
	\quad
	\text{with}
	\quad
\stratofi{\time}
	= \mirror_{\play}(\scoreofi{\time})  \, ,
\end{equation}
for all $\play \in \players, \pure_{\play} \in \pures_{\play}$. A straightforward computation shows that this is equivalent to \cref{eq:FTRL-x} from \cref{sec:dynamics} in the main text, that governs the evolution of the mixed strategy $\stratof\time\in\strats$:
\[
\EqFtrlx \eqdot
\]
Importantly, \cref{eq:FTRL-y} can be cast in the form \eqref{eq:dynamical-system} of a dynamical system in the game's payoff space. For each $\play \in \players$, differentiation with respect to $\time$ yields
\[
\EqFtrlCont \, ,
\]
and by aggregating the player indices we obtain the system of \acp{ODE}
\begin{equation}
\label{eq:FTRL-ode-system}
\dot\score = \dualdyn (\score) \, ,
\end{equation}
where $\dualdyn \defeq \payfield \circ \mirror : \scores \to \scores$ is a continuous vector field on $\scores$; \cf \cref{diag:FTRL-diagram}.

Existence and uniqueness of a global solution $\scoreof\time \in \scores$ of \cref{eq:FTRL-ode-system} for any initial condition $\scoreof\tstart \in \scores$
are guaranteed by standard arguments \citep[Prop. 3.1]{MS16}; in line with the terminology of the previous section we will refer to such a solution as a \define{\acl{DO}}.

\subsection{Constant of motion for harmonic games}
\label{sec:constant-fenchel}
The following result shows that \ac{FTRL} in harmonic games admits a constant of motion.

\begin{propositionApp}
Let $\fingame = \fingamefull$ be a finite game and consider a vector $\wei \in \R^{\nPlayers}_{++}$ and a fully mixed strategy $\heq \in \strats$. Then the weighted Fenchel coupling $\henergy \from \scores \to \R$ defined by
\begin{equation}
\label{eq:constant-fenchel}
\henergy\of\score
\defeq \insum_{\play} \wei_{\play} \fench_{\play}(\heq_{\play}, \score_{\play})
=  \insum_{\play} \wei_{\play} \left( \hreg_{\play}(\heq_{\play}) + \hconj_{\play}\of\scori - \dualp{\heq_{\play}}{\scori} \right)
\end{equation}
is a constant of motion under \eqref{eq:FTRL-cont} if and only if $\fingame$ is harmonic with strategic center $\stratcenter$.
\end{propositionApp}

\begin{proof}
Let $\scoreof\time$ be a \acl{DO}. Then by chain rule
\begin{equation}
\begin{split}
\frac{d}{dt} \henergy(\scoreof\time)
& = \insum_{\play} \wei_{\play} \big[  \dualp{\nabla\hconj_{\play}(\scori)}{\dot{\scori}}  -
\dualp{\heq_{\play}}{\dot{\scori}}   \big]
= \insum_{\play} \wei_{\play}  \,  \dualp{\strati\of\time - \heq_{\play}}{\payfield_{\play}(\stratof\time)}
\end{split}
\end{equation}
where the second equality holds by \eqref{eq:FTRL-cont} and \cref{eq:Danskin}. Then, by the characterization of harmonic games in terms of a strategic center \eqref{eq:harmonic-center}, the time derivative of the weighted Fenchel coupling vanishes identically along a dual orbit of \eqref{eq:FTRL-cont} precisely if the underlying game is harmonic.
\end{proof}
The existence of this constant of motion is fundamental for proving \cref{thm:recurrence}, \ie the Poincaré recurrence of continuous-time \ac{FTRL} in harmonic games. With this key element established, the remainder of this appendix closely follows the proof technique described by \citep{MPP18} for the analogous result in the context of \aclp{2ZSG}.

\subsection{\ac{FTRL} in the space of payoff differences}
\label{sec:ftrl-z-space}

For any initial condition $\scoreof\tstart \in \scores$, a \acl{DO} of \eqref{eq:FTRL-cont} induces a curve $\stratof\time = \mirror(\scoreof\time)$ in the game's strategy space $\strats$
which solves \cref{eq:FTRL-x} for all $\time \geq 0$; in the following we will refer to such curve as \define{\acl{TP}}. Crucially, a \acl{TP} is in general \textit{not} the global solution of a dynamical system $\dot{\strat} = \primaldyn\of\strat$ for some vector field $\primaldyn\from\strats\to\strats$ in the game's strategy space.
The reason for this is that the map $\mirror\from\scores\to\strats$ is not necessarily invertible, so
there is in general no way to identify a unique a vector field $\primaldyn$ on $\strats$ that is \define{related} to the vector field $ \dualdyn$ on $\scores$ via $\mirror$.%

\para{Related vector fields and induced dynamical systems}
The concept of vector fields related by a smooth map is standard in differential geometry (\eg \cite[p. 181]{Lee12}). Let $\mfld, \mfldalt$ be open subsets of Euclidean space: given a vector field $\dualdyn$ on $\mfld$ and a smooth map $F \from \mfld \to \mfldalt$,  a vector field $\primaldyn$ on $\mfldalt$ is called \define{$F$-related to $\dualdyn$} if, for all $\score \in \mfld$, $  (\jac{F})_{\score} \cdot \dualdyn\of{\score}  = \primaldyn\of\strat$, with $\strat = F\of\score$. Here  $\jac{F}$ is the Jacobian matrix of $F$, and $\cdot$ represents matrix-vector multiplication. If $F$ is invertible then such vector field exists always and is unique;
else, it might exist and not be unique, or not exist at all.

Vector fields that are related via a smooth map are useful inasmuch as they generate ``compatible'' dynamical systems:
\begin{lemmaApp}
\label{lem:related-vector-fields}
Let $F \from \mfld \to \mfldalt$ be a smooth map between open subsets of Euclidean spaces, and let $\dot{\score} = \dualdyn\of\score$ be a dynamical system on $\mfld$. Let $\scoreof\time$ be an orbit with initial condition $\score_{\tstart} \in \mfld$, and consider the curve on $\mfldalt$ defined by  $\stratof\time \defeq F(\scoreof\time)$. If there exists a vector field $\primaldyn$ on $\mfldalt$ that is $F$-related to $\dualdyn$, then the curve $\stratof\time$ is an orbit of the dynamical system $\dot{\strat} = \primaldyn(\strat)$ with initial condition $\strat_{\tstart} = F(\score_{\tstart})$.
\end{lemmaApp}
\begin{proof}
By chain rule, 
\begin{equation}
\frac{d}{dt}\stratof\time
= \frac{d}{dt} F(\scoreof\time)
= (\jac{F})_{\scoreof\time} \cdot \dot{\score}(\time)
= (\jac{F})_{\scoreof\time} \cdot \dualdyn(\scoreof\time)
= \primaldyn( \stratof\time ) \, ,
\end{equation}
where the last equality follows by the assumption that $\primaldyn$ is $F$-related to $\dualdyn$.
\end{proof}
In the following, if $F \from \mfld \to \mfldalt$ is a smooth function between open subsets of Euclidean spaces, and $\dualdyn, \primaldyn$ are vector fields fulfilling the assumptions of \cref{lem:related-vector-fields}, we say that the dynamical system $\dot{\score} = \dualdyn\of\score$ on $\mfld$ \define{induces the dynamical system $\dot{\strat} = \primaldyn\of\strat$ on $\mfldalt$ via $F$}.

\para{\ac{FTRL} induced in the space of payoff differences}
The choice map $\mirror \from \scores \to \strats$ is in general not smooth, and neither injective nor surjective \cite[Sec.3]{FVGL+20}, so it generally does not allow
to induce the dynamical system \eqref{eq:FTRL-ode-system} from the game's payoff space $\scores$ to the game's strategy space $\strats$.
\footnote{A detailed treatment of the conditions under which a \acl{TP} $\stratof\time$ actually \textit{is} a solution of dynamical system in the game's strategy space $\strats$ is beyond the scope of this work; we refer the interested reader to \citep{MS16,MerSan18} for an in-depth treatment.} In other words, the learning process \eqref{eq:FTRL-cont} in a finite game gives rise to a dynamical system in the game's payoff space $\scores$, to which the theorems presented in \cref{sec:dynamical-systems-app} can in principle be applied; however, it can be shown that the orbits of \cref{eq:FTRL-ode-system} in $\scores$ are \textit{not} bounded, preventing the application of Poincaré's theorem. Furthermore, the \aclp{DO} do not convey direct information on the day-to-day behavior of the players, due to the lack of invertibility of the choice map.

Conversely, the objects of interest from a dynamical, learning viewpoint \textendash{} that is, the \aclp{TP} in the game's strategy space $\strats$ \textendash{}  present technical difficulties and do not easily fit the dynamical systems framework depicted in \cref{sec:dynamical-systems-app}. In the following we show how these difficulties can be circumvented by analyzing the dynamics induced by \eqref{eq:FTRL-cont} in yet a third space $\scorediffs$, that arises by taking the \textit{differences} between payoffs \textendash{} rather than their absolute values \textendash{} as the objects of study.
\begin{figure}
\centering
\begin{tikzcd}[row sep=large,
column sep=large,
every label/.append style = {font = \normalsize},
nodes={font=\normalsize}]
\vambient = \ambient \arrow[r, "\payfield"] & \scores = \vambientdual  \arrow[d, "\diffquot"] \arrow[r, "\fench"] \arrow[dl, "\mirror"] & \R \\
\strats \arrow[u, hookrightarrow]                           & \scorediffs \arrow[l, "\effmirror"]
\end{tikzcd}
\caption{Commutative diagram of the maps discussed in \cref{sec:continuous-ftrl-app,sec:constant-fenchel,sec:ftrl-z-space}; note in particular that $\payfield \circ \mirror$ is a vector field on $\scores$. The notation $\strats \hookrightarrow \vambient$ is equivalent to $\strats \subseteq\vambient$.}
\label{diag:FTRL-diagram}
\end{figure}

To make this precise, given the game $\fingame = \fingamefull$ fix a benchmark strategy $\purebench_{\play} \in \pures_{\play}$ for every player $\play\in\players$, and consider the hyperplane 
$\scorediffs_{\play}
\defeq \setdef{\scorediff_{\play} \in \R^{\nPures_{\play}}}{ \scorediff_{\benchfindex} =0 }
\subset  \R^{\nPures_{\play}}$. Clearly, $\scorediffs_{\play}  \cong  \R^{\nPures_{\play}-1}$. Each player's strategy space $\scores_{\play} = \R^{\nPures_{\play}}$ can be mapped onto $\scorediffs_{\play}$ by the linear operator
\begin{equation}
\label{eq:Z-quotienti}
\diffquot_{\play} \from \scores_{\play} \to \scorediffs_{\play}
\quad \text{with} \quad
\scorediff_{\findex} = \score_{\findex} - \score_{\benchfindex}
\end{equation}
for all $\pure_{\play} \in \pures_{\play}$.

$\diffquot_{\play}$ is clearly smooth, and a standard check shows that $\diffquot_{\play}$ is surjective and not injective:
$\ker\diffquot_{\play} = \setdef{\score_{\play}}{\score_{\findex} = \score_{\play \purealt_{\play}} \text{ for all } \pure_{\play}, \purealt_{\play} \in \pures_{\play}}$
is the $1$-dimensional linear subspace spanned by the vector $\ones_{\play} = (1, \dots, 1) \in \scores_{\play}$; and $\diffquot^{-1}(\scorediff_{\play}) = \scorediff_{\play} + \ker{\diffquot_{\play}}$ for any $\scorediff_{\play} \in \scorediffs_{\play}$. In particular, for all $\score_{\play}, \scorealt_{\play} \in \scores_{\play}$, we have that $\diffquot_{\play}(\score_{\play}) = \diffquot_{\play}(\scorealt_{\play}) $ if and only if $\score_{\play} - \scorealt_{\play}$ is proportional to $\ones_{\play}$. 

Since every $\scorediff_{\play} \in \scorediffs_{\play}$ is the image of some payoff $\score_{\play}$ via $\PI_{\play}$, the space $\scorediffs \defeq \prod_{\playalt}\scorediffs_{\playalt}$ is called the game's \define{payoff-difference space}; we will denote by $\PI$ the product map $\PI \equiv \prod_{\playalt}\PI_{\playalt}$, \ie (\cf \cref{diag:FTRL-diagram})
\begin{equation}
\label{eq:Z-quotient}
\PI \from \scores \to \scorediffs,
\quad
\PI\of\score \defeq (\PI_{\play}(\score_{\play}))_{\play\in\players} \eqdot
\end{equation}
\begin{lemmaApp}
\label{lem:Q-invarianti-PI-fibers}
The choice map $\mirror\from\scores\to\strats$ is invariant on the level sets of $\PI$.
\end{lemmaApp}
\begin{proof}
Let $\score, \scorealt \in \scores$. By the discussion above, $\PI\of\score = \PI\of\scorealt$ iff $ \scorealt_{\play} - \score_{\play} = \lambda \ones_{\play}$ for some $\lambda \in \R$. Then for each $\play \in \players$,
\[
\mirror_{\play}\of{\scorealt_{\play}}
= \argmax_{\strat_{\play} \in \strats_{\play}} \setof{ \dualp{\scorealt_{\play}}{\strat_{\play}} - \hreg_{\play}(\strati) }
= \argmax_{\strat_{\play} \in \strats_{\play}} \setof{ \dualp{\score_{\play}}{\strat_{\play}} + \lambda\dualp{\ones_{\play}}{\strat_{\play}} - \hreg_{\play}(\strati) }
= \mirror_{\play}\of{\scori} \eqdot
\qedhere
\]
\end{proof}
\begin{propositionApp}
The dynamical system \eqref{eq:FTRL-ode-system} in the game's payoff space $\scores$ induces a dynamical system
\begin{equation}
\label{eq:Z-dynamical-system}
\dot{\scorediff} = \diffdyn\of\scorediff
\end{equation}
in the game's payoff-difference space $\scorediffs$ via the map \eqref{eq:Z-quotient}.
\end{propositionApp}
\begin{proof}
By the discussion in the previous section (and in particular \cref{lem:related-vector-fields}), if we exhibit a vector field $\diffdyn$ on $\scorediffs$ that is $\PI$-related to $\dualdyn$, then our proof is complete. Thus we look for a vector field $\diffdyn$ such that, for all $\score \in \scores$,
\begin{equation}
(\jac\PI)_{\score} \cdot \dualdyn\of{\score}
= \diffdyn\of\piscore,
\end{equation}
with $\piscore = \PI\of\score$. By \cref{eq:Z-quotienti}, $(\jac\Pii)_{ \puri \purialt } = \delta_{ \puri \purialt } - \delta_{ \purebench_{\play} \purialt }$. Since $\dualdyn = \payfield \circ \mirror$, the sought-after vector field $\diffdyn$ must fulfill, for all $\score \in \scores$ and all $\puri \in \puresi$,
\begin{equation}
\label{eq:sought-after-related-field}
 \parens{  \payfield_{\findex}  - \payfield_{\bindex} } \circ \mirrori\of\scori
 = \Zdyn_{\findex}(\diffi) \, ,
\end{equation}
with $\piscore = \PI\of\score$. For each $\play \in \players$ define now (\cf \cref{diag:FTRL-diagram})
\begin{equation}
\label{eq:effective-mirror}
\effmirrori\from\diffsi\to\stratsi, \quad \effmirrori\of\diffi = \mirror\of\scori
\end{equation}
for \textit{any} $\scori \in \Pii^{-1}(\diffi) $, and denote by $\effmirror\from\diffs\to\strats$ the induced product map. Such map exists since $\Pii$ is surjective, and is well-defined by \cref{lem:Q-invarianti-PI-fibers}. It follows that the vector field on $\diffs$ defined by
\begin{equation}
\label{eq:Zdyn}
\Zdyn_{\findex}(\diffi)
\defeq \parens{  \payfield_{\findex}  - \payfield_{\bindex} } \circ \effmirrori\of\diffi
\end{equation}
for all $ \play \in \players, \diffi \in \diffsi,  \puri \in \puresi$ fulfills \cref{eq:sought-after-related-field}, and is hence $\PI$-related to $\dualdyn$.
\end{proof}
This result shows that, for every dual orbit $\scoreof\time$ of \cref{eq:FTRL-ode-system} with initial condition $\score_{\tstart} \in \scores$, the curve $\scorediffof\time = \PI(\scoreof\time)$ is  an orbit of the dynamical system \eqref{eq:Z-dynamical-system} in $\diffs$ with initial condition $\PI(\score_{\tstart})$. To conclude this section we give a result implying that if the weighted Fenchel coupling \eqref{eq:constant-fenchel} is a constant of motion constant then the solution orbits of \eqref{eq:Z-dynamical-system} in $\diffs$ are bounded.

\begin{lemmaApp}
\label{lem:bounded-score-differences}
For any $\play \in \players$, let $\scorin$ be a sequence in $\scores_{\play}$, and let $\base_{\play}$ be a point in the relative interior of $\stratsi$. If $\sup_{\run}\abs{\hconj_{\play}(\scorin)-\dualp{\scorin}{\base_{\play}}} < \infty$, then also the score differences remain bounded, \ie $\abs{ \score_{\findex, \run} - \score_{\play \purealt_{\play}, \run} } < \infty$ for all $\puri, \purialt \in \puresi$ and all $\run$.
\end{lemmaApp}
\begin{proof}
See \citep[Appendix D]{MPP18}.
\end{proof}
\begin{lemmaApp}
\label{lem:bounded-score-differences-1}
If the weighted Fenchel coupling \eqref{eq:constant-fenchel} is a constant of motion under \eqref{eq:FTRL-cont} for some fully mixed $\base \in \strats$ then the orbits of $\dot{\scorediff} = \diffdyn\of\scorediff$ as in \cref{eq:Z-dynamical-system} are bounded in $\diffs$.
\end{lemmaApp}
\begin{proof}
Assume that 
$\henergy\of\score
=\insum_{\play} \wei_{\play} \fench_{\play}(\base_{\play}, \score_{\play})
=  \insum_{\play} \wei_{\play} \left( \hreg_{\play}(\base_{\play}) + \hconj_{\play}\of\scori - \dualp{\base_{\play}}{\scori} \right)$
is a constant of motion for \eqref{eq:FTRL-cont} for some fully mixed $\base \in \strats$ and some $\mass \in \R^{\nPlayers}_{++} $.
Let $\score\of\time$ be an orbit of \eqref{eq:FTRL-cont} in $\scores$, and let $\score_{\play, \run} \defeq \score_{\play}\of{\time_{\run}}$ for any sequence of time $\time_{\run}$. Let furthermore $\fench_{\play, \run} \defeq  \hconj_{\play}(\score_{\play, \run}) - \dualp{\base_{\play}}{\score_{\play, \run}}  $. Then $\sup_{\run}\abs{\fench_{\play,\run}} < \infty$. By \cref{lem:bounded-score-differences}, this implies that $\abs{\scorediff_{\findex, \run}} < \infty$ for all $\pure_{\play} \in \pures_{\play}$, all $\play \in \players$, and all $\run$.
\end{proof}
\subsection{Recurrence of \ac{FTRL} in harmonic games}
\label{sec:ftrl-recurrence-app}
We now have all the ingredients to prove that almost every \acl{TP} $\stratof\time$ of \eqref{eq:FTRL-cont} in harmonic games returns arbitrarily close to its starting point infinitely often.

\ThRecurrence*

\begin{proof}[Proof of \cref{thm:recurrence}]
The proof relies on the following steps:
\begin{enumerate}

\item  the vector field $\diffdyn$ defined in \cref{eq:Zdyn} has vanishing divergence, so its induced flow is volume-preserving in $\diffs$ by Liouville's theorem;

\item the orbits of the dynamical system $\dot{\scorediff} = \diffdyn\of\scorediff$ of \cref{eq:Z-dynamical-system} are bounded in $\diffs$ since the weighted Fenchel coupling \eqref{eq:constant-fenchel} is a constant of motion for \ac{FTRL} in harmonic games;

\item the dynamical system $\dot{\scorediff} = \diffdyn\of\scorediff$ is recurrent in $\diffs$ by Poincaré theorem;

\item by continuity of \cref{eq:effective-mirror}, almost every \acl{TP} $\stratof\time$ of \eqref{eq:FTRL-cont} with initial condition in the image of $\effmirror$ returns arbitrarily close to its starting point infinitely often.

\end{enumerate}
Indeed, 
$
\diver \diffdyn \of \scorediff =
\insum_{\play} \insum_{\puri} \frac{\pd}{\pd\scorediff_{\findex}} ( \parens{  \payfield_{\findex}  - \payfield_{\bindex} } \circ \effmirrori\of\diffi)
$. For the first term, by chain rule,
\[
\begin{split}
\diver \diffdyn \of \scorediff & =
\insum_{\play} \insum_{\puri} \frac{\pd \payfield_{\findex} }{\pd\scorediff_{\findex}}    \parens{ \effmirrori\of\diffi}
= \insum_{\play} \insum_{\puri} \insum_{\playalt} \insum_{\purealt_{\playalt}} \frac{\pd \payfield_{\findex}}{\pd \strat_{\playalt \purealt_{\playalt}}} (\effmirror\of\scorediff)\frac{\pd \effmirror_{\playalt \purealt_{\playalt}}}{\pd \scorediff_{\findex}}\of\scorediff \\
& = \insum_{\play} \insum_{\puri}  \insum_{\purealt_{\play}} \frac{\pd \payfield_{\findex}}{\pd \strat_{\play \purealt_{\play}}} (\effmirror\of\scorediff)\frac{\pd \effmirror_{\play \purealt_{\play}}}{\pd \scorediff_{\findex}}\of\scorediff \equiv 0
\end{split}
\]
since $\frac{\pd \payfield_{\findex}}{\pd \strat_{\play \purealt_{\play}}} \equiv 0$ by multilinearity of the payoff functions. The second term yields identical result with $\purebench_{\play} \leftarrow \puri$, so we conclude that $\diver{\diffdyn} = 0$. By \cref{lem:bounded-score-differences-1}, the invariance of the weighted Fenchel coupling under \eqref{eq:FTRL-cont} implies that the payoff differences $\scorediff_{\findex}\of\time = \score_{\findex}\of\time - \scorediff_{\bindex}\of\time $ remain bounded for all $\time \in \postime$. So, by Poincaré theorem, the dynamical system $\dot{\scorediff} = \diffdyn\of\scorediff$ is Poincaré recurrent, \ie there exists a sequence of time $\time_{\run} \uparrow \infty$ such that $\lim_{\run\to\infty}\scorediffof{\time_{\run}} = \scorediff_{\tstart}$ for almost every $\scorediff_{\tstart} \in \scorediffs$. By continuity of \eqref{eq:effective-mirror}, almost every trajectory of play $\stratof\time = \mirror(\scoreof\time) = \effmirror(\scorediff\of\time)$ with $\strat_{\tstart} \in \im{\effmirror}$ fulfills $ \lim_{\run\to\infty} \stratof{\time_{\run}} = \strat_{\tstart}$,  which concludes our proof by noting that the image of $\effmirror$ is the same as the image of $\mirror$.
\end{proof}


\section{Discrete-time analysis}
\label{app:discrete}

In this appendix, our aim is to provide the full proofs for the discrete-time guarantees of \eqref{eq:FTRL+}, as presented in \cref{sec:algorithms}.
Our analysis hinges on a series of energy functions and associated template inequalities, which we introduce in the next section.

\subsection{Lyapunov functions and template inequalities for \eqref{eq:FTRL+}}

The main building block of our analysis is a suitable Lyapunov function for the discrete-time algorithmic template \eqref{eq:FTRL+}.
Motivated by the continuous-time analysis of \cref{app:continuous}, we begin by considering the player-specific Fenchel couplings
\begin{equation}
\label{eq:Fenchi}
\fench_{\play}(\base_{\play},\score_{\play})
	= \hreg_{\play}(\base_{\play}) + \hconj_{\play}(\score_{\play}) - \braket{\score_{\play}}{\base_{\play}}
	\quad
	\text{for $\base_{\play}\in\strats_{\play}$, $\score_{\play} \in \scores_{\play}$}
\end{equation}
induced by the regularizer $\hreg_{\play}$ of player $\play\in\players$.

Suppose now that the game is harmonic relative to some measure $\meas = (\meas_{1},\dotsc,\meas_{\nPlayers})$, let $\mass_{\play} = \sum_{\pure_{\play}\in\pures_{\play}} \meas_{\play\pure_{\play}}$ denote the mass of $\meas_{\play}$, and assume further that each player is running \eqref{eq:FTRL+} with learning rate $\learn_{\play} > 0$.
Our analysis will hinge on the energy function
\begin{align}
\tag{\ref{eq:energy}}
\energy(\base,\score)
	&= \sum_{\play\in\players} \frac{\mass_{\play}}{\learn_{\play}} \fench_{\play}(\base_{\play},\score_{\play})
	\qquad
	\base\in\strats,
	\score\in\scores,
\end{align}
which, as we show below, satisfies the following template inequality:

\begin{proposition}
\label{prop:template}
Suppose that each player is running \eqref{eq:FTRL+} with learning rate $\learn_{\play}>0$ in a harmonic game as above.
Then, for all $\base_{\play}\in\strats_{\play}$, $\play\in\players$, and all $\run=\running$, the algorithm's energy $\curr[\energy] \defeq \energy(\base,\curr[\statealt])$ enjoys the iterative bound:
\begin{align}
\label{eq:template}
\next[\energy]
	\leq \curr[\energy]
		&+ \sum_{\play\in\players} \mass_{\play}
			\braket{\payv_{\play}(\lead)}{\leadi - \base_{\play}}
	\notag\\
		&+ \sum_{\play\in\players} \mass_{\play}
			\braket{\payv_{\play}(\lead) - \payfield_{\play}(\curr)}{\nexti - \leadi}
	\notag\\
		&+ \sum_{\play\in\players} \mass_{\play} (1 - \coef_{\play})
			\braket{\payfield_{\play}(\curr) - \payfield_{\play}(\beforelead)}{\nexti - \leadi}
	\notag\\
		&- \sum_{\play\in\players} \frac{\mass_{\play}}{\learn_{\play}} \fench_{\play}(\nexti,\leadi[\statealt])
	\notag\\
		&- \sum_{\play\in\players} \frac{\mass_{\play}}{\learn_{\play}} \fench_{\play}(\leadi,\curri[\statealt])
		\eqdot
\end{align}
\end{proposition}

\begin{proof}
We begin by applying the bound \eqref{eq:twostep2} of \cref{lem:twostep} with the array of substitutions
\begin{enumerate}
\setlength{\itemsep}{0pt}
\item
$\base \gets \base_{\play}$
\item
$\dvec_{1} \gets \learn_{\play} \curri[\signal] = \learn_{\play} \coef_{\play} \, \payfield_{\play}(\curr) + \learn_{\play} (1-\coef_{\play}) \, \payfield_{\play}(\beforelead)$
\item
$\dvec_{2} \gets \learn_{\play} \leadi[\signal] = \learn_{\play} \payfield_{\play}(\lead)$
\item
$\dpoint \gets \curri[\statealt]$
	\tabto{6em}
	so
	\tabto{8em}
	$\point \gets \mirror_{\play}(\curri[\statealt]) = \curri$
\item
$\new[\dpoint_{1}] \gets \leadi[\statealt]$
	\tabto{6em}
	so
	\tabto{8em}
	$\new[\point_{1}] \gets \leadi$
\item
$\new[\dpoint_{2}] \gets \nexti[\statealt]$
		\tabto{6em}
	so
	\tabto{8em}
	$\new[\point_{2}] \gets \nexti$
\end{enumerate}
We then get
\begin{align}
\braket{\dvec_{2} - \dvec_{1}}{\new[\point_{2}] - \new[\point_{1}]}
	&= \learn_{\play} \braket
		{\payfield_{\play}(\lead) - \coef_{\play} \, \payfield_{\play}(\curr) - (1-\coef_{\play}) \, \payfield_{\play}(\beforelead)}{\nexti - \leadi}
	\notag\\
	&= \learn_{\play} \braket{\payfield_{\play}(\lead) - \payfield_{\play}(\curr)}{\nexti - \leadi}
	\notag\\
	&\quad
		+ \learn_{\play} (1-\coef_{\play})
			\braket{\payfield_{\play}(\curr) - \payfield_{\play}(\beforelead)}{\nexti - \leadi}
\end{align}
and hence, by \eqref{eq:twostep2}:
\begin{align}
\label{eq:Fenchi-bound}
\fench_{\play}(\base_{\play},\nexti[\statealt])
	\leq \fench_{\play}(\base_{\play},\curri[\statealt])
		&+ \learn_{\play} \braket{\payfield_{\play}(\lead)}{\leadi - \base_{\play}}
		\notag\\
		&+ \learn_{\play} \braket{\payfield_{\play}(\lead) - \payfield_{\play}(\curr)}{\nexti - \leadi}
		\notag\\
		&+ \learn_{\play} (1-\coef_{\play})
			\braket{\payfield_{\play}(\curr) - \payfield_{\play}(\beforelead)}{\nexti - \leadi}
		\notag\\
		&- \fench_{\play}(\nexti,\leadi[\statealt])
		\notag\\
		&- \fench_{\play}(\leadi,\curri[\statealt])
		\eqdot
\end{align}
Accordingly, with $\curr[\energy] = \energy(\base,\curr[\statealt])$, the bound \eqref{eq:template} follows by multiplying both sides by $\mass_{\play} / \learn_{\play}$ and summing over $\play\in\players$.
\end{proof}

Thanks to \cref{prop:template}, we are now in a position to state and prove the following summability guarantee for \eqref{eq:FTRL+}.

\begin{proposition}
\label{prop:summable}
Suppose that each player in a harmonic game $\fingame$ with harmonic measure $\meas$ is following \eqref{eq:FTRL+} with learning rate $\learn_{\play} \leq \mass_{\play} \hstr_{\play} \bracks{2(\nPlayers+2) \max_{\playalt} \mass_{\playalt} \lips_{\playalt}}^{-1}$.
Then, for all $\nRuns$, we have:
\begin{equation}
\label{eq:summable}
\sum_{\run=\start}^{\nRuns} \norm{\lead - \curr}^{2}
	+ \sum_{\run=\afterstart}^{\nRuns} \norm{\curr - \beforelead}^{2}
	\leq \frac{2\init[\energy]}{(\nPlayers+2) \max_{\play} \mass_{\play} \lips_{\play}}
	\eqdot
\end{equation}
In particular, the sequences $\curr[A] \defeq \norm{\lead-\curr}^{2}$ and $\curr[B] \defeq \norm{\next-\lead}^{2}$ are both summable.
\end{proposition}

\begin{proof}
By reshuffling the terms of the template inequality \eqref{eq:template}, we get
\begin{subequations}
\label{eq:template-terms}
\begin{align}
\sum_{\play\in\players} \mass_{\play}
	&\braket{\payv_{\play}(\lead)}{\base_{\play} - \leadi}
	\notag\\
	&\leq \curr[\energy] - \next[\energy]
	\notag\\
	&+ \sum_{\play\in\players} \mass_{\play}
		\braket{\payv_{\play}(\lead) - \payfield_{\play}(\curr)}{\nexti - \leadi}
	\label{eq:term2}
	\\
	&+ \sum_{\play\in\players} \mass_{\play} (1 - \coef_{\play})
		\braket{\payfield_{\play}(\curr) - \payfield_{\play}(\beforelead)}{\nexti - \leadi}
	\label{eq:term3}
	\\
	&- \sum_{\play\in\players} \frac{\mass_{\play}}{\learn_{\play}} \fench_{\play}(\nexti,\leadi[\statealt])
		- \sum_{\play\in\players} \frac{\mass_{\play}}{\learn_{\play}} \fench_{\play}(\leadi,\curri[\statealt])
	\eqdot
	\label{eq:term4}
\end{align}
\end{subequations}
We now proceed to bound each term of \eqref{eq:template-terms} individually, paying no heed to make the resulting bounds as tight as possible.

\para{Bounding \eqref{eq:term2}}
By the Fenchel-Young inequality, we have:
\begin{align}
\eqref{eq:term2}
	&\leq \sum_{\play\in\players}
		\frac{\mass_{\play}}{2\lips_{\play}} \dnorm{\payfield_{\play}(\lead) - \payfield_{\play}(\curr)}^{2}
	+ \sum_{\play\in\players}
		\frac{\mass_{\play}\lips_{\play}}{2} \norm{\nexti - \leadi}^{2}
	\notag\\
	&\leq \sum_{\play\in\players}
		\frac{\mass_{\play}\lips_{\play}}{2} \norm{\lead - \curr}^{2}
	+ \sum_{\play\in\players}
		\frac{\mass_{\play}\lips_{\play}}{2} \norm{\nexti - \leadi}^{2}
	\explain{$\payfield_{\play}(\strat)$ is $\lips_{\play}$-Lipschitz}
	\\
	&\leq \tfrac{1}{2} \nPlayers \max\nolimits_{\play} \mass_{\play} \lips_{\play}
		\cdot \norm{\lead - \curr}^{2}
	+ \tfrac{1}{2} \max\nolimits_{\play} \mass_{\play} \lips_{\play}
		\cdot \norm{\next - \lead}^{2}
	\label{eq:term2+}
\end{align}

\para{Bounding \eqref{eq:term3}}
Again, by the Fenchel-Young inequality, we obtain:
\begin{align}
\eqref{eq:term3}
	&\leq \sum_{\play\in\players}
		\frac{\mass_{\play}(1-\coef_{\play})}{2\lips_{\play}} \dnorm{\payfield_{\play}(\curr) - \payfield_{\play}(\beforelead)}^{2}
	+ \sum_{\play\in\players}
		\frac{\mass_{\play}(1-\coef_{\play})\lips_{\play}}{2} \norm{\nexti - \leadi}^{2}
	\notag\\
	&\leq \sum_{\play\in\players}
		\frac{\mass_{\play}(1-\coef_{\play})\lips_{\play}}{2} \norm{\curr - \beforelead}^{2}
	+ \sum_{\play\in\players}
		\frac{\mass_{\play}(1-\coef_{\play})\lips_{\play}}{2} \norm{\nexti - \leadi}^{2}
	\explain{$\payfield_{\play}(\strat)$ is $\lips_{\play}$-Lipschitz}
	\notag\\
	&\leq \tfrac{1}{2} \nPlayers \max\nolimits_{\play} \mass_{\play} \lips_{\play}
		\cdot \norm{\curr - \beforelead}^{2}
	+ \tfrac{1}{2} \max\nolimits_{\play} \mass_{\play} \lips_{\play}
		\cdot \norm{\next - \lead}^{2}
	\label{eq:term3+}
\end{align}

\para{Bounding \eqref{eq:term4}}
Finally, by the lower bound on the Fenchel coupling of \cref{prop:Fench}, we get:
\begin{align}
&- \sum_{\play\in\players} \frac{\mass_{\play}}{\learn_{\play}} \fench_{\play}(\nexti,\leadi[\statealt])
	- \sum_{\play\in\players} \frac{\mass_{\play}}{\learn_{\play}} \fench_{\play}(\leadi,\curri[\statealt])
	\notag\\
	&\qquad
		\leq -\sum_{\play\in\players} \frac{\mass_{\play}\hstr_{\play}}{2\learn_{\play}} \norm{\nexti - \leadi}^{2}
		- \sum_{\play\in\players} \frac{\mass_{\play}\hstr_{\play}}{2\learn_{\play}} \norm{\leadi - \curri}^{2}
	\explain{by \eqref{eq:Fench-norm}}
	\\
	&\qquad
	\leq -\min\nolimits_{\play} \frac{\mass_{\play}\hstr_{\play}}{2\learn_{\play}}
		\cdot \bracks{ \norm{\next - \lead}^{2} + \norm{\lead - \curr}^{2}}
	\label{eq:term4+}
\end{align}
Thus, by folding \cref{eq:term2+,eq:term3+,eq:term4+} back into \eqref{eq:template-terms}, we obtain the bound
\begin{align}
\sum_{\play\in\players} \mass_{\play}
	&\braket{\payv_{\play}(\lead)}{\base_{\play} - \leadi}
	\notag\\
	&\leq \curr[\energy] - \next[\energy]
	\notag\\
	&+ \frac{1}{2} \parens*{
		\nPlayers \max\nolimits_{\play} \mass_{\play} \lips_{\play}
		- \min\nolimits_{\play} \frac{\mass_{\play}\hstr_{\play}}{\learn_{\play}}}
		\, \norm{\lead - \curr}^{2}
	\notag\\
	&+ \frac{1}{2} \parens*{
		2 \max\nolimits_{\play} \mass_{\play} \lips_{\play}
		- \min\nolimits_{\play} \frac{\mass_{\play}\hstr_{\play}}{\learn_{\play}}}
		\, \norm{\next - \lead}^{2}
	\notag\\
	&+ \frac{1}{2} \nPlayers \max\nolimits_{\play} \mass_{\play} \lips_{\play}
		\cdot \norm{\curr - \beforelead}^{2}
	\eqdot
	\label{eq:terms+}
\end{align}
Now, if we instantiate \eqref{eq:terms+} to $\base \gets \heq$ where $\heq$ is the strategic center of $\fingame$, its \ac{LHS} will vanish by \eqref{eq:harmonic-center}.
Hence, summing over all $\run=\running,\nRuns$, \eqref{eq:terms+} ultimately yields
\begin{align}
\label{eq:presum1}
0
	\leq \init[\energy]
	&+ \frac{1}{2} \parens*{
		\nPlayers \max\nolimits_{\play} \mass_{\play} \lips_{\play}
		- \min\nolimits_{\play} \frac{\mass_{\play}\hstr_{\play}}{\learn_{\play}}}
		\sum_{\run=\start}^{\nRuns} \norm{\lead - \curr}^{2}
	\notag\\
	&+ \frac{1}{2} \parens*{
		(\nPlayers+2) \max\nolimits_{\play} \mass_{\play} \lips_{\play}
		- \min\nolimits_{\play} \frac{\mass_{\play}\hstr_{\play}}{\learn_{\play}}}
		\, \sum_{\run=\afterstart}^{\nRuns} \norm{\curr - \beforelead}^{2}
	\notag\\
	&+ \frac{1}{2} \parens*{
		2 \max\nolimits_{\play} \mass_{\play} \lips_{\play}
		- \min\nolimits_{\play} \frac{\mass_{\play}\hstr_{\play}}{\learn_{\play}}}
		\, \norm{\afterlast - \state_{\nRuns+1/2}}^{2}
	\notag\\
	&+ \frac{1}{2} \nPlayers \max\nolimits_{\play} \mass_{\play} \lips_{\play}
		\cdot \norm{\init - \state_{1/2}}^{2}
	\eqdot
\end{align}
Now, by our step-size assumption, we readily obtain
\begin{equation}
(\nPlayers + 2) \max\nolimits_{\play} \mass_{\play} \lips_{\play}
	\leq \frac{1}{2} \min\nolimits_{\play} \frac{\mass_{\play}\hstr_{\play}}{\learn_{\play}} 
\end{equation}
so \eqref{eq:presum1} becomes
\begin{equation}
\label{eq:presum2}
0
	\leq \init[\energy]
	- \frac{1}{4} \min\nolimits_{\play} \frac{\mass_{\play}\hstr_{\play}}{\learn_{\play}}
		\sum_{\run=\start}^{\nRuns} \norm{\lead - \curr}^{2}
	- \frac{1}{4} \min\nolimits_{\play} \frac{\mass_{\play}\hstr_{\play}}{\learn_{\play}}
		\sum_{\run=\afterstart}^{\nRuns} \norm{\curr - \beforelead}^{2}
\end{equation}
where we used our initialization convention $\init = \state_{1/2}$ and the fact that the third line of \eqref{eq:presum1} is negative.
We thus get
\begin{equation}
\sum_{\run=\start}^{\nRuns} \norm{\lead - \curr}^{2}
	+ \sum_{\run=\afterstart}^{\nRuns} \norm{\curr - \beforelead}^{2}
	\leq \frac{4\init[\energy]}{\min_{\play} \mass_{\play}\hstr_{\play}/\learn_{\play}}
\end{equation}
from which our assertion follows immediately.
\end{proof}

\subsection{Proof of \cref{thm:regret}}

We are now in a position to prove the regret guarantees of \eqref{eq:FTRL+}, which we restate below for convenience.

\regret*

\begin{proof}
By a minor reshuffling of terms in \eqref{eq:Fenchi-bound}, we readily get
\begin{align}
\braket{\payfield_{\play}(\lead)}{\bench_{\play} - \leadi}
	&\leq \frac{1}{\learn_{\play}}
		\bracks{\fench_{\play}(\bench_{\play},\curri[\statealt]) - \fench_{\play}(\bench_{\play},\nexti[\statealt])}
	\notag\\
	&+ \vphantom{\frac{1}{\learn_{\play}}} \braket{\payfield_{\play}(\lead) - \payfield_{\play}(\curr)}{\nexti - \leadi}
	\notag\\
	&+ (1-\coef_{\play})
		\braket{\payfield_{\play}(\curr) - \payfield_{\play}(\beforelead)}{\nexti - \leadi}
	\notag\\
	&- \frac{1}{\learn_{\play}} \fench_{\play}(\nexti,\leadi[\statealt])
		- \frac{1}{\learn_{\play}} \fench_{\play}(\leadi,\curri[\statealt])
\end{align}
and hence, by a repeated application of the Fenchel-Young inequality in its Peter-Paul form:
\begin{align}
\braket{\payfield_{\play}(\lead)}{\bench_{\play} - \leadi}
	&\leq \frac{1}{\learn_{\play}}
		\bracks{\fench_{\play}(\bench_{\play},\curri[\statealt]) - \fench_{\play}(\bench_{\play},\nexti[\statealt])}
	\notag\\
	&+ \frac{1}{2\lips_{\play}} \dnorm{\payfield_{\play}(\lead) - \payfield_{\play}(\curr)}^{2}
		+ \frac{\lips_{\play}}{2} \norm{\nexti - \leadi}^{2}
	\notag\\
	&+ \frac{1-\coef_{\play}}{2\lips_{\play}} \dnorm{\payfield_{\play}(\curr) - \payfield_{\play}(\beforelead)}^{2}
		+ \frac{(1-\coef_{\play})\lips_{\play}}{2} \norm{\nexti - \leadi}^{2}
	\notag\\
	&- \frac{\hstr_{\play}}{2\learn_{\play}}
		\bracks*{\norm{\nexti - \leadi}^{2} + \norm{\leadi - \curri}^{2}}
	\eqdot
\end{align}
Hence, by using the Lipschitz continuity of $\payfield_{\play}$, we finally get
\begin{align}
	\braket{\payfield_{\play}(\lead)}{\bench_{\play} - \leadi}
	&\leq \frac{1}{\learn_{\play}}
		\bracks{\fench_{\play}(\bench_{\play},\curri[\statealt]) - \fench_{\play}(\bench_{\play},\nexti[\statealt])}
	\notag\\
	&+ \frac{\lips_{\play}}{2} \norm{\lead - \curr}^{2}
		+ \frac{\lips_{\play}}{2} \norm{\nexti - \leadi}^{2}
	\notag\\
	&+ \frac{\lips_{\play}}{2} \norm{\curr - \beforelead}^{2}
		+ \frac{\lips_{\play}}{2} \norm{\nexti - \leadi}^{2}
	\notag\\
	&- \frac{\hstr_{\play}}{2\learn_{\play}}
		\bracks*{\norm{\nexti - \leadi}^{2} + \norm{\leadi - \curri}^{2}}
\end{align}
Thus, summing over $\run=\running,\nRuns$, and keeping in mind that our assumptions for $\learn_{\play}$ also give $\lips_{\play} < \hstr_{\play}/(2\learn_{\play})$, we finally get
\begin{equation}
\label{eq:reg-bound1}
\sum_{\run=\start}^{\nRuns} \braket{\payfield_{\play}(\lead)}{\bench_{\play} - \leadi}
	\leq \frac{\hrange_{\play}}{\learn_{\play}}
		+ \frac{\lips_{\play}}{2} \bracks*{\sum_{\run=\start}^{\nRuns} \norm{\lead - \curr}^{2}
	+ \sum_{\run=\afterstart}^{\nRuns} \norm{\curr - \beforelead}^{2}}
\end{equation}
where we used the fact that $\fench_{\play}(\bench_{\play},0) = \hreg(\bench) - \min\hreg_{\play} \leq \max\hreg_{\play} - \min\hreg_{\play} \eqdef \hrange_{\play}$.
Our result then follows by invoking \eqref{eq:summable} and using the fact that $\mass_{\play}\lips_{\play} \leq \max_{\playalt} \mass_{\playalt} \lips_{\playalt}$ for all $\play\in\players$.
\end{proof}

\subsection{Proof of \cref{thm:convergence}}

With all this in hand, we are finally in a position to prove our main equilibrium convergence result for \eqref{eq:FTRL+}.
For convenience, we restate the relevant theorem below.

\convergence*

\begin{proof}
Our proof proceeds in a series of steps, as detailed below.

\begin{proofstep}{Convergence of energy levels}
We begin by showing that the energy $\curr[\energy] \equiv \energy(\heq,\curr[\statealt])$ of \eqref{eq:FTRL+} relative to the game's harmonic center converges to some finite value $\energy_{\infty}$.
That this is so follows from a well-known property of quasi-Fejér sequences \citep[Lemma~3.1]{Com01}, whose proof we reproduce below for completeness. 

Indeed, by \cref{eq:terms+,prop:summable}, we have
\begin{equation}
\label{eq:energy-Fejer}
\next[\energy]
	\leq \curr[\energy] + \curr[\eps]
\end{equation}
with $\curr[\eps]$, $\run=\running$ summable.
Letting $\curr[\alt\energy] = \curr[\energy] + \sum_{\runalt=\run}^{\infty} \iter[\eps]$, we further get
\begin{equation}
\next[\alt\energy]
	= \next[\energy] + \sum_{\runalt=\run+1}^{\infty} \iter[\eps]
	\leq \curr[\energy] + \sum_{\runalt=\run}^{\infty} \iter[\eps]
	= \curr[\alt\energy]
\end{equation}
by \eqref{eq:energy-Fejer}, so $\curr[\alt\energy]$ converges.
Since $\curr[\eps]$ is summable, it follows that $\curr[\energy]$ also converges, as claimed.
\endenv
\end{proofstep}

\begin{proofstep}{Boundedness of score differences}
We now proceed to show that the normalized score differences $\curr[\scorediff] = \PI(\curr[\statealt])$ where $\PI$ is the normalization operator \eqref{eq:Z-quotient} are bounded.
Indeed, by the definition of $\curr[\energy] = \energy(\heq,\curr[\statealt]) = \sum_{\play\in\players} (\mass_{\play}/\learn_{\play}) \fench_{\play}(\heq_{\play},\curri[\statealt])$, it follows that $\sup_{\run} \fench_{\play}(\heq_{\play},\curri[\statealt]) < \infty$ for all $\play\in\players$.
Thus, by \cref{lem:bounded-score-differences}, we conclude that each component of $\curr[\statealtalt]$ is bounded, so $\statealtalt$ is itself bounded.
\endenv
\end{proofstep}

\begin{proofstep}{Convergent subsequences of \eqref{eq:FTRL+}}
\label{step:subseq}
We now observe that \eqref{eq:FTRL+} enjoys the following series of properties:
\begin{enumerate}
\item
The sequence $\curr[\statealtalt] = \PI(\curr[\statealt])$ admits a subsequence $\statealtalt_{\iter[\run]}$ that converges to some limit point $\scorediff_{\infty} \in \scorediffs$ (a consequence of the fact that $\curr[\statealtalt]$ is bounded, see above).
\item
In turn, this implies that the subsequence $\state_{\iter[\run]} = \mirror(\statealt_{\iter[\run]}) = \effmirror(\statealtalt_{\iter[\run]})$ converges to some $\strat_{\infty} \in \strats$.
\item
Since the sequences $\curr[A] = \norm{\lead-\curr}^{2}$ and $\curr[B] = \norm{\next-\lead}^{2}$ are both summable (by \cref{prop:summable}), we further have $\lim_{\runalt\to\infty} \state_{\iter[\run] + 1/2} = \strat_{\infty}$ and, more generally, by a straightforward induction:
\begin{equation}
\label{eq:limitpoints}
\lim_{\runalt\to\infty} \state_{\iter[\run]+\offset}
	= \strat_{\infty}
	\quad
	\text{for any (fixed) $\offset = \halfrunning$}
\end{equation}
\item
Likewise, for the sequence of payoff signals $\curr[\signal]$, we have
\begin{equation}
\signal_{\play,\iter[\run]}
	= \coef_{\play} \,  \payfield_{\play}(\state_{\iter[\run]})
		+ (1 - \coef_{\play}) \, \payfield_{\play}(\state_{\iter[\run]-1/2})
	\xrightarrow[\runalt\to\infty]{} \coef_{\play} \payfield_{\play}(\strat_{\infty})
		+ (1 - \coef_{\play}) \payfield_{\play}(\strat_{\infty})
	= \payfield_{\play}(\strat_{\infty})
\end{equation}
so $\lim_{\runalt\to\infty} \payfield(\state_{\iter[\run]}) = \payfield(\strat_{\infty})$.
\endenv
\end{enumerate}
\end{proofstep}

\begin{proofstep}{Variational characterization of limit points}
We now proceed to show that $\payfield_{\infty} \defeq \payfield(\strat_{\infty})$ belongs to the polar cone
\(
\pcone(\strat_{\infty})
	= \setdef{\dvec\in\dpoints}{\braket{\dvec}{\strat - \strat_{\infty}} \leq 0 \; \text{ for all $\strat\in\strats$}}
\)
to $\strats$ at $\strat_{\infty}$.
To do so, suppose that \eqref{eq:FTRL+} performs $\offset$ steps from $\iter[\run]$ so
\begin{equation}
\statealt_{\iter[\run] + \offset}
	= \statealt_{\iter[\run]} + \learn \sum_{j=\start}^{\offset} \signal_{\iter[\run] + 1/2}
\end{equation}
where, to ease notation, we have made the simplifying assumption that $\learn_{\play} = \learn$ for all $\play\in\players$.%
\footnote{This assumption does not affect the core of our arguments, but it greatly streamlines the presentation.}
Then, by invoking \cref{lem:subsel-var-new} with $\dpoint \gets \statealt_{\iter[\run]}$ and $\new[\dpoint] \gets \statealt_{\iter[\run] + \offset}$, we obtain
\begin{align}
\braket*{\learn \sum_{j=\start}^{\offset} \signal_{\iter[\run] + 1/2}}{\base - \state_{\iter[\run] + \offset}}
	&\leq \braket{\subsel\hreg(\state_{\iter[\run] + \offset}) - \statealt_{\iter[\run]}}{\base - \state_{\iter[\run] + \offset}}
	\notag\\
	&= \braket{\subsel\hreg(\state_{\iter[\run] + \offset}) - \statealtalt_{\iter[\run]}}{\base - \state_{\iter[\run] + \offset}}
\end{align}
where, in the second line, we have used the fact that $\braket{\dpoint}{\stratalt - \strat} = \braket{\PI(\dpoint)}{\stratalt - \strat}$ for all $\strat,\stratalt\in\strats$ and all $\dpoint\in\dpoints$.
Thus, letting $\runalt\to\infty$, we get from \cref{step:subseq} and the continuity of $\subsel\hreg$ that
\begin{equation}
\learn\offset\braket{\payfield(\strat_{\infty})}{\strat - \strat_{\infty}}
	\leq \braket{\subsel\hreg(\strat_{\infty}) - \scorediff_{\infty}}{\strat - \strat_{\infty}}
\end{equation}
for all $\offset = \running$ and all $\strat\in\strats$.%
\footnote{The fact that $\strat_{\infty}\in\dom\subd\hreg$ is a consequence of \cref{lem:bounded-score-differences} and the convergence of $\curr[\energy]$ to $\energy_{\infty} < \infty$.}
Since $\offset$ can be chosen arbitrarily, we must have $\braket{\payfield(\strat_{\infty})}{\strat - \strat_{\infty}} \leq 0$ for all $\strat\in\strats$.
Hence, by the variational characterization \eqref{eq:VI} of \aclp{NE}, we conclude that $\strat_{\infty}$ must be itself a \acl{NE} of $\fingame$, and our proof is complete.
\end{proofstep}
\end{proof}

\bibliographystyle{icml}
\bibliography{bibtex/IEEEabrv,bibtex/Bibliography-PM,bibtex/Bibliography-DL}

\begin{thebibliography}{62}
\providecommand{\natexlab}[1]{#1}
\providecommand{\url}[1]{\texttt{#1}}
\expandafter\ifx\csname urlstyle\endcsname\relax
  \providecommand{\doi}[1]{doi: #1}\else
  \providecommand{\doi}{doi: \begingroup \urlstyle{rm}\Url}\fi

\bibitem[Abdou et~al.(2022)Abdou, Pnevmatikos, Scarsini, and Venel]{APSV22}
Abdou, J., Pnevmatikos, N., Scarsini, M., and Venel, X.
\newblock Decomposition of games: {Some} strategic considerations.
\newblock \emph{Mathematics of Operations Research}, 47\penalty0 (1):\penalty0
  176--208, February 2022.

\bibitem[Arnold(1989)]{arnoldMathematicalMethodsClassical1989}
Arnold, V.~I.
\newblock \emph{Mathematical {{Methods}} of {{Classical Mechanics}}}.
\newblock Springer-Verlag, 1989.

\bibitem[Azizian et~al.(2024)Azizian, Iutzeler, Malick, and
  Mertikopoulos]{AIMM24}
Azizian, W., Iutzeler, F., Malick, J., and Mertikopoulos, P.
\newblock The rate of convergence of {Bregman} proximal methods: {Local}
  geometry vs. regularity vs. sharpness.
\newblock \emph{SIAM Journal on Optimization}, 34\penalty0 (3):\penalty0
  2440--2471, September 2024.

\bibitem[Bertsekas(2015)]{Ber15}
Bertsekas, D.~P.
\newblock \emph{Convex optimization algorithms}.
\newblock Athena Scientific, 2015.

\bibitem[Cai et~al.(2022)Cai, Oikonomou, and Zheng]{Cai2022TightLC}
Cai, Y., Oikonomou, A., and Zheng, W.
\newblock Tight last-iterate convergence of the extragradient and the
  optimistic gradient descent-ascent algorithm for constrained monotone
  variational inequalities.
\newblock \url{https://arxiv.org/abs/2204.09228}, 2022.

\bibitem[Candogan et~al.(2011)Candogan, Menache, Ozdaglar, and Parrilo]{CMOP11}
Candogan, O., Menache, I., Ozdaglar, A., and Parrilo, P.~A.
\newblock Flows and decompositions of games: {Harmonic} and potential games.
\newblock \emph{Mathematics of Operations Research}, 36\penalty0 (3):\penalty0
  474--503, 2011.

\bibitem[Chen \& Teboulle(1993)Chen and Teboulle]{CT93}
Chen, G. and Teboulle, M.
\newblock Convergence analysis of a proximal-like minimization algorithm using
  {Bregman} functions.
\newblock \emph{SIAM Journal on Optimization}, 3\penalty0 (3):\penalty0
  538--543, August 1993.

\bibitem[Chen \& Peng(2020)Chen and Peng]{NEURIPS2020_db346ccb}
Chen, X. and Peng, B.
\newblock Hedging in games: {Faster} convergence of external and swap regrets.
\newblock In \emph{NeurIPS '20: Proceedings of the 34th International
  Conference on Neural Information Processing Systems}, 2020.

\bibitem[Chiang et~al.(2012)Chiang, Yang, Lee, Mahdavi, Lu, Jin, and
  Zhu]{CYLM+12}
Chiang, C.-K., Yang, T., Lee, C.-J., Mahdavi, M., Lu, C.-J., Jin, R., and Zhu,
  S.
\newblock Online optimization with gradual variations.
\newblock In \emph{COLT '12: Proceedings of the 25th Annual Conference on
  Learning Theory}, 2012.

\bibitem[Combettes(2001)]{Com01}
Combettes, P.~L.
\newblock Quasi-{Fej{\'e}rian} analysis of some optimization algorithms.
\newblock In Butnariu, D., Censor, Y., and Reich, S. (eds.), \emph{Inherently
  Parallel Algorithms in Feasibility and Optimization and Their Applications},
  pp.\  115--152. Elsevier, New York, NY, USA, 2001.

\bibitem[Daskalakis \& Panageas(2019)Daskalakis and Panageas]{DP19}
Daskalakis, C. and Panageas, I.
\newblock Last-iterate convergence: {Zero}-sum games and constrained min-max
  optimization.
\newblock In \emph{ITCS '19: Proceedings of the 10th Conference on Innovations
  in Theoretical Computer Science}, 2019.

\bibitem[Daskalakis et~al.(2009)Daskalakis, Goldberg, and
  Papadimitriou]{DGP09-acm}
Daskalakis, C., Goldberg, P.~W., and Papadimitriou, C.~H.
\newblock The complexity of computing a {Nash} equilibrium.
\newblock \emph{Communications of the ACM}, 52\penalty0 (2):\penalty0 89--97,
  2009.

\bibitem[Daskalakis et~al.(2021)Daskalakis, Fishelson, and
  Golowich]{daskalakis2021nearoptimal}
Daskalakis, C., Fishelson, M., and Golowich, N.
\newblock Near-optimal no-regret learning in general games.
\newblock In \emph{NeurIPS '21: Proceedings of the 35th International
  Conference on Neural Information Processing Systems}, 2021.

\bibitem[Farina et~al.(2022)Farina, Anagnostides, Luo, Lee, Kroer, and
  Sandholm]{farina2022nearoptimal}
Farina, G., Anagnostides, I., Luo, H., Lee, C.-W., Kroer, C., and Sandholm, T.
\newblock Near-optimal no-regret learning dynamics for general convex games.
\newblock In \emph{NeurIPS '22: Proceedings of the 36th International
  Conference on Neural Information Processing Systems}, 2022.

\bibitem[Fasoulakis et~al.(2022)Fasoulakis, Markakis, Pantazis, and
  Varsos]{pmlr-v151-fasoulakis22a}
Fasoulakis, M., Markakis, E., Pantazis, Y., and Varsos, C.
\newblock Forward-looking best-response multiplicative weights update methods
  for bilinear zero-sum games.
\newblock In \emph{AISTATS '22: Proceedings of the 25th International
  Conference on Artificial Intelligence and Statistics}, 2022.

\bibitem[Flokas et~al.(2020)Flokas, Vlatakis-Gkaragkounis, Lianeas,
  Mertikopoulos, and Piliouras]{FVGL+20}
Flokas, L., Vlatakis-Gkaragkounis, E.~V., Lianeas, T., Mertikopoulos, P., and
  Piliouras, G.
\newblock No-regret learning and mixed {Nash} equilibria: {They} do not mix.
\newblock In \emph{NeurIPS '20: Proceedings of the 34th International
  Conference on Neural Information Processing Systems}, 2020.

\bibitem[Friedman(1991)]{Fri91}
Friedman, D.
\newblock Evolutionary games in economics.
\newblock \emph{Econometrica}, 59\penalty0 (3):\penalty0 637--666, 1991.

\bibitem[Gidel et~al.(2019)Gidel, Berard, Vignoud, Vincent, and
  Lacoste-Julien]{GBVV+19}
Gidel, G., Berard, H., Vignoud, G., Vincent, P., and Lacoste-Julien, S.
\newblock A variational inequality perspective on generative adversarial
  networks.
\newblock In \emph{ICLR '19: Proceedings of the 2019 International Conference
  on Learning Representations}, 2019.

\bibitem[Golowich et~al.(2020)Golowich, Pattathil, and Daskalakis]{GPD20}
Golowich, N., Pattathil, S., and Daskalakis, C.
\newblock Tight last-iterate convergence rates for no-regret learning in
  multi-player games.
\newblock In \emph{NeurIPS '20: Proceedings of the 34th International
  Conference on Neural Information Processing Systems}, 2020.

\bibitem[Gorbunov et~al.(2022)Gorbunov, Taylor, and
  Gidel]{gorbunov2022lastiterate}
Gorbunov, E., Taylor, A., and Gidel, G.
\newblock Last-iterate convergence of optimistic gradient method for monotone
  variational inequalities.
\newblock In \emph{NeurIPS '22: Proceedings of the 36th International
  Conference on Neural Information Processing Systems}, 2022.

\bibitem[Harper(2011)]{Har11}
Harper, M.
\newblock Escort evolutionary game theory.
\newblock \emph{Physica D: Nonlinear Phenomena}, 240\penalty0 (18):\penalty0
  1411--1415, September 2011.

\bibitem[Hart \& Mas-Colell(2000)Hart and Mas-Colell]{HMC00}
Hart, S. and Mas-Colell, A.
\newblock A simple adaptive procedure leading to correlated equilibrium.
\newblock \emph{Econometrica}, 68\penalty0 (5):\penalty0 1127--1150, September
  2000.

\bibitem[Hart \& Mas-Colell(2006)Hart and Mas-Colell]{HMC06}
Hart, S. and Mas-Colell, A.
\newblock Stochastic uncoupled dynamics and {Nash} equilibrium.
\newblock \emph{Games and Economic Behavior}, 57:\penalty0 286--303, 2006.

\bibitem[H{\'e}liou et~al.(2017)H{\'e}liou, Cohen, and Mertikopoulos]{HCM17}
H{\'e}liou, A., Cohen, J., and Mertikopoulos, P.
\newblock Learning with bandit feedback in potential games.
\newblock In \emph{NIPS '17: Proceedings of the 31st International Conference
  on Neural Information Processing Systems}, 2017.

\bibitem[Hofbauer \& Sigmund(1998)Hofbauer and Sigmund]{HS98}
Hofbauer, J. and Sigmund, K.
\newblock \emph{Evolutionary Games and Population Dynamics}.
\newblock Cambridge University Press, Cambridge, UK, 1998.

\bibitem[Hsieh et~al.(2019)Hsieh, Iutzeler, Malick, and Mertikopoulos]{HIMM19}
Hsieh, Y.-G., Iutzeler, F., Malick, J., and Mertikopoulos, P.
\newblock On the convergence of single-call stochastic extra-gradient methods.
\newblock In \emph{NeurIPS '19: Proceedings of the 33rd International
  Conference on Neural Information Processing Systems}, pp.\  6936--6946, 2019.

\bibitem[Hsieh et~al.(2021)Hsieh, Antonakopoulos, and Mertikopoulos]{HAM21}
Hsieh, Y.-G., Antonakopoulos, K., and Mertikopoulos, P.
\newblock Adaptive learning in continuous games: {Optimal} regret bounds and
  convergence to {Nash} equilibrium.
\newblock In \emph{COLT '21: Proceedings of the 34th Annual Conference on
  Learning Theory}, 2021.

\bibitem[Hsieh et~al.(2022{\natexlab{a}})Hsieh, Antonakopoulos, Cevher, and
  Mertikopoulos]{HACM22}
Hsieh, Y.-G., Antonakopoulos, K., Cevher, V., and Mertikopoulos, P.
\newblock No-regret learning in games with noisy feedback: {Faster} rates and
  adaptivity via learning rate separation.
\newblock In \emph{NeurIPS '22: Proceedings of the 36th International
  Conference on Neural Information Processing Systems}, 2022{\natexlab{a}}.

\bibitem[Hsieh et~al.(2022{\natexlab{b}})Hsieh, Iutzeler, Malick, and
  Mertikopoulos]{HIMM22}
Hsieh, Y.-G., Iutzeler, F., Malick, J., and Mertikopoulos, P.
\newblock Multi-agent online optimization with delays: {Asynchronicity},
  adaptivity, and optimism.
\newblock \emph{Journal of Machine Learning Research}, 23\penalty0
  (78):\penalty0 1--49, May 2022{\natexlab{b}}.

\bibitem[Jiang et~al.(2011)Jiang, Lim, Yao, and
  Ye]{jiangStatisticalRankingCombinatorial2011}
Jiang, X., Lim, L.-H., Yao, Y., and Ye, Y.
\newblock Statistical ranking and combinatorial {{Hodge}} theory.
\newblock \emph{Mathematical Programming}, 127\penalty0 (1):\penalty0 203--244,
  2011.

\bibitem[Juditsky et~al.(2011)Juditsky, Nemirovski, and Tauvel]{JNT11}
Juditsky, A., Nemirovski, A.~S., and Tauvel, C.
\newblock Solving variational inequalities with stochastic mirror-prox
  algorithm.
\newblock \emph{Stochastic Systems}, 1\penalty0 (1):\penalty0 17--58, 2011.

\bibitem[Korpelevich(1976)]{Kor76}
Korpelevich, G.~M.
\newblock The extragradient method for finding saddle points and other
  problems.
\newblock \emph{{\`E}konom. i Mat. Metody}, 12:\penalty0 747--756, 1976.

\bibitem[Kwon \& Mertikopoulos(2017)Kwon and Mertikopoulos]{KM17}
Kwon, J. and Mertikopoulos, P.
\newblock A continuous-time approach to online optimization.
\newblock \emph{Journal of Dynamics and Games}, 4\penalty0 (2):\penalty0
  125--148, April 2017.

\bibitem[Lee(2003)]{Lee12}
Lee, J.~M.
\newblock \emph{Introduction to Smooth Manifolds}.
\newblock Number 218 in Graduate Texts in Mathematics. Springer-Verlag, New
  York, NY, 2 edition, 2003.

\bibitem[Legacci et~al.(2024)Legacci, Mertikopoulos, and Pradelski]{LMP24}
Legacci, D., Mertikopoulos, P., and Pradelski, B. S.~R.
\newblock A geometric decomposition of finite games: {Convergence} vs.
  recurrence under exponential weights.
\newblock In \emph{ICML '24: Proceedings of the 41st International Conference
  on Machine Learning}, 2024.

\bibitem[Lin et~al.(2020)Lin, Zhou, Mertikopoulos, and Jordan]{LZMJ20}
Lin, T., Zhou, Z., Mertikopoulos, P., and Jordan, M.~I.
\newblock Finite-time last-iterate convergence for multi-agent learning in
  games.
\newblock In \emph{ICML '20: Proceedings of the 37th International Conference
  on Machine Learning}, 2020.

\bibitem[Mertikopoulos \& Moustakas(2010)Mertikopoulos and Moustakas]{MM10}
Mertikopoulos, P. and Moustakas, A.~L.
\newblock The emergence of rational behavior in the presence of stochastic
  perturbations.
\newblock \emph{The Annals of Applied Probability}, 20\penalty0 (4):\penalty0
  1359--1388, July 2010.

\bibitem[Mertikopoulos \& Sandholm(2016)Mertikopoulos and Sandholm]{MS16}
Mertikopoulos, P. and Sandholm, W.~H.
\newblock Learning in games via reinforcement and regularization.
\newblock \emph{Mathematics of Operations Research}, 41\penalty0 (4):\penalty0
  1297--1324, November 2016.

\bibitem[Mertikopoulos \& Sandholm(2018)Mertikopoulos and Sandholm]{MerSan18}
Mertikopoulos, P. and Sandholm, W.~H.
\newblock Riemannian game dynamics.
\newblock \emph{Journal of Economic Theory}, 177:\penalty0 315--364, September
  2018.

\bibitem[Mertikopoulos \& Zhou(2019)Mertikopoulos and Zhou]{MZ19}
Mertikopoulos, P. and Zhou, Z.
\newblock Learning in games with continuous action sets and unknown payoff
  functions.
\newblock \emph{Mathematical Programming}, 173\penalty0 (1-2):\penalty0
  465--507, January 2019.

\bibitem[Mertikopoulos et~al.(2018)Mertikopoulos, Papadimitriou, and
  Piliouras]{MPP18}
Mertikopoulos, P., Papadimitriou, C.~H., and Piliouras, G.
\newblock Cycles in adversarial regularized learning.
\newblock In \emph{SODA '18: Proceedings of the 29th annual ACM-SIAM Symposium
  on Discrete Algorithms}, 2018.

\bibitem[Mertikopoulos et~al.(2019)Mertikopoulos, Lecouat, Zenati, Foo,
  Chandrasekhar, and Piliouras]{MLZF+19}
Mertikopoulos, P., Lecouat, B., Zenati, H., Foo, C.-S., Chandrasekhar, V., and
  Piliouras, G.
\newblock Optimistic mirror descent in saddle-point problems: {Going} the extra
  (gradient) mile.
\newblock In \emph{ICLR '19: Proceedings of the 2019 International Conference
  on Learning Representations}, 2019.

\bibitem[Mertikopoulos et~al.(2024)Mertikopoulos, Hsieh, and Cevher]{MHC24}
Mertikopoulos, P., Hsieh, Y.-P., and Cevher, V.
\newblock A unified stochastic approximation framework for learning in games.
\newblock \emph{Mathematical Programming}, 203:\penalty0 559--609, January
  2024.

\bibitem[Monderer \& Shapley(1996)Monderer and Shapley]{MS96}
Monderer, D. and Shapley, L.~S.
\newblock Potential games.
\newblock \emph{Games and Economic Behavior}, 14\penalty0 (1):\penalty0 124 --
  143, 1996.

\bibitem[Nemirovski(2004)]{Nem04}
Nemirovski, A.~S.
\newblock Prox-method with rate of convergence ${O}(1/t)$ for variational
  inequalities with {Lipschitz} continuous monotone operators and smooth
  convex-concave saddle point problems.
\newblock \emph{SIAM Journal on Optimization}, 15\penalty0 (1):\penalty0
  229--251, 2004.

\bibitem[Nesterov(2007)]{Nes07}
Nesterov, Y.
\newblock Dual extrapolation and its applications to solving variational
  inequalities and related problems.
\newblock \emph{Mathematical Programming}, 109\penalty0 (2):\penalty0 319--344,
  2007.

\bibitem[Nisan et~al.(2007)Nisan, Roughgarden, Tardos, and Vazirani]{NRTV07}
Nisan, N., Roughgarden, T., Tardos, {\'E}., and Vazirani, V.~V. (eds.).
\newblock \emph{Algorithmic Game Theory}.
\newblock Cambridge University Press, 2007.

\bibitem[Piliouras \& Shamma(2014)Piliouras and Shamma]{PS14}
Piliouras, G. and Shamma, J.~S.
\newblock Optimization despite chaos: {Convex} relaxations to complex limit
  sets via {Poincar{\'e}} recurrence.
\newblock In \emph{SODA '14: Proceedings of the 25th annual ACM-SIAM Symposium
  on Discrete Algorithms}, 2014.

\bibitem[Popov(1980)]{Pop80}
Popov, L.~D.
\newblock A modification of the {Arrow}\textendash{Hurwicz} method for search
  of saddle points.
\newblock \emph{Mathematical Notes of the Academy of Sciences of the USSR},
  28\penalty0 (5):\penalty0 845--848, 1980.

\bibitem[Rakhlin \& Sridharan(2013)Rakhlin and Sridharan]{RS13-NIPS}
Rakhlin, A. and Sridharan, K.
\newblock Optimization, learning, and games with predictable sequences.
\newblock In \emph{NIPS '13: Proceedings of the 27th International Conference
  on Neural Information Processing Systems}, 2013.

\bibitem[Robinson(1998)]{robinsonDynamicalSystemsStability1998}
Robinson, C.
\newblock \emph{Dynamical {{Systems}}: {{Stability}}, {{Symbolic Dynamics}},
  and {{Chaos}}}.
\newblock CRC Press, November 1998.
\newblock ISBN 978-1-4822-2787-1.

\bibitem[Rockafellar(1970)]{Roc70}
Rockafellar, R.~T.
\newblock \emph{Convex Analysis}.
\newblock Princeton University Press, Princeton, NJ, 1970.

\bibitem[Rockafellar \& Wets(1998)Rockafellar and Wets]{RW98}
Rockafellar, R.~T. and Wets, R. J.~B.
\newblock \emph{Variational Analysis}, volume 317 of \emph{A Series of
  Comprehensive Studies in Mathematics}.
\newblock Springer-Verlag, Berlin, 1998.

\bibitem[Rustichini(1999)]{Rus99}
Rustichini, A.
\newblock Optimal properties of stimulus-response learning models.
\newblock \emph{Games and Economic Behavior}, 29\penalty0 (1-2):\penalty0
  244--273, 1999.

\bibitem[Sandholm(2001)]{San01}
Sandholm, W.~H.
\newblock Potential games with continuous player sets.
\newblock \emph{Journal of Economic Theory}, 97:\penalty0 81--108, 2001.

\bibitem[Sato et~al.(2002)Sato, Akiyama, and Farmer]{SAF02}
Sato, Y., Akiyama, E., and Farmer, J.~D.
\newblock Chaos in learning a simple two-person game.
\newblock \emph{Proceedings of the National Academy of Sciences of the USA},
  99\penalty0 (7):\penalty0 4748--4751, April 2002.

\bibitem[Shalev-Shwartz(2011)]{SS11}
Shalev-Shwartz, S.
\newblock Online learning and online convex optimization.
\newblock \emph{Foundations and Trends in Machine Learning}, 4\penalty0
  (2):\penalty0 107--194, 2011.

\bibitem[Shalev-Shwartz \& Singer(2006)Shalev-Shwartz and Singer]{SSS06}
Shalev-Shwartz, S. and Singer, Y.
\newblock Convex repeated games and {Fenchel} duality.
\newblock In \emph{NIPS' 06: Proceedings of the 19th Annual Conference on
  Neural Information Processing Systems}, pp.\  1265--1272. MIT Press, 2006.

\bibitem[Syrgkanis et~al.(2015)Syrgkanis, Agarwal, Luo, and Schapire]{SALS15}
Syrgkanis, V., Agarwal, A., Luo, H., and Schapire, R.~E.
\newblock Fast convergence of regularized learning in games.
\newblock In \emph{NIPS '15: Proceedings of the 29th International Conference
  on Neural Information Processing Systems}, pp.\  2989--2997, 2015.

\bibitem[Taylor \& Jonker(1978)Taylor and Jonker]{TJ78}
Taylor, P.~D. and Jonker, L.~B.
\newblock Evolutionary stable strategies and game dynamics.
\newblock \emph{Mathematical Biosciences}, 40\penalty0 (1-2):\penalty0
  145--156, 1978.

\bibitem[Viossat \& Zapechelnyuk(2013)Viossat and Zapechelnyuk]{VZ13}
Viossat, Y. and Zapechelnyuk, A.
\newblock No-regret dynamics and fictitious play.
\newblock \emph{Journal of Economic Theory}, 148\penalty0 (2):\penalty0
  825--842, March 2013.

\bibitem[Wei et~al.(2021)Wei, Lee, Zhang, and Luo]{WLZL21}
Wei, C.-Y., Lee, C.-W., Zhang, M., and Luo, H.
\newblock Linear last-iterate convergence in constrained saddle-point
  optimization.
\newblock In \emph{ICLR '21: Proceedings of the 2021 International Conference
  on Learning Representations}, 2021.

\end{thebibliography}

\end{document}